\documentclass[a4paper,11pt]{article}
\usepackage{jheppub}
\usepackage{lineno}
\usepackage{subcaption}
\usepackage{tcolorbox}
\usepackage{booktabs,tabularx,multirow,bm}
\usepackage[capitalize]{cleveref}
\usepackage{amsmath,amssymb}
\usepackage{amsthm}
\usepackage[percent]{overpic}
\usepackage{float}
\usepackage[utf8]{inputenc}
\tcbuselibrary{theorems}
\newtheorem{theorem}{Theorem}
\newtheorem{conjecture}{Conjecture}
\newtheorem{corollary}{Corollary}
\newtheorem{proposition}{Proposition}

\crefname{conjecture}{conjecture}{conjectures}
\Crefname{conjecture}{Conjecture}{Conjectures}
\crefname{corollary}{corollary}{corollaries}
\Crefname{corollary}{Corollary}{Corollaries}
\crefname{appendix}{Appendix}{Appendices}
\Crefname{appendix}{Appendix}{Appendices}

\newcommand{\R}{\mathbb R}

\newcommand{\dd}{\,\mathrm d}

\title{\boldmath Classical fractons with cosmological fixed points}
\author[a,b]{Akash Singh,}
\author[c]{Dileep P. Jatkar,}
\author[d]{S. L. Sondhi,}
\author[a,b]{and Abhishodh Prakash}
\affiliation[a]{Harish-Chandra Research Institute (HRI), Prayagraj (Allahabad) 211019, India}
\affiliation[b]{Homi Bhabha National Institute (HBNI),  Mumbai 400094, India}
\affiliation[c]{School of Mathematics and Natural Sciences, Chanakya University, Bengaluru 562165, India}
\affiliation[d]{Rudolf Peierls Centre for Theoretical Physics, University of Oxford, Oxford OX1 3PU, U.K.}
\emailAdd{akashsingh@hri.res.in}
\emailAdd{dileep.j@chanakyauniversity.edu.in}
\emailAdd{shivaji.sondhi@physics.ox.ac.uk}
\emailAdd{abhishodhprakash@hri.res.in}

\abstract{Classical fracton Hamiltonians are unusual examples of conservative many-body systems that can develop attractors after projection onto configuration or shape variables, although the full Hamiltonian phase space admits none. In this work, we study the scale-invariant, dipole-conserving two-parameter family of fracton Hamiltonians \(H_{\alpha,\beta}=\sum_{i<j}|\mathbf p_i-\mathbf p_j|^\alpha|\mathbf x_i-\mathbf x_j|^\beta\). By separating the coordinates into an overall scale and a residual shape, we obtain autonomous shape dynamics that admit fixed points which leave a purely scale evolution of the form \(R(t)\propto |t|^{\alpha/(\alpha-\beta)}\). These shape fixed points, which determine the spatial distribution of the expanding particles, are central configurations of power-law Riesz potentials with exponent \(\alpha+\beta\). The distinguished model \((\alpha,\beta)=(-2,1)\) is unique: its scale evolution takes the Einstein-de Sitter form \(R(t)\propto |t|^{2/3}\), its fixed-point equation is exactly the equal-mass Newtonian central-configuration equation, its large-\(N\) distribution is a homogeneous expanding ball, and its homothetic trajectories admit a zero-energy Newtonian gravitational dual. Combining analytic stability calculations with numerics, we show that the fixed points are locally stable; direct simulations at moderate \(N\) approach them from random initial data. At large \(N\), simulations instead reveal a richer class of fixed-points, bound clusters of approximately fixed physical size retain internal motion, while their coarse-grained centers approach unequal-mass Newtonian central configurations and preserve large-scale homogeneity. A scale-separation conjecture yields an effective unequal-mass fracton dynamics for the centers and a corresponding zero-energy Newtonian gravitational dual. We also find that trajectories generically exhibit a bidirectional arrow of time: scale and shape complexity grow away from a Janus point, while the leading contribution to Boltzmann entropy grows logarithmically. Together, the Einstein-de Sitter expansion of cluster centers, bound internal motion, large-scale homogeneity from central configurations, critical Newtonian duality and an emergent arrow of time reproduce the salient dynamical structure of a flat matter-dominated cosmology.  In the distinguished fracton model, all these cosmological analogues emerge as attractor properties, making it a toy model for cosmological dynamics without fine-tuning.}

\begin{document}
	\maketitle
	\flushbottom
	
	\section{Introduction}
	Attractor fixed points provide one of the clearest mechanisms by which complicated microscopic dynamics can produce simple and universal late-time behavior.  Trajectories in the same basin progressively lose memory of their initial conditions and approach a common invariant state, allowing a small amount of fixed-point data to govern a broad class of evolutions~\cite{Strogatz2015}.  The renormalization group famously applies the same idea in a different setting: coupling constants are treated as coordinates of a dynamical system, the logarithmic scale plays the role of time, and fixed points organize universal long-distance physics~\cite{WilsonKogut1974}.
	
	However, for ordinary Hamiltonian systems, this mechanism appears to be inaccessible.  A smooth Hamiltonian dynamics, including that generated by the familiar kinetic-plus-potential form \(H=\sum_i\mathbf p_i^2/(2m_i)+V(\{\mathbf x_i\})\), is symplectic and hence preserves phase-space volume in accordance with the Liouville theorem~\cite{Arnold1978}.  An asymptotically stable fixed point with an open basin would contract that volume and is therefore forbidden in the full Hamiltonian phase space~\cite{Arnold1978}.  The obstruction does not apply in the same way after projection or reduction: the contraction in the retained variables can be compensated by expansion along the discarded directions~\cite{bravetti2022scalingsymmetriescontactreduction,Sloan2018}.  A Hamiltonian system may consequently display attractor patterns in the configuration space, or in a space of relational variables, while its complete phase-space flow remains volume preserving~\cite{PrakashSadkiSondhi_Fractons_2024,bravetti2022scalingsymmetriescontactreduction}.
	
	Classical fractons offer striking examples of this possibility~\cite{PrakashGorielySondhi_Fractons_2024,PrakashSadkiSondhi_Fractons_2024}.  Their higher-moment conservation laws make mobility intrinsically collective, so that the motion of a subsystem depends on the presence and arrangement of the surrounding particles~\cite{PrakashGorielySondhi_Fractons_2024,PrakashSadkiSondhi_Fractons_2024}.  Although microscopic dynamics is Hamiltonian, its projection onto the configuration space can have attractor loci~\cite{PrakashSadkiSondhi_Fractons_2024}.  Earlier classical-fracton models exhibited attractor behavior in the configuration space, along with clustering and non-Gibbsian steady states~\cite{PrakashGorielySondhi_Fractons_2024,PrakashSadkiSondhi_Fractons_2024}. In addition, they exhibited the Machian dependence on the many-body environment~\cite{PrakashSadkiSondhi_Fractons_2024}, and local chaos, global ergodicity breaking, and an emergent arrow of time~\cite{BabbarSadkiPrakashSondhi_Fractons_2025}.  These phenomena suggest that attractors are a basic organizing principle of classical fracton dynamics. However, the earlier examples did not yield a general analytic classification of their fixed points: their nonlinear many-body dynamics was explored largely through numerical evolution~\cite{PrakashGorielySondhi_Fractons_2024,PrakashSadkiSondhi_Fractons_2024,BabbarSadkiPrakashSondhi_Fractons_2025}.
	
	In this work, we study a class for which an important set of attractor fixed points can be tracked analytically.  We consider the scale-invariant, dipole-conserving two-parameter family of Hamiltonians,
	\begin{equation*}
		H_{(\alpha,\beta)}
		=\sum_{i<j}
		|\mathbf p_i-\mathbf p_j|^\alpha
		|\mathbf x_i-\mathbf x_j|^\beta .
	\end{equation*}
	The mechanical similarity of this class of Hamiltonians permits a dynamical-similarity reduction~\cite{landau1976mechanics,bravetti2022scalingsymmetriescontactreduction,Sloan2018}.  After removing the center variables, the remaining degrees of freedom separate into an overall scale and dimensionless position-momentum shape variables evolving in a rescaled time.  The reduced shape equations are autonomous and may not preserve the symplectic volume of the original Hamiltonian system~\cite{bravetti2022scalingsymmetriescontactreduction,Sloan2018}.  They can therefore possess attractor fixed points without contradicting the Liouville theorem~\cite{Arnold1978,bravetti2022scalingsymmetriescontactreduction}.
	
	The radial fixed points of this reduced dynamics admit a simple geometric characterization.  Their spatial shapes are central configurations of a power-law Riesz potential with exponent \(q=\alpha+\beta\)~\cite{Battye_2003,DiacuPerezChavelaSantoprete2006,Maderna2013HomogeneousNBody}, while the remaining motion is a pure expansion or contraction, \(R(t)\propto |t|^{\alpha/(\alpha-\beta)}\).  Known results for Riesz equilibria and their particle approximation determine global regular large-\(N\) fixed-point profiles, including weighted balls and spherical shells~\cite{FrankMatzke2025,CarazzatoPratelliTopaloglu2026}.  On the \(\alpha=-2\beta\) branch, radial fixed points have the Einstein-de Sitter scale law \(R(t)\propto |t|^{2/3}\).  We prove a conditional local-stability theorem for the expanding regular extrema on this branch: the required positivity condition is established analytically for the equilateral three-body configuration, and the predicted stability is supported numerically for representative larger minima, while saddle-type configurations retain unstable directions.
	
	These results pick the model \((\alpha,\beta)=(-2,1)\), whose distinguishing feature is that it is the unique member of the Einstein-de Sitter branch for which the fracton fixed-point equation is identical to the equal-mass Newtonian gravitational central-configuration equation~\cite{Battye_2003}.  Its regular large-\(N\) profile is a homogeneous solid ball~\cite{FrankMatzke2025,CarazzatoPratelliTopaloglu2026} and each homothetic fixed-point trajectory admits an exact dual description as a trajectory in the Newtonian inverse-square gravitational force with zero Newtonian energy.  Both the critical energy and the effective gravitational coupling are selected by the fracton solution rather than imposed as independent Newtonian data.  This is an equivalence of the homothetic trajectories, not a canonical equivalence between the complete fracton and Newtonian theories.

	The numerical evolution indicates how this structure is selected beyond the linear neighborhood of a fixed point.  At moderate particle number (\(N\)), the random initial conditions studied here approach an expanding regular fixed point with the Einstein-de Sitter scale law. At larger \(N\), the simulations reveal a richer class of clustered attractors. The tight clusters retain internal dynamics while their coarse-grained centers approach unequal-mass Newtonian central configurations and separate as \(t^{2/3}\).  These motivate a conjectured effective unequal-mass fracton dynamics for these centers and, as a corollary, their homothetic center dynamics has a corresponding unequal-mass zero-energy Newtonian gravitational dual. Throughout the simulation, the clusters remain approximately fixed in physical size and coexist with large-scale homogeneity. This local dynamics has the Machian character because while an isolated pair cannot support such a bound motion,  pairs embedded in the many-body system can~\cite{PrakashSadkiSondhi_Fractons_2024,brans1961mach,barbour1995mach,sciama1953origin}.  Generic trajectories  exhibit a bidirectional scale-and-complexity arrow-of-time, {\em{i.e.}}, both grow away from a Janus point in the opposite time directions, while at large-\(R\), the leading contribution to the finite coarse-grained Boltzmann entropy grows logarithmically despite the time-reversal invariance of the microscopic Hamiltonian~\cite{BarbourKoslowskiMercati_PhysRevLett.113.181101,BabbarSadkiPrakashSondhi_Fractons_2025}.
	
	The resulting late-time dynamics closely parallels a flat matter-dominated cosmology~\cite{Ellis_2013,Ellis_2015}: the Einstein-de Sitter expansion, the critical Newtonian dynamics~\cite{Ellis_2013,Ellis_2015}, large-scale homogeneity~\cite{dodelson2020modern,Peebles1993}, bound local structures~\cite{einstein1945influence,carrera2010influence}, and an arrow of time~\cite{Penrose1979,albert2000time} are all emergent properties of the attractor point of the distinguished model.  In the corresponding purely Newtonian description of cosmology, several of these properties require special initial data~\cite{Ellis_2013,Ellis_2015,Battye_2003,Penrose1979,albert2000time}; here they arise through attraction within the basins explored analytically and numerically.  Thus, this Hamiltonian can be regarded as a controlled toy model showing how features that appear fine-tuned in a cosmological effective description can instead be emergent robust late-time consequences.

    The rest of the paper is organized as follows.  \Cref{sec:newtonian_fractons} reviews classical-fracton mechanics, develops two- and many-particle warm-ups, and then states the Newtonian-cosmology benchmark used later. \Cref{sec:scale_invt_fractons} develops the scale-shape reduction and radial fixed-point analysis, including the exact Newtonian dual in \cref{sec:full_newtonian_dual}.  \Cref{sec:distinguished_model} studies the distinguished model and its clustered large-\(N\) regime, while \cref{sec:structure} examines bound local dynamics, its Machian character and the emergent arrow of time.  \Cref{sec:conclusions} summarizes the results and discusses their implications for cosmological fine-tuning. The technical derivations and numerical details are collected in \cref{app:two_particle_3d,app:CentralConfigs,app:LargeNDistributions,app:StabilityDetails,app:NumericalImplementation,app:Entropy}.

\section{Classical Fracton Dynamics: Review and Warm-up}
	\label{sec:newtonian_fractons}
    
	\subsection{Classical Nonrelativistic Fractons: a Brief Review}
	Fracton phases were introduced as quantum many-body systems with excitations that are immobile or are confined to move on lower-dimensional submanifolds~\cite{Chamon2005,Haah2011,Vijay2015}.  Their restricted mobility, subextensive degenerate structure, and sensitivity to geometry are naturally organized by higher-moment conservation laws and admit descriptions in terms of tensor gauge theories, elasticity dualities, and multipole algebras~\cite{Pretko2017a,Pretko2017b,PretkoRadzihovsky2018,Gromov2019,Gromov2019curved,Gorantla2020,SeibergShao2021}.  Related constructions connect fractons to emergent gravity, holographic toy model, and non-Lorentzian geometric structures~\cite{Pretko2017c,Yan2019,SlagleKim2017,JainJensen2022,Bidussi2023}. 
    
    Classical non-relativistic fractons, whose study was initiated recently~\cite{PrakashGorielySondhi_Fractons_2024}, were found to exhibit several unusual many-body phenomena. To start with, the dynamics is Machian, {\em{i.e}}, isolated particles are immobile and motion is activated only by the presence of nearby particles. Secondly, ergodicity is strongly broken, and particles appear to cluster in an arbitrary dimension and with an arbitrary energy.  Finally, non-equilibrium steady states appear outside the usual Gibbsian description, and their dynamics spontaneously generates a global arrow of time~ \cite{PrakashGorielySondhi_Fractons_2024,PrakashSadkiSondhi_Fractons_2024,BabbarSadkiPrakashSondhi_Fractons_2025}.  Closely related works have extended the same logic to phase-space multipole conservation, continuum quantum mechanics, and freezing transitions in dipole-conserving chains~ \cite{SadkiPrakashSondhiArovas_PhaseSpaceFractons_2026,SadkiPrakashSondhi_ContinuumFractons_2025,ClassenHowesSenesePrakash_FreezingTransitions_2025}. The striking features of classical fracton dynamics can be attributed to the existence of attractor fixed points in their solution space. In spite of the fact that the fracton dynamics is governed by a Hamiltonian, which is constrained by the Liouville theorem, stable fixed points were shown to emerge when we track the dynamics in configuration space, spanned by positions and velocities rather than phase space spanned by positions and momenta, due to the unusual relationship between velocities and momenta imposed by symmetries. These configuration-space attractors were established analytically for small numbers of particles $N$ and inferred numerically for larger $N$. In this work, we introduce a new class of scale-invariant classical fracton Hamiltonians, whose fixed points are more tractable by using a so-called scale-shape reduction and whose steady states exhibit striking similarities to aspects of Newtonian cosmological dynamics. First, we begin with a review and warm-up using standard fracton Hamiltonians~\cite{PrakashGorielySondhi_Fractons_2024,BabbarSadkiPrakashSondhi_Fractons_2025,SadkiPrakashSondhiArovas_PhaseSpaceFractons_2026,SadkiPrakashSondhi_ContinuumFractons_2025}.
    
	We consider the Hamiltonian dynamics of \(N\) identical point particles.  In addition to total momentum conservation, we impose conservation of the total dipole moment.  For identical unit charges this is simply conservation of the `center-of-mass'~\footnote{We mean this in reference to an equivalent Newtonian system, where the particles have identical unit masses. Note that in our framework, `mass' is not well-defined.}
	\begin{equation}
		\mathbf{P}_{\text{C}} = \frac{1}{N} \sum_{i=1}^N \mathbf{p}_i,~\mathbf{X}_{\text{C}} = \frac{1}{N}\sum_{i=1}^N \mathbf{x}_i .\label{eq:conservation_posnmom}
	\end{equation}
	Throughout this work, bold symbols denote vectors, and the corresponding non-bold symbols their magnitudes.  The two conservation laws are equivalently implemented by invariance under translations in position and momentum space,
	\begin{equation}
		\mathbf{x}_i \mapsto \mathbf{x}_i + \mathbf{a},~\mathbf{p}_i \mapsto \mathbf{p}_i + \mathbf{b}.\label{eq:symmetry_actions}
	\end{equation}
	A minimal local Hamiltonian with these symmetries is
	\begin{equation}
		H = \frac{1}{2} \sum_{j<i} |\mathbf{p}_i - \mathbf{p}_j|^2 K(|\mathbf{x}_i - \mathbf{x}_j|). \label{eq:H_fracton}
	\end{equation}
	The dependence on momentum differences makes the collective character explicit: an isolated particle is immobile, while the surrounding particles supply its effective inertia~\cite{PrakashGorielySondhi_Fractons_2024,PrakashSadkiSondhi_Fractons_2024}. \(K(r)\) is the pair-inertia, or the mobility function, which is non-negative and decays at large \(r\), thus implementing the locality. For $K(x)$ with a strictly compact support, sufficiently separated particles cease to influence each other with inertia. The inverse-distance choice for $K(x)$ considered in \cref{sec:Warmup} is the simplest scale-invariant warm-up for the broader family studied in \cref{sec:scale_invt_fractons}.
	
	\subsection{Critical Newtonian Dynamics of a Pair of Fractons}
	\label{sec:Warmup}
    We first isolate the basic mechanism in an exactly solvable two-body setting.  Its dynamics already displays three many-body features encountered later: asymptotic `radialization',  scale dynamics and an auxiliary zero-energy Newtonian description.
	
	\subsubsection*{Exact One-dimensional Two-Fracton Dynamics}
	Consider the Hamiltonian given in \eqref{eq:H_fracton} for \(N=2\), that is \((i,j)=(2,1)\).  We begin with the dynamics in one spatial
	dimension, where every trajectory is necessarily radial and the
	Newtonian identification is exact.  The Hamiltonian is
	\begin{equation}
		H = \frac{1}{2} (p_1-p_2)^2 K(x_1 - x_2). \label{eq:H_2particles}
	\end{equation} 
	For initial conditions with conserved fracton energy \(E_f\neq0\), Hamilton's equations of motion are as follows.
	\begin{align}
		\dot{x}_1 &= - \dot{x}_2 = (p_1-p_2) K(x_1 - x_2)
		\label{eq:2_fracton_eom},\quad\\
		\dot{p}_1 &= - \dot{p}_2  = -\frac{(p_1 - p_2)^2}{2} K'(x_1 - x_2). 
	\end{align}
	Eliminating the momenta gives us
	\begin{equation}
		\ddot{x}_{a} =(p_1- p_2)^2 K(x_1 - x_2) \frac{\partial K}{\partial x_a}= 2 E_f \frac{\partial K}{\partial x_a}\ , \label{eq:Newtons_eom_2particles}
	\end{equation}
	where \(E_f\) is the conserved fracton energy. \Cref{eq:Newtons_eom_2particles} is precisely Newton's equation for unit-mass particles with effective interaction potential $U_{\text{eff}}=-2E_f K(x_1-x_2)$. This equivalence is not one of full phase spaces: only a special Newtonian energy surface is reached by fracton initial data.  Using the equations of motion in \cref{eq:2_fracton_eom}, one finds that the Newtonian energy vanishes,
	\begin{equation}
		E_N = \frac{\dot{x}_1^2 }{2} + \frac{\dot{x}_2^2 }{2} - 2 E_f K(x_1 - x_2) = 0.
	\end{equation}
	Thus, the Newtonian dynamics generated by two fractons is \emph{critical}, separating bound \((E_N<0)\) and unbounded trajectories \((E_N>0)\).  The mapping places every such fracton trajectory on the Newtonian surface \(E_N=0\), while the true conserved energy \(E_f\ge0\) remains unrestricted.
	
	Let us now specialize to the choice $K(x)=1/|x|$.  The effective potential is then $U_{\text{eff}}=-2E_f/|x_1-x_2|$, identical in form to the Newtonian gravitational interaction with an effective gravitational constant $G\propto E_f$.  The equations of motion \eqref{eq:2_fracton_eom} can be integrated exactly. Separation $r(t)\equiv |x_1(t)-x_2(t)|$ evolves as
	\begin{equation}
		r(t) = (18E_f)^{1/3}|t-t_0|^{2/3}. \label{eq:2particle_1d_r}
	\end{equation}
	where \(t_0\) is defined by \(r(t_0)=0\). Thus, the two-fracton dynamics with inverse-distance mobility function has homothetic Einstein-de Sitter evolution \(r(t)\propto |t-t_0|^{2/3}\). 
	
	\subsubsection*{Exact Three-dimensional Two-Fracton Dynamics}
	In three dimensions too the equations of motion can be integrated exactly. The canonical relative angular momentum
	\(\boldsymbol\ell=(\mathbf x_1-\mathbf x_2)\times
	(\mathbf p_1-\mathbf p_2)\) is conserved.  For \(K(r)=1/r\),
	its two most useful consequences are
	\begin{equation}
		r(t)^3=\frac{\ell^2}{2E_f}+18E_f(t-t_*)^2,
		\qquad
		\sin\chi=
		\frac{\ell}{\sqrt{\ell^2+36E_f^2(t-t_*)^2}},
		\label{eq:two_fracton_3d_highlights}
	\end{equation}
	where $r = |\mathbf{x}_1 - \mathbf{x}_2|$, \(\chi\) is the angle between the relative position and canonical
	momentum and $t_*$ is the time at minimum separation.  Thus every solution with finite \(\ell\) approaches a radial
	fixed point of the orientation dynamics and
	\begin{equation}
		r(t)\sim(18E_f)^{1/3}|t|^{2/3},
		\qquad |t-t_*|\longrightarrow\infty
	\end{equation}
	just as in \cref{eq:2particle_1d_r}. Nonzero angular momentum changes the finite-time orbit and prevents collision,
	but it does not change the asymptotic scale law.
	
	The Newtonian statement is more restrictive.  On the invariant radial
	submanifold \(\boldsymbol\ell=0\), introduce Newtonian canonical momenta
	\(\boldsymbol\pi_a=\dot{\mathbf x}_a\).  The radial fracton trajectory is
	then generated exactly by the pairwise Newtonian Hamiltonian
	\begin{equation}
		H^{\rm dual}
		=\frac12\sum_{a=1}^2|\boldsymbol\pi_a|^2
		-\frac{G_{\rm eff}^{(2)}}{|\mathbf x_1-\mathbf x_2|},
		\qquad
		G_{\rm eff}^{(2)}=2E_f,
		\qquad E_{N}=0 .
		\label{eq:two_fracton_3d_dual_hamiltonian}
	\end{equation}
	After removing the trivial center-of-mass motion, the required
	Newtonian dynamics is therefore both radial and fine-tuned to zero energy.
	The gravitational constant \(G_{\rm eff}^{(2)}\) is not a parameter of
	the fracton Hamiltonian: it is generated dynamically by the conserved
	fracton energy set by the initial conditions.  A fracton orbit with \(\ell\ne0\) is not generated by
	\cref{eq:two_fracton_3d_dual_hamiltonian} at finite time.  Its asymptotic curve can nevertheless be reproduced by the zero-energy
	radial Newtonian orbit.
	
	\subsection{Beyond Two Particles: Emergent Einstein-de Sitter Scale Dynamics}
	\label{sec:manybody_warmup}
    \begin{figure}[!htbp]
		\centering
		\includegraphics[width=\textwidth]{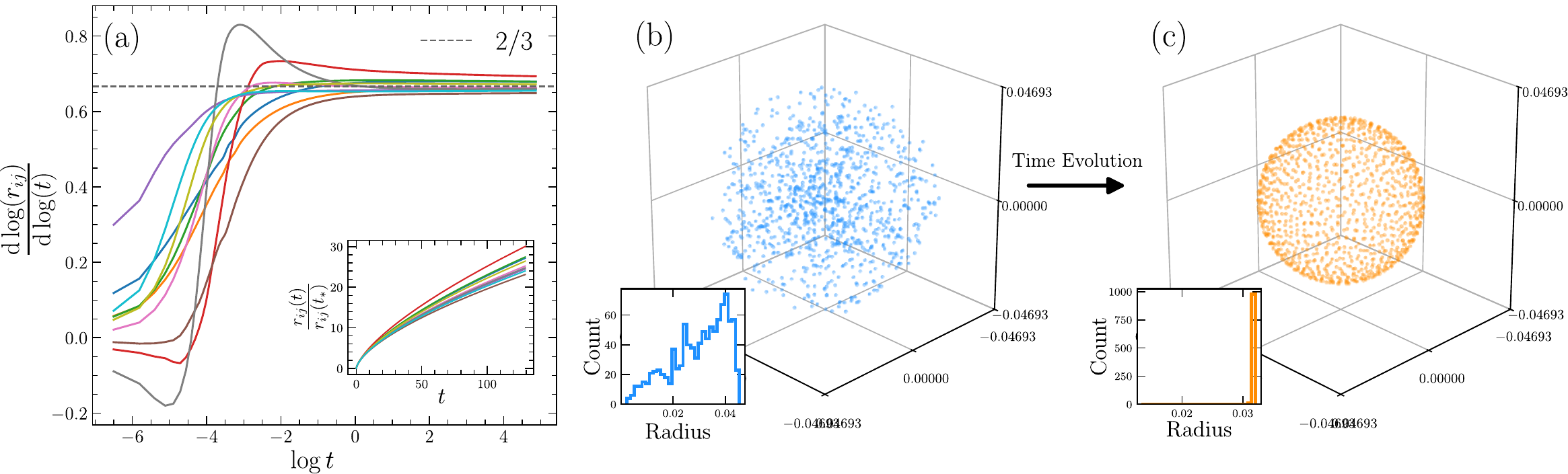}
			\caption{(a): \(N=256\) evolution of separations of ten randomly selected particle pairs for the pair-inertia model \(K(r)=1/r\) showing instantaneous logarithmic slope \(d\log r_{ij}/d\log t\), approaching the Einstein-de Sitter value \(2/3\), (inset) normalized separations of ten randomly selected particle pairs. (b,c): Time evolution from randomly generated initial conditions for \(N=1000\) (b).  The scale-normalized radial distribution settles into a thin spherical shell of radius \(1/\sqrt N\simeq0.032\) (c).}
		\label{fig:p2x-1}
	\end{figure}
    
	We next consider the \(N\)-body Hamiltonian in three dimensions with the same inverse-distance pair inertia,
	\begin{equation}
		H = \frac{1}{2}  \sum_{j<i} |\mathbf{p}_i - \mathbf{p}_j|^2 |\mathbf{x}_i-\mathbf{x}_j|^{-1} \label{eq:H_fracton_local}
	\end{equation}
	The two-particle calculation now suggests a sharper many-body question:
	which parts of the radialization, Einstein-de Sitter scaling, and
	critical Newtonian description survive when \(N>2\)? We begin this study numerically with randomly generated Gaussian-sampled phase-space data. As shown in \cref{fig:p2x-1}, the early motion is irregular, but the system then
	settles into an overall scale expansion close to the Einstein-de Sitter
	power law.  The figure follows the separations \(r_{ij}\) of several
	randomly chosen pairs.  After the initial transient evolution, these
	separations approach a common power law with logarithmic slope near
	\(2/3\).  Thus the asymptotic radial expansion seen in the two-body dynamics survives in the many-body problem as homothetic scale expansion for the random initial data studied.  Whether the resulting homothetic state also admits a full critical Newtonian dual will be answered later. The evolution also selects a nontrivial asymptotic distribution of the particles. As the initial transient settles and the dynamics enters a homothetic expansion, the scale-normalized configuration, or shape, converges to a hollow shell as shown in \cref{fig:p2x-1}.

	\subsection{Particle-based Newtonian Matter Cosmology as a Comparison}
	The \(t^{2/3}\) scaling of the two- and many-body trajectories, and the zero-energy Newtonian dual for two particles found above have a familiar counterpart in flat matter-dominated cosmology.  To state that comparison precisely, we now recall its Newtonian particle-mechanical formulation.  In the relativistic description, the cosmological principle of spatial homogeneity and isotropy leads to the Friedmann-Lema\^{\i}tre-Robertson-Walker ansatz, with a single scale factor \(R(t)\) governed by
	\begin{equation}
		\left(\frac{\dot{R}}{R}\right)^2 = \frac{8 \pi G \rho}{3} - \frac{k}{R^2} + \frac{\Lambda}{3}.\label{eq:Friedmann_equation}
	\end{equation}
	For the matter-dominated epoch, the same scale evolution follows from Newtonian gravity~\cite{10.1093/qmath/os-5.1.64,Ellis_2013,Ellis_2015}.  Consider \(N\) gravitating point particles in the homothetic form
	\begin{equation}
		\mathbf{x}_i(t)=R(t)\,\mathbf{q}_i,
		\qquad
		\sum_i m_i\mathbf q_i=0,
	\end{equation}
	where the dimensionless shape coordinates \(\{\mathbf q_i\}\) are time independent and form a comoving frame.  Newton's equations,
	\begin{equation}
		\ddot{\mathbf{x}}_i = \sum_{j \neq i} G m_j \frac{\mathbf x_j-\mathbf x_i}{|\mathbf x_j-\mathbf x_i|^3},\qquad i=1,\ldots,N, \label{eq:Newton_eom}
	\end{equation}
	close on the single scale factor only if the shape is a central configuration~\cite{Ellis_2013},
	\begin{equation}
		\sum_{j\neq i}Gm_j\frac{\mathbf q_j-\mathbf q_i}{|\mathbf q_j-\mathbf q_i|^3}
		= -\lambda\,\mathbf q_i,
		\qquad i=1,\ldots,N ,
		\label{eq:newtonian_central_config}
	\end{equation}
	for a constant \(\lambda\) independent of \(i\).  The scale then satisfies
	\begin{equation}
		\ddot R=-\frac{\lambda}{R^2},
		\qquad
		\frac{1}{2}\dot R^2-\frac{\lambda}{R}=\varepsilon ,
		\label{eq:newtonian_scale_energy}
	\end{equation}
	where \(\varepsilon\) is the Newtonian energy per unit scale degree of freedom.  Dividing the second equation by \(R^2\) gives the Friedmann form for dust, with \(\rho\sim R^{-3}\) and curvature proportional to \(-\varepsilon\). Thus, \(\varepsilon=0\) is the flat Einstein-de Sitter solution \(R(t)\propto t^{2/3}\), while \(\varepsilon<0\) and \(\varepsilon>0\) give closed recollapsing and open universes, respectively.

    Furthermore, the perturbed equations precisely replicate the linear structure growth rate characteristic of relativistic Friedmann-Lema\^{i}tre-Robertson-Walker models~\cite{Ellis_2015}.
	This construction underlies the intuition behind cosmological \(N\)-body simulations and provides a precise particle-mechanical version of matter-dominated FLRW evolution~\cite{Bagla_1997,Bertschinger,Ellis_2013}.  It also separates the ingredients combined in that solution: homothetic scale evolution, the spatial shape and its bulk distribution, and the scale energy are independent restrictions on Newtonian initial data.  Viewed as an autonomous Newtonian description, reproducing a flat, homogeneous matter-dominated solution together with an arrow of time associated with irreversible growth of structure therefore involves three distinct conditions:
	
	\textbf{A. Central configurations and homogeneity:} Generic Newtonian \(N\)-body initial data do not produce homothetic expansion.  The shape coordinates must satisfy \cref{eq:newtonian_central_config}, which defines a special subset of initial data.  Ellis and Gibbons emphasize this point~\cite{Ellis_2013}:
	
	\emph{``To good approximation, we currently see a FLRW type homothetic expansion. But in order to get such a flow, the initial positions of the particles must be constrained to satisfy [central configurations]. What kind of explanation can one give for such a fine-tuning of the initial data in Newtonian cosmology?''}
	
	A central configuration can, moreover, be highly inhomogeneous.  Obtaining an approximately homogeneous bulk profile requires selecting the appropriate central configurations, associated with distinguished extrema---in particular, global or near-global minima---of the gravitational shape potential~\cite{Battye_2003,Lourenco:2026lbr,Lourenco:2026uto}.  Homogeneity\footnote{In this discrete framework, ``homogeneity'' is restricted to a statistical uniformity of the discrete distribution when coarse-grained over suitable volumes, achievable only approximately in the interior for large N. Global exact homogeneity fails due to the finite boundary and unique centre. The cosmological principle therefore emerges as a restricted, approximate property rather than a globally exact postulate.} is therefore an additional condition within an already restricted class of shapes.
	
	\textbf{B. Zero energy and flatness:} Even after imposing homothetic motion, \cref{eq:newtonian_scale_energy} admits open, flat and closed solutions.  In the Newtonian formulation, curvature is controlled by the conserved total energy \(\varepsilon\): the flat Einstein-de Sitter universe is the critical solution \(\varepsilon=0\), for which the kinetic and potential terms balance and \(R(t)\propto t^{2/3}\); generic scale data have \(\varepsilon\neq0\). Thus the observed near-flatness of the universe translates into the statement that the Newtonian energy must lie very close to the critical surface.  This is the Newtonian version of the flatness problem.
	
	\textbf{C. Low-entropy conditions and the arrow of time:} In the standard cosmological account, the aligned arrows of scale expansion and structure growth are traced to a special smooth, low-entropy past---the past hypothesis~\cite{albert2000time,Penrose1979,Ellis_2015}. In the Newtonian formulation, it corresponds to a perturbed, expanding, homogeneous, homothetic initial condition that serves as a special, low-entropy past~\cite{Ellis_2015}. Alternative Janus-point models generate a bidirectional gravitational arrow dynamically~\cite{CarrollChen2004spontaneousinflationoriginarrow,CarrollChen_2005,BarbourKoslowskiMercati_PhysRevLett.113.181101,GoldsteinTumulkaZanghi_arrow_PhysRevD.94.023520}, although the constructions considered so far impose restrictions such as vanishing total energy, momentum and angular momentum and do not simultaneously select homothetic expansion {and hence cannot be regarded as precise cosmological models.}

   The inverse-distance fracton model therefore bears a partial but striking resemblance to the dynamics of a flat matter-dominated Newtonian cosmology.  Its exact two-particle solution asymptotically exhibits Einstein-de Sitter scaling, \(r(t)\propto |t|^{2/3}\), while its radial sector has an exact, critical, zero-energy Newtonian gravitational dual.  Our numerical results indicate that the Einstein-de Sitter scale law persists in the many-body system.  The late-time scale-normalized distribution approaches the shell shown in \cref{fig:p2x-1}, inconsistent with a homogeneous bulk profile to be expected from a valid cosmological model.  Moreover, although the critical Newtonian gravitational dual is exact for two particles, we will show in \cref{sec:full_newtonian_dual} that it does not extend to the many-body dynamics of this model.  In the next section, we embed the inverse-distance warm-up in a more general two-parameter family of scale-invariant Hamiltonians, whose fixed points can be analyzed systematically, thereby making precise both the similarities and the differences between this model and flat matter-dominated cosmology.  Within this family we will identify a different distinguished model whose attractor dynamics combines Einstein-de Sitter scaling, a homogeneous large-\(N\) distribution, Newtonian central configurations and a critical Newtonian gravitational dual, and which also exhibits an emergent arrow of time.  Thus, in the distinguished model, all the fine-tuned conditions of Newtonian cosmology identified above arise dynamically as attractor properties.

	\section{Scale-invariant Fractons and Their Fixed Points}
	\label{sec:scale_invt_fractons}
	
	\subsection{Scale-invariant Fractons: Hamiltonians, Two-Particle Dynamics and Symmetries} \label{Scale_inv}
	To make these comparisons precise and identify the distinguished model described above, we now embed the inverse-distance Hamiltonian \cref{eq:H_fracton_local} in a two-parameter family of scale-invariant fracton Hamiltonians,
	\begin{equation}
		H_{(\alpha,\beta)}\big(\{\mathbf{x}_i\},\{\mathbf{p}_i\}\big) = \sum_{j<i} 
		|\mathbf{p}_i-\mathbf{p}_j|^{\alpha}\;|\mathbf{x}_i-\mathbf{x}_j|^{\beta},
		\qquad \alpha,\beta\in\mathbb{R}. \label{eq:H_alpha_beta}
	\end{equation}
Unless otherwise stated, we will only consider $\alpha, \beta \neq 0$, $\alpha \neq \beta$, and assume that the system has non-vanishing energy.  All dynamical statements are made on time intervals contained in a smooth patch of phase space, away from any collisions or vanishing pair-momentum differences at which there is a singularity or non-differentiability. 

To get an immediate sense of \cref{eq:H_alpha_beta}, consider the two-particle dynamics in one spatial dimension.  As before, the trajectories are equivalent to those of a pair of Newtonian particles with vanishing Newtonian energy $E_N=0$.  The effective interaction and the separation $r = |x_1 - x_2|$ are
	\begin{align}
		U_{\text{eff}} &= -\alpha^2 E_f^{\frac{2(\alpha-1)}{\alpha}} r^{\frac{2\beta}{\alpha}}, \qquad
		r(t)= \left|2 (\alpha - \beta) E_f^{\frac{(\alpha -1)}{\alpha}} (t - t_0)\right|^{\frac{\alpha}{(\alpha - \beta)}}  .
		\label{eq:two_particle_alphabeta}
	\end{align}
	The three-dimensional problem has an additional invariant,
	\(\boldsymbol\ell=(\mathbf x_1-\mathbf x_2)\times
	(\mathbf p_1-\mathbf p_2)\).  For \(\alpha\ne\beta\), its exact solution
	shows that every finite-\(\ell\) trajectory approaches the radial scaling
	form shown in \cref{eq:two_particle_alphabeta}, {\em{i.e.}},
   \begin{equation}
		\sin\chi=O(|t|^{-1}),\qquad
		r(t)\sim
		\left|
		2(\alpha-\beta)E_f^{\frac{\alpha-1}{\alpha}}t
		\right|^{\frac{\alpha}{\alpha-\beta}},
		\qquad |t|\longrightarrow\infty ,
		\label{eq:alphabeta_two_body_3d_highlights}
	\end{equation}
	where \(\chi\) is the angle between the relative position and the canonical
	momentum. The exact Newtonian dual lives again on the invariant radial
	submanifold \(\boldsymbol\ell=0\).  With Newtonian canonical momenta
	\(\boldsymbol\pi_a=\dot{\mathbf x}_a\), define
	\begin{equation}
		\kappa:=\frac{2\beta}{\alpha},\qquad
		g_{\alpha,\beta}(E_f)
		:=\alpha^2E_f^{\frac{2(\alpha-1)}{\alpha}} .
	\end{equation}
	The radial fracton trajectories are generated by the pairwise
	Newtonian Hamiltonian
	\begin{equation}
		H^{\rm dual}
		=\frac12\sum_{a=1}^2|\boldsymbol\pi_a|^2
		-g_{\alpha,\beta}(E_f)
		|\mathbf x_1-\mathbf x_2|^\kappa,
		\qquad E_{N}=0 .
		\label{eq:alphabeta_two_body_dual_highlight}
	\end{equation}
	Thus an exact dual requires radial, zero-energy Newtonian initial data
	(after removing center-of-mass motion).  For nonzero \(\ell\), the
	finite-time fracton orbit is not generated by this position-only
	Hamiltonian, although \cref{eq:two_particle_alphabeta} shows
	that its leading asymptotic curve agrees with an appropriately selected
	radial zero-energy orbit.  The explicit solution is shown in
	\cref{app:two_particle_3d}.
	
	For the following one-parameter \emph{Einstein-de Sitter (EdS) branch}, 
	\begin{equation}
		\alpha = -2\beta. \label{eq:EdS_branch}
	\end{equation}
	one has \(\kappa=-1\) and
	\(\alpha/(\alpha-\beta)=2/3\).  Hence every finite-angular-momentum
	two-particle trajectory is asymptotically Einstein-de Sitter, while
	the invariant radial trajectories have the exact critical Newtonian
	gravitational dual \eqref{eq:alphabeta_two_body_dual_highlight}, with the
	dynamically generated constant
	\(G_{\rm eff}^{(2)}=g_{\alpha,\beta}(E_f)\).  The model in
	\cref{eq:H_fracton} with pair-inertia function \(K(x)=1/|x|\)
	corresponds, up to its factor of \(1/2\), to
	\((\alpha,\beta)=(2,-1)\) and lies on this line.  A second useful feature
	is the softening of non-locality. Unlike the original models in
	\cref{eq:H_fracton}, the Hamiltonian \cref{eq:H_alpha_beta} is not local
	for \(\beta>0\).  Nevertheless, when \(\alpha/\beta<0\), the radial
	two-particle trajectories are governed by an effective Newtonian
	potential that decays with separation.  This includes
	\eqref{eq:EdS_branch}.  We next ask how these features extend to
	many-body dynamics.

	The Hamiltonian family \cref{eq:H_alpha_beta} has independent translation
	symmetries in \(\mathbf{x}\) and \(\mathbf{p}\), and a common
	(diagonal) rotation symmetry
	\((\mathbf x_i,\mathbf p_i)\mapsto
	(O\mathbf x_i,O\mathbf p_i)\)~\footnote{Although the scalar function \(H\) is
	unchanged by rotating the two sets separately, independent position and
	momentum rotations are not canonical and do not map Hamiltonian
	trajectories to Hamiltonian trajectories.}.   The model also has an
	important scaling symmetry.  A classical Hamiltonian system is
	scale-invariant, or mechanically similar \cite{landau1976mechanics}, if
	there exists a rescaling of phase-space coordinates and time
	\begin{equation}
		\mathbf{x}_i \mapsto \tilde{\mathbf{x}}_i = a\,\mathbf{x}_i,\qquad
		\mathbf{p}_i \mapsto \tilde{\mathbf{p}}_i = b\,\mathbf{p}_i,\qquad
		t \mapsto \tilde{t} = c\,t,
		\label{eq:scaling}
	\end{equation}
	with positive constants \(a,b,c\), such that the equations of motion are covariant and every solution of Hamilton's equations is mapped to another. Hamilton's equations of motion for \eqref{eq:H_alpha_beta} are
	\begin{align}
		\dot{\mathbf{x}}_i &= \frac{\partial H_{(\alpha,\beta)}}{\partial \mathbf{p}_i}
		= \alpha\sum_{j\neq i}|\mathbf{p}_{ij}|^{\alpha-2}\,\mathbf{p}_{ij}\,
		|\mathbf{x}_{ij}|^{\beta}, \quad\\
		\dot{\mathbf{p}}_i &= -\frac{\partial H_{(\alpha,\beta)}}{\partial \mathbf{x}_i}
		= -\beta\sum_{j\neq i}|\mathbf{p}_{ij}|^{\alpha}\,
		|\mathbf{x}_{ij}|^{\beta-2}\,\mathbf{x}_{ij}, \label{eq:H_alpha_beta_eom}
	\end{align}
	with \(\mathbf{x}_{ij}\equiv\mathbf{x}_i-\mathbf{x}_j\) and \(\mathbf{p}_{ij}\equiv\mathbf{p}_i-\mathbf{p}_j\).  Under the scaling transformation in \eqref{eq:scaling},
	\begin{equation}
		H_{(\alpha,\beta)}(\{\mathbf{x}_i\},\{\mathbf{p}_i\}) \mapsto a^{\beta}b^{\alpha} H_{(\alpha,\beta)}(\{\mathbf{x}_i\},\{\mathbf{p}_i\}). \label{eq:Hscale}
	\end{equation}
	The scaled equations of motion \eqref{eq:H_alpha_beta_eom} are covariant provided
	\begin{equation}
		\frac{a}{c} = \frac{a^{\beta}b^{\alpha}}{b},\qquad
		\frac{b}{c} = \frac{a^{\beta}b^{\alpha}}{a} \implies c =  a^{1-\beta}\,b^{1-\alpha} \label{eq:tau}
	\end{equation}
	Thus any positive \(a\) and \(b\) define a scaling symmetry, once the time dilation is chosen as in \eqref{eq:tau}.  This two-parameter scaling freedom is stronger than the equivalent single-parameter one in the Kepler problem,
	\begin{equation}
		H_{\rm Kepler}=\frac{\mathbf p^2}{2m}-\frac{k}{r},\qquad
		\mathbf r\to\lambda\mathbf r,\quad
		\mathbf p\to\lambda^{-1/2}\mathbf p,\quad
		t\to\lambda^{3/2}t ,
		\label{eq:Kepler_scale}
	\end{equation}
	By contrast, in \cref{eq:H_alpha_beta} positions and momenta can be rescaled independently, producing the same trajectory with a different clock rate.

	\subsection{Dynamical Similarity Reduction and Scale-Shape Separation}
	The homogeneity of \eqref{eq:H_alpha_beta} allows the overall scale to be separated from the shape degrees of freedom.  This operation is a \emph{dynamical similarity reduction} \cite{bravetti2022scalingsymmetriescontactreduction}: the full Hamiltonian flow remains symplectic, but the reduced shape dynamics need not preserve a symplectic volume and may therefore possess attractors.  We begin with the following redefinitions of phase-space coordinates,
	\begin{equation}
		\mathbf{s}_i = \frac{\mathbf{x}_i-\mathbf{X}_{\rm C}}{R}, \quad \mathbf{u}_i = R^{\beta/\alpha}\left(\mathbf{p}_i-\mathbf{P}_{\rm C}\right),    \quad    R = \sqrt{\sum_{i=1}^N |\mathbf{x}_i-\mathbf{X}_{\rm C}|^2}, \label{eq:shape_coordinates}
	\end{equation}
	where $\mathbf{X}_{\rm C}$ and $\mathbf{P}_{\rm C}$ are defined in \cref{eq:conservation_posnmom}, $R>0$ is the overall position-space scale and the shape coordinates are thus constrained to satisfy
	\begin{equation}
		\sum_{i=1}^N |\mathbf{s}_i|^2=1, \qquad  \sum_{i=1}^N \mathbf{s}_i = \sum_{i=1}^N \mathbf{u}_i = 0. \label{eq:shape_coordinate_constraints}
	\end{equation}
	We also define a reparametrized time coordinate as
	\begin{equation}
		d\tau = \frac{dt}{R^{(\alpha - \beta)/\alpha}}.
		\label{eq:shape_t}
	\end{equation}
	In these variables, the Hamiltonian and thus energy become independent of the overall scale $R$,
	\begin{equation}
		H_{(\alpha,\beta)} = H^0_{(\alpha,\beta)}(\{\mathbf{s}_i,\mathbf{u}_i\}),\qquad
		H^0_{(\alpha,\beta)}=\sum_{j<i}|\mathbf{u}_i-\mathbf{u}_j|^\alpha|\mathbf{s}_i-\mathbf{s}_j|^\beta .
	\end{equation}
	The function \(H^0_{(\alpha,\beta)}\) is conserved along the full motion, although it is not itself the generator of the reduced \((\mathbf{s}_i,\mathbf{u}_i)\) dynamics.  The latter is determined from the equations of motion for the original coordinates, generated by the Hamiltonian, and then by making the transformation in \cref{eq:shape_coordinates}.  This gives
	\begin{align}
		\frac{d\mathbf{s}_i}{d\tau} &= \alpha\left(\mathbf F^i_{(\alpha,\beta)}-A_{(\alpha,\beta)}\mathbf{s}_i\right), \quad \frac{d\mathbf{u}_i}{d\tau} =- \beta\left(\mathbf G^i_{(\alpha,\beta)} - A_{(\alpha,\beta)}\mathbf{u}_i\right),\label{eq:shape_coordinates_eom}\\
		\frac{dR}{d\tau} &= \alpha A_{(\alpha,\beta)} R, \label{eq:scale_tau}
	\end{align}
	where
	\begin{align}
		\mathbf F^i_{(\alpha,\beta)}(\{\mathbf{s},\mathbf{u}\}) &=
		\sum_{j\neq i} 
		|\mathbf{s}_i-\mathbf{s}_j|^\beta |\mathbf{u}_i-\mathbf{u}_j|^{\alpha-2} (\mathbf{u}_i-\mathbf{u}_j),\\
		\mathbf G^i_{(\alpha,\beta)}(\{\mathbf{s},\mathbf{u}\}) &=
		\sum_{j\neq i}|\mathbf{u}_i-\mathbf{u}_j|^\alpha
		|\mathbf{s}_i-\mathbf{s}_j|^{\beta-2}(\mathbf{s}_i-\mathbf{s}_j),\\
		A_{(\alpha,\beta)} &= \sum_i \mathbf{s}_i\cdot\mathbf F^i_{(\alpha,\beta)}.
	\end{align}

    The Hamiltonian \eqref{eq:H_alpha_beta} is built only from pairwise differences \(\mathbf{x}_i-\mathbf{x}_j\) and \(\mathbf{p}_i-\mathbf{p}_j\), and is therefore invariant under translations and common rotations in both position and momentum space.  Together with the scaling symmetry described in \cref{Scale_inv}, this means that the physically relevant data can be expressed in terms of dimensionless ratios and relative angles in the combined position-momentum configuration.  In this sense the model implements a version of spatial relationalism familiar from relational particle mechanics~\cite{10.1098/rspa.1982.0102,Barbour_2003,Anderson_2006}, and extends the same idea to phase space. 
	
	\subsection{Radial Fixed Point Solutions are Central Configurations of Riesz Potentials}
	At this stage we can see how homothetic evolution can emerge. The key point is that the equations for \((\mathbf{s}_i,\mathbf{u}_i)\) in \cref{eq:shape_coordinates_eom}  close on themselves; the evolution of the scale \(R\) is determined only after the shape trajectory is known.  The shape dynamics is therefore autonomous. Suppose there exist fixed points for the latter: 
	\begin{equation}
		\frac{d\mathbf{s}_i}{d\tau} = 0 \implies \mathbf F^i_{(\alpha,\beta)}=A_{(\alpha,\beta)}\mathbf{s}_i, \quad \frac{d\mathbf{u}_i}{d\tau} = 0 \implies \mathbf G^i_{(\alpha,\beta)} = A_{(\alpha,\beta)}\mathbf{u}_i. \label{eq:fixed_point_equations}
	\end{equation}
	Then the shape coordinates $\{\mathbf{s}_i,\mathbf{u}_i\}$ would become time-independent, as would $A$. Thus the evolution of the remaining scale degree of freedom can be determined in the scaled time coordinate $\tau$:
	\begin{equation}
		R(\tau) \sim e^{\alpha A_{(\alpha,\beta)} \tau}. \label{eq:scale_exp}
	\end{equation}
	Depending on the sign of $\alpha A_{(\alpha,\beta)}$, we have an expanding or contracting solution. Our main interest is in the expanding branch of fixed points, which gives the following scale evolution in real time $t$:
	\begin{equation}
		R(t)  \sim |t|^{\alpha/(\alpha - \beta)}.\label{eq:scale_power_law}
	\end{equation}
	Remarkably, this is the same scaling dynamics as the one obtained from the two-particle trajectory in \cref{eq:two_particle_alphabeta}. In particular, every model on the EdS branch of \cref{eq:EdS_branch} exhibits Einstein-de Sitter scaling \(R\sim t^{2/3}\) at radial fixed points, consistent with the numerical observations of \cref{sec:Warmup}. The questions before us are
	\begin{enumerate}
		\item Do such fixed points exist for general ($\alpha,\beta$) and on the EdS branch \eqref{eq:EdS_branch} in particular?
		\item Are they stable (attractors) on the expanding branch?
		\item Do they lead to asymptotic homogeneous and isotropic distributions for large $N$?
	\end{enumerate}
	
	If such dynamical fixed points exist, the Einstein-de Sitter evolution becomes stable at late times.  We will show that the fixed points exist, use known Riesz-equilibrium results to identify their regular large-\(N\) distributions, and prove a local stability result on the EdS branch under a transparent auxiliary positivity condition.  
	
	The scale-shape reduction in \cref{eq:shape_coordinates,eq:shape_t} is the mechanism by which attractors can coexist with Hamiltonian evolution.  The full phase-space flow still satisfies the Liouville theorem.  However, after quotienting by scale, the shape flow is no longer a Hamiltonian flow with the original symplectic volume, so stable fixed points in the shape space are not forbidden.

	We look for radial fixed-point solutions using the ansatz
	\begin{equation}
		\mathbf{u}_i=\gamma\,\mathbf{s}_i,\qquad \gamma\in\mathbb R ,
		\label{eq:fixed_ansatz}
	\end{equation}
	with the same \(\gamma\) for every particle.  Substituting \cref{eq:fixed_ansatz} into \cref{eq:fixed_point_equations} gives $\mathbf G^i_{(\alpha,\beta)} = \gamma \mathbf F^i_{(\alpha,\beta)}$ and reduces the two fixed-point equations in \cref{eq:fixed_point_equations} to one,

	\begin{equation}
		\mathbf H_{\alpha + \beta}^i = S_{\alpha + \beta}\,\mathbf{s}_i, \qquad i=1,\ldots,N .
		\label{eq:shape_constraint_cc}
	\end{equation}
	We have introduced the following quantities indexed by the single real number \(q=\alpha+\beta\):
	\begin{align}
		\mathbf H_{q}^i &= 
		\sum_{j\neq i}
		\frac{\mathbf{s}_i-\mathbf{s}_j}
		{|\mathbf{s}_i-\mathbf{s}_j|^{2-q}}, \quad  S_{q}=\sum_i \mathbf{s}_i\cdot\mathbf H^{i}_q =\sum_{i<j} |\mathbf{s}_i - \mathbf{s}_j|^{q}. \label{eq:HS_definition}
	\end{align}
	Thus, the fixed-point condition in \cref{eq:fixed_point_equations} reduces to the purely geometric constraint \eqref{eq:shape_constraint_cc}.  If the shape reaches such a fixed point, the remaining dynamics is the scale evolution in \cref{eq:scale_power_law}.  The solutions of \cref{eq:shape_constraint_cc}, \emph{central configurations}, determine the asymptotic particle distribution.  We also have
	\begin{equation}
		\mathbf F_{(\alpha,\beta)}^i = \gamma|\gamma|^{\alpha-2}
		\mathbf H_{\alpha+\beta}^i, \qquad A_{(\alpha,\beta)}=\gamma|\gamma|^{\alpha-2}S_{\alpha + \beta} \label{eq:FA_def}
	\end{equation}
	The equations in \cref{eq:shape_constraint_cc} can be recast as the critical-point equations of a \emph{shape potential} $\Xi_q(\{ \mathbf{s}_i\})$~\cite{Battye_2003}, in shape coordinates $\{\mathbf{s}_j\}$ subject to the constraint in \cref{eq:shape_coordinate_constraints}
	\begin{equation}
		\nabla_{\!\mathrm{shape}} \Xi_q =0, \qquad \Xi_q = \begin{cases}
			\sum_{i<j} |\mathbf{s}_i - \mathbf{s}_j |^q & q \neq 0\\
			\sum_{i<j} \log|\mathbf{s}_i - \mathbf{s}_j | & q = 0
		\end{cases}. \label{eq:shape_potential}
	\end{equation}
	Note that \(\Xi_q=S_q\) for \(q\neq0\).  After blowing up the shape coordinates by \(\mathbf x_i=\sqrt N\,\mathbf s_i\), the fixed-size problem can be mapped exactly to the attractive-repulsive power-law energy classified by Frank and Matzke~\cite{FrankMatzke2025}.  The normalization and sign dictionary are given in \cref{app:LargeNDistributions}.  The shape potential is distinct from the effective two-particle Newtonian potential in \cref{eq:two_particle_alphabeta}, except at \((\alpha,\beta)=(-2,1)\) when both have the Newtonian gravitational exponent. This distinguished model will be the focus of the following sections.  For now, we keep \(\alpha,\beta\) general.
	
	\subsection{Regular Extremal Central Configurations and Their Large-\texorpdfstring{$N$}{N} distributions}
    \label{sec:regular_asymptotic_distributions}
	\begin{figure}[!ht]
		\centering
		\includegraphics[height=0.3\linewidth]{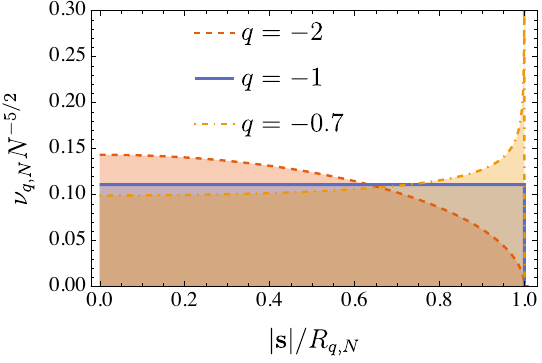}
		\includegraphics[height=0.3\linewidth]{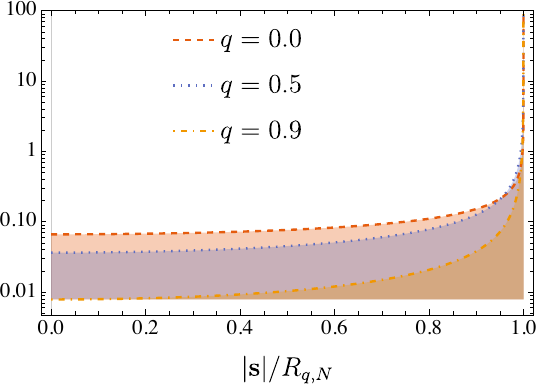}
		\caption{Large-\(N\) regular isotropic shape densities.  The filled-ball configuration has the radial density in \cref{thm:distribution}; the special point \(q=-1\) is a homogeneous solid ball, while \(q=1\) is the condensation point to a uniform spherical shell.}
		\label{fig:nu}
	\end{figure}
	Finite-\(N\) central configurations, corresponding to solutions of \cref{eq:shape_constraint_cc}, can take various forms: collinear configurations, clustered configurations, and regular non-clustered isotropic arrangements~\cite{Battye_2003} .  In this subsection we focus on the regular extrema of \(\Xi_q\) and their large-\(N\) continuum particle-number density
	\begin{equation}
		\nu_N(\mathbf s)
		=\sum_{i=1}^N\delta^{(3)}(\mathbf s-\mathbf s_i),
		\qquad
		\int_{\R^3}\nu_N(\mathbf s)\dd^3s=N.
		\label{eq:shape-empirical-density}
	\end{equation}
	where the constraints on the shape coordinates in \cref{eq:shape_coordinate_constraints} become
	\begin{equation}
		\int_{\R^3}\mathbf s~\nu_N(\mathbf s)\dd^3s=0,\qquad
		\int_{\R^3}|\mathbf s|^2~\nu_N(\mathbf s)\dd^3s=1.
	\end{equation}

	The regular global extrema depend only on the exponent \(q=\alpha+\beta\) in the shape potential \(\Xi_q\).  As a function of the squared pair distances, \(\Xi_q\) is strictly convex for \(q<0\) and \(q>2\), strictly concave for \(0<q<2\), and flat for \(q=2\).  The continuum formulas below follow from the classification of the corresponding attractive-repulsive power-law energy~\cite{FrankMatzke2025}, after imposing our unit-second-moment normalization as explained in \cref{app:LargeNDistributions}.
	
\begin{theorem}[Power-law equilibrium measures]
		\label{thm:distribution}
		Regular central configurations obtain at the global minimum \( (-3<q<0\) and \(2<q<4)\) and maximum   \( (0 \le q<2)\) of $\Xi_q$. The large-$N$ distributions weakly converge to isotropic forms as follows. For \(-3<q<1\), the distribution is supported on a weighted ball $\mathbb B^3$ of radius $R_{q,N}$ with density,
		\begin{equation}
			\nu_{q,N}(\mathbf s)
			=\begin{cases}
				C_{q,N}
				\left[(R_{q,N})^2-|\mathbf s|^2\right]^{-(q+1)/2} \quad & |\mathbf s| < R_{q,N}\\
				0  \quad &   |\mathbf s| \ge  R_{q,N} 
			\end{cases} \label{eq:nu_ball}
		\end{equation}
		where
		\begin{equation}
			C_{q,N}
			=\frac{N}{2\pi R_{q,N}^{\,2-q}
				B\!\left(\frac32,\frac{1-q}{2}\right)}, \qquad R_{q,N}
			=\sqrt{\frac{4-q}{3N}} \label{eq:Sphere_Normalization_Radius}
		\end{equation}
		and \(B(a,b)\) is the Euler beta function,
		\begin{equation}
			B(a,b)=\int_0^1t^{a-1}(1-t)^{b-1}\dd t. \label{eq:Beta_function}
		\end{equation}
		
		At \(q=1\), the distribution condenses onto its boundary.  For
		\(1\le q<2\) and \(2<q<4\), the distribution is supported on a sphere
		$\mathbb S^2$ of radius \(R_N=1/\sqrt N\):
		\begin{equation}
			\nu_{q,N}(\mathbf s) =   \frac{N^2}{4\pi}\,
			\delta\!\left(|\mathbf s|-\frac1{\sqrt N}\right). \label{eq:nu_shell}
		\end{equation}
        At \(q=2\), by contrast,
		\(\Xi_2=N\sum_i|\mathbf s_i|^2-|\sum_i\mathbf s_i|^2=N\)
		for every centered, normalized configuration.  The variational problem is
		therefore degenerate and does not select a large-\(N\) distribution.
	\end{theorem}
	
	\Cref{app:LargeNDistributions} gives the coordinate, normalization, and sign dictionary needed to connect \cref{thm:distribution} to Ref.~\cite{FrankMatzke2025}; it also states the precise scope of the large-\(N\) convergence claim.  \Cref{thm:distribution} explains the numerical observations of \cref{sec:Warmup}.  For \((\alpha,\beta)=(2,-1)\), \(q=1\), the large-\(N\) shape distribution is a shell of radius \(1/\sqrt N\), consistent with the radial histogram in \cref{fig:p2x-1}.   \Cref{fig:nu} shows representative weighted ball distributions for \(-3<q<1\).
	
	For the distinguished model \((\alpha,\beta)=(-2,1)\), we have \(q=-1\), and the distribution is a minimizer of $\Xi_{-1}$ and converges to a homogeneous solid ball of radius \(R_{-1,N}=\sqrt{5/(3N)}\).
	\begin{equation}
		\nu_{-1,N}(\mathbf s)
		=\begin{cases}
			\left(\frac{3}{5} \right)^{\frac{3}{2}} \frac{3}{4 \pi} N^{\frac{5}{2}}  \approx 0.11 N^{\frac{5}{2}}
			\qquad & |\mathbf s| < \sqrt{5/(3N)}\\
			0  \qquad &   |\mathbf s| \ge  \sqrt{5/(3N)} 
		\end{cases} \label{eq:nu_m1}
	\end{equation}
	
	Reintroducing the scale \(R(t)\), the radial fixed-point dynamics on the
	expanding branch produces an expanding ball for \(-3<q<1\), or an
	expanding shell for \(1\le q<2\) and \(2<q<4\), of radius
	\(\sim t^{\alpha/(\alpha-\beta)}N^{-1/2}\).  The case \(q=2\) is
	degenerate and selects no profile.  For the distinguished model
	\((\alpha,\beta)=(-2,1)\) on the EdS branch~\eqref{eq:EdS_branch}, the
	global continuum profile is an expanding homogeneous ball of radius
	\(\sim t^{2/3}N^{-1/2}\).  We still have to establish whether a
	corresponding finite-\(N\) fixed point is dynamically selected.

	\subsection{Conditional Local Stability on the EdS Branch}
	\label{sec:conditional_local_stability}
	We now specialize to the EdS branch~\eqref{eq:EdS_branch} of scale-invariant fracton dynamics.  Since \(q=\alpha+\beta\), the EdS condition \(\alpha=-2\beta\) is equivalently
	\begin{equation}
		\alpha=2q,\qquad \beta=-q .
	\end{equation}
	We prove the following local stability theorem.
	
	\begin{theorem}
		\label{thm:stability}
		On the EdS branch, let \(q\neq0\) and consider a collision-free radial fixed point \(\mathbf u_i=\gamma\mathbf s_i\) whose shape is a nondegenerate regular extremum of \(\Xi_q\).  Work on a fixed-energy surface and quotient the translation and common-rotation zero modes.  If the auxiliary quadratic form \(\mathcal{C}_q\) defined below is strictly positive, then the expanding fixed point is linearly asymptotically stable whenever the constrained Hessian sign is \(\mathcal{Q}_q<0\).  This includes the regular minima for \(-3<q<0\) and the regular maxima for \(0<q<2\). Under the same positivity condition, saddle-type central configurations are linearly unstable.
	\end{theorem}
	\begin{proof}
		We sketch the proof briefly here and give the details in \cref{app:StabilityDetails}. We use \(\mathbf z\) to denote the list of vectors \(\{\mathbf z_i\}\), define \(\langle \mathbf r,\mathbf z\rangle:=\sum_i\mathbf r_i\cdot\mathbf z_i\), and write \(\mathbf z_{ij}:=\mathbf z_i-\mathbf z_j\).  At a central configuration, define the linear operations
		\begin{align}
			&(L\mathbf z)_i
			:=\sum_{j\ne i}s_{ij}^{q-2}\mathbf z_{ij}, \qquad
			(P\mathbf z)_i
			:=\sum_{j\ne i}
			s_{ij}^{q-4}\mathbf s_{ij}
			\left(
			\mathbf s_{ij}\mathbin{\cdot}\mathbf z_{ij}
			\right),\qquad (I \mathbf z)_i := \mathbf z_i
			\label{eq:LP-definition}\\
			&Q_q:=L+(q-2)P-S_qI,\qquad C_q:=L-2P+S_qI \label{eq:CQ-definition}\\
			&\mathcal{Q}_q[\mathbf z] := \langle \mathbf z, Q_q \mathbf z \rangle, \qquad\mathcal{C}_q[\mathbf z] := \langle \mathbf z, C_q \mathbf z \rangle .
		\end{align} 
		Here \(s_{ij}=|\mathbf s_i-\mathbf s_j|\), and \(S_q\) is defined in \cref{eq:HS_definition}. We now consider infinitesimal perturbations of the fixed point ansatz in \cref{eq:fixed_ansatz} on the expanding branch,
		\begin{equation}
			\mathbf s_i\longmapsto\mathbf s_i+\epsilon\mathbf a_i,
			\qquad
			\mathbf u_i\longmapsto\gamma\mathbf s_i+\epsilon\mathbf b_i,
			\label{eq:perturbations}
		\end{equation}
		where \(\mathbf s\) satisfies the central-configuration equation~\eqref{eq:shape_constraint_cc}.  Fixed energy gives the additional linear constraint \(\langle \mathbf s,\mathbf b\rangle=0\).  The linearized equations for \(\mathbf a,\mathbf b\) follow from \cref{eq:shape_coordinates_eom}; it is convenient to work with
		\begin{equation}
			\mathbf w:=\frac{\mathbf b}{\gamma}-\mathbf a, \qquad \mathbf v:=\mathbf a+2\mathbf w
		\end{equation}
		and with the expanding time
		\begin{equation}
			\rho:=q\gamma|\gamma|^{2(q-1)}\tau ,
		\end{equation}
		which increases with \(\tau\) on the expanding branch.  The fixed-energy linearization reduces to
		\begin{equation}
			\frac{\dd\mathbf v}{\dd\rho}
			=-2C_q\mathbf w, \qquad \frac{\dd\mathbf w}{\dd\rho}
			=-Q_q\mathbf v-3S_q\mathbf w. \label{eq:vw_eom}
		\end{equation}
		The two equations in \eqref{eq:vw_eom} combine into
		\begin{equation}
			\frac{\dd^2\mathbf v}{\dd\rho^2}
			+3S_q\frac{\dd\mathbf v}{\dd\rho}
			-2C_qQ_q\mathbf v=0. \label{eq:v_DampedOscillator}
		\end{equation}
		\Cref{eq:v_DampedOscillator} is a damped oscillator equation with positive friction coefficient \(3S_q\).  Moreover \(\mathcal Q_q\) is the constrained Hessian of the shape potential:
		\begin{equation}
			\delta^2\Xi_q
			=q\,\mathcal Q_q[\mathbf a].
		\end{equation}
		For  \(-3<q<0\), the extremum is a minimum of \(\Xi_q\), and since \(q<0\) this gives \(\mathcal Q_q<0\).  For \(0<q<2\), the extremum is a maximum of \(\Xi_q\), and again \(\mathcal Q_q<0\).  If \(C_q\) is positive, \(C_qQ_q\) has the same sign as \(Q_q\).  Thus on each physical mode we may write \(C_qQ_q\mathbf v=-\kappa\mathbf v\), with \(\kappa>0\).  For \(\mathbf v\sim e^{\lambda\rho}\), \cref{eq:v_DampedOscillator} gives
		\begin{equation}
			\lambda^2+3S_q\lambda+2\kappa=0 \implies \lambda_\pm
			=\frac{-3S_q\pm\sqrt{9S_q^2-8\kappa}}{2}.
		\end{equation}
		Both roots have negative real part.  Conversely, if \(Q_q\) has a positive direction, then the corresponding characteristic equation has a negative constant term and one positive root.  This gives the advertised stability and instability statements.
	\end{proof}
	
	The auxiliary positivity condition in \cref{thm:stability} is not vacuous.  The simplest regular central configuration already satisfies it analytically:
	\begin{proposition}[Equilateral three-body stability]
		\label{prop:equilateral_stability}
		For the centered unit equilateral \(N=3\) configuration on the EdS branch, \(C_q\) is positive on the physical perturbation space.  More explicitly, after removing translations, the tangent modes split into common rotations and two genuine shape distortions; on the shape modes
		\begin{equation}
			C_q=3I,\qquad Q_q=\frac32(q-2)I .
		\end{equation}
		Hence for every nondegenerate EdS value \(q<2\), \(q\ne0\), the expanding radial equilateral fixed point is linearly asymptotically stable after quotienting the common rotations.  For the distinguished model \(q=-1\), the shape-mode eigenvalues in the \(\rho\)-time of \cref{eq:vw_eom} are
		\begin{equation}
			\lambda_\pm=-\frac92\pm\frac{3\sqrt3}{2}i .
		\end{equation}
	\end{proposition}
	\begin{proof}
		Choose the centered equilateral configuration with side length one, so that \(\sum_i|\mathbf s_i|^2=1\) and \(S_q=3\).  On centered perturbations one has \(L=3I\).  The angular operator \(P\) vanishes on common rotations and equals \((3/2)I\) on the two shape-distortion modes.  Therefore \(C_q=L-2P+S_qI\) gives \(C_q=6I\) on rotations and \(C_q=3I\) on shape distortions, while \(Q_q=L+(q-2)P-S_qI\) gives \(Q_q=0\) on rotations and \(Q_q=\frac32(q-2)I\) on shape distortions.  The zero rotational modes are quotiented; the nonzero rotational companions have negative decay rate, and the shape modes have the stable sign for \(q<2\).  Substituting \(q=-1\), \(S_q=3\), \(C_q=3I\), and \(Q_q=-9I/2\) into \cref{eq:v_DampedOscillator} gives the displayed roots.  The same calculation is written out in coordinates in \cref{app:TriangleCheck}.  
	\end{proof}
	
	For larger $N$, we verify \(C_q>0\) numerically for the distinguished model~\footnote{We directly verify the local stability of central configurations forming minima of $\Xi_{-1}$ as fixed points; the positivity of $C_q$ naturally follows.}, but we do not yet have a uniform finite-\(N\) proof.  We therefore formulate the following conjecture.
	\begin{conjecture}
		\label{conj:Cq}
		On the EdS branch, \(\alpha=2q,\beta=-q\), the regular extrema with \(-3<q<2\) and \(q\neq0\) have \(C_q>0\) after quotienting the symmetry zero modes.  Hence their expanding radial fixed points are locally asymptotically stable.
	\end{conjecture}
	\Cref{thm:distribution,thm:stability,prop:equilateral_stability,conj:Cq} explain the shell-forming \((\alpha,\beta)=(2,-1)\) model of \cref{sec:Warmup} and motivate the distinguished homogeneous-ball model \((\alpha,\beta)=(-2,1)\), which we study for the rest of the paper.

	Some comments are in order here.  First, \cref{thm:stability} is a local result: it shows stability of radial fixed points whose shapes are regular central configurations for part of the EdS family, but it does not rule out other stable fixed points.  This caveat matters for the cluster-forming large-\(N\) dynamics discussed in \cref{sec:structure}.  Second, away from \(\alpha=-2\beta\) we do not yet have a comparable general stability criterion. Third, when \(\alpha=-2\beta\), for the range $2<q<4$ where regular distributions are possible, at the maximum of $\Xi_q$ as shown in \cref{thm:distribution,eq:nu_shell}, the fixed points are stable if $C_q<0$. However, even the central configuration of the equilateral triangle has $C_q>0$ and is therefore unstable. Similarly, for $-3<q<2$, saddle-type configurations are unstable. The corresponding \(N=3\) collinear instability is worked out in \cref{app:CollinearCheck}.

	\subsection{Flatness Without Fine-tuning: Dual Critical Newtonian Dynamics}
	\label{sec:full_newtonian_dual}
	We now return to a question anticipated by the two-particle discussion
	in \cref{sec:Warmup}: \emph{when can a radial fracton fixed point be
	lifted to a full pairwise Newtonian dynamics?}  To answer it, it is
	useful to begin with an independent fact about ordinary Newtonian
	mechanics.

	Consider \(N\) equal unit masses governed by the pairwise
	Hamiltonian
	\begin{equation}
		H_N^{(\kappa)}
		=\frac12\sum_i|\boldsymbol\pi_i|^2
		-G_\kappa\sum_{i<j}
		|\mathbf x_i-\mathbf x_j|^\kappa ,
		\qquad G_\kappa>0 .
		\label{eq:kappa_newtonian_hamiltonian}
	\end{equation}
	For \(\kappa\ne0\), Newtonian \(N\)-body systems with homogeneous potentials have
	homothetic solutions generated by central configurations
	\cite{DiacuPerezChavelaSantoprete2006,Maderna2013HomogeneousNBody};
	the inverse-distance case is reviewed in
	\cite{Battye_2003,Ellis_2013,Ellis_2015}.  In our normalization, choose
	initial data of the form
	\begin{equation}
		\mathbf x_i(t_0)=R_0\mathbf s_i,\qquad
		\boldsymbol\pi_i(t_0)=\dot R_0\mathbf s_i,\qquad
		\sum_i|\mathbf s_i|^2=1 ,
		\label{eq:kappa_homothetic_initial_data}
	\end{equation}
	and require the shape to be \(\kappa\)-central,
	\begin{equation}
		\mathbf H_\kappa^i
		=\sum_{j\ne i}
		\frac{\mathbf s_i-\mathbf s_j}
		{|\mathbf s_i-\mathbf s_j|^{2-\kappa}}
		=S_\kappa\mathbf s_i .
		\label{eq:dual_kappa_central}
	\end{equation}
	Newton's equations then preserve the shape:
	\begin{equation}
		\mathbf x_i(t)=R(t)\mathbf s_i,\qquad
		\ddot R
		=\kappa G_\kappa S_\kappa R^{\kappa-1}.
		\label{eq:kappa_newtonian_scale_equation}
	\end{equation}
	{The Newtonian energy for such homothetic solutions then reduces to}
	\begin{equation}
		E_N
		=\frac12\dot R^2-G_\kappa S_\kappa R^\kappa .
		\label{eq:kappa_newtonian_energy}
	\end{equation}
	On the critical surface \(E_N=0\), the radial velocity is fixed by
	\begin{equation}
		\dot R^2=2G_\kappa S_\kappa R^\kappa ,
	\end{equation}
	and, for \(\kappa\ne2\), integration gives the scaling solution
	\begin{equation}
		R(t)
		=\left[
		\frac{|2-\kappa|}{2}\sqrt{2G_\kappa S_\kappa}\,
		|t-t_0|
		\right]^{\frac{2}{2-\kappa}} .
		\label{eq:kappa_zero_energy_scaling}
	\end{equation}
	For \(\kappa=2\) the corresponding solution is exponential.  The case
	\(\kappa=0\) is degenerate: the pair potential is constant, the force
	vanishes, and homothetic free motion does not require a central 
	configuration.  Apart from this free case, a
	power law \(R\sim|t-t_0|^{2/(2-\kappa)}\) is not by itself enough to
	define a Newtonian \(N\)-body trajectory: the initial shape must also
	satisfy \cref{eq:dual_kappa_central}, and the initial radial velocity
	must place the solution on \cref{eq:kappa_newtonian_energy} with
	\(E_N=0\).

	We can now compare this standard Newtonian construction with our radial
	fracton fixed point.  From \cref{eq:scale_tau,eq:shape_t}, this obeys
	\begin{equation}
		\dot R=\alpha A_*R^{\beta/\alpha},
		\qquad
		R(t)\sim |t-t_0|^{\alpha/(\alpha-\beta)} .
		\label{eq:fracton_fixed_point_scale_recalled}
	\end{equation}
	This is precisely the scaling in
	\cref{eq:kappa_zero_energy_scaling} after identifying
	\begin{equation}
		\kappa=\frac{2\beta}{\alpha},
		\qquad
		G_\kappa S_\kappa=\frac{\alpha^2A_*^2}{2}.
		\label{eq:fracton_newtonian_identification}
	\end{equation}
	The scale equation therefore has a critical Newtonian interpretation.
	The full position-space trajectory, however, carries an additional
	condition.  The fracton fixed-point equations do not impose
	\(\kappa\)-centrality; by
	\cref{eq:shape_constraint_cc,eq:HS_definition}, they impose
	\begin{equation}
		\mathbf H_q^i=S_q\mathbf s_i,
		\qquad q=\alpha+\beta .
		\label{eq:dual_q_central}
	\end{equation}
	Consequently, a generic radial fracton fixed point has the same
	\emph{scale law} as a zero-energy Newtonian solution but cannot be
	identified with a solution of the full Hamiltonian
	\eqref{eq:kappa_newtonian_hamiltonian}.  Apart from the degenerate free case \(\kappa=0\), such an identification is automatic only when the two central-configuration problems are the same,
	\begin{equation}
		q=\kappa .
		\label{eq:q_equals_kappa}
	\end{equation}
	A particular highly symmetric shape may happen to satisfy both
	conditions even when their exponents differ; \cref{eq:q_equals_kappa}
	is the condition for equality of the fixed-point problems themselves.

	Finally impose the Einstein-de Sitter condition
	\(\alpha=-2\beta\).  Equations
	\eqref{eq:fracton_newtonian_identification} and
	\eqref{eq:dual_q_central} then give
	\begin{equation}
		\kappa=-1,\qquad q=-\beta .
	\end{equation}
	The matching condition \eqref{eq:q_equals_kappa} therefore requires
	\(\beta=1\) and \(\alpha=-2\).  Among the models with
	\(\alpha=-2\beta\), the distinguished model is the unique one whose
	radial fixed-point equation automatically produces the full
	zero-energy Newtonian trajectory, rather than only its scale evolution.

	At \((\alpha,\beta)=(-2,1)\),
	\cref{eq:kappa_newtonian_hamiltonian} becomes
	\begin{equation}
		H^{\mathrm{dual}}
		=\frac12\sum_i|\boldsymbol\pi_i|^2
		-G_{\mathrm{eff}}\sum_{i<j}
		\frac{1}{|\mathbf x_i-\mathbf x_j|},
		\qquad
		\boldsymbol\pi_i=\dot{\mathbf x}_i ,
		\label{eq:full_dual_hamiltonian}
	\end{equation}
	where \cref{eq:fracton_newtonian_identification,eq:FA_def} give an expression for the dual Newtonian gravitational constant,
	\begin{equation}
		G_{\mathrm{eff}}
		=\frac{2A_*^2}{S_{-1}}
		=\frac{2S_{-1}}{|\gamma|^6},
		\qquad
		A_*=\gamma|\gamma|^{-4}S_{-1}.
		\label{eq:full_dual_Geff}
	\end{equation}
	Using \(\dot R=-2A_*R^{-1/2}\), the Newtonian energy along the
	homothetic {trajectory} is
	\begin{equation}
		E_N
		=\frac12\dot R^2
		-\frac{G_{\mathrm{eff}}S_{-1}}{R}
		=0 .
		\label{eq:full_dual_zero_energy}
	\end{equation}
	Both the critical energy and \(G_{\mathrm{eff}}\) are therefore fixed by
	the fracton solution rather than chosen as independent Newtonian data.
	The equivalence applies to these homothetic fixed-point trajectories; it
	is not a canonical equivalence between the complete fracton and
	Newtonian Hamiltonian systems.

    \cref{thm:distribution,thm:stability,conj:Cq} and the dual Newtonian description above taken together tell us that for a range of initial conditions within the basin of attraction, the dynamics converges asymptotically to a critical Newtonian one giving us the Newtonian version of flatness discussed in \cref{sec:newtonian_fractons} without fine-tuning initial data. 

	\section{Large-Scale Analysis of the Distinguished Model}
    \label{sec:distinguished_model}
	\subsection{Why the Distinguished Model is Special}
	The EdS branch of \cref{eq:EdS_branch} already has two key Newtonian
	features: by \cref{eq:two_particle_alphabeta}, the two-particle effective
	potential is \(U_{\rm eff}\propto-1/|x_1-x_2|\), and radial fixed points
	have the scale law \(R(t)\sim t^{2/3}\).  These facts alone do not select a
	unique many-body model.  The remaining selection comes from the shape
	distribution and the fixed-point equation.  Along the EdS branch,
	\(q=\alpha+\beta=-\beta\).  \Cref{thm:distribution} shows that for the \(q=1\)
	model with which we started, the regular extremum approaches a shell, whereas for \(q=-1\) it {converges to} a homogeneous ball; the same value \(q=-1\) is where the geometric fixed-point
	equation acquires the Newtonian inverse-square form.  Thus \(q=-1\) on the
	EdS branch uniquely picks out
	
	\begin{equation}
		(\alpha,\beta)=(-2,1).
	\end{equation}
	For this choice, the shape fixed-point condition \eqref{eq:shape_constraint_cc} becomes
	\begin{equation}
		\mathbf H_{-1}^i
		= S_{-1}\,\mathbf{s}_i,\qquad \mathbf H_{-1}^i=\sum_{j\neq i}\frac{\mathbf{s}_i-\mathbf{s}_j}{|\mathbf{s}_i-\mathbf{s}_j|^3},\quad
		S_{-1}:=\sum_{j<k}\frac{1}{|\mathbf{s}_j-\mathbf{s}_k|},\quad i=1,\ldots,N .
		\label{central_config}
	\end{equation}
	\Cref{central_config} is the equal-mass Newtonian central-configuration
	equation, up to the conventional sign choice for gravitational acceleration
\cite{Battye_2003}.  As shown in \cref{sec:full_newtonian_dual}, it is
	also the only model on the EdS branch whose radial fixed-point dynamics
	can be lifted automatically to a full pairwise Newtonian dual at exactly
	zero Newtonian energy.  The distinguished model therefore aligns all the
	Newtonian ingredients reviewed in \cref{sec:newtonian_fractons}: the
	two-particle interaction has the gravitational \(1/r\) form, the homothetic
	scale law is Einstein-de Sitter, the many-body fixed-point equation is the
	Newtonian central-configuration equation, and the regular large-\(N\)
	minimum  approximates a homogeneous ball of radius
	\(R_{-1,N}=\sqrt{5/(3N)}\).  In discrete Newtonian cosmology, such
	homothetic gravitating solutions satisfy the dust Friedmann equations, with
	the Einstein-de Sitter case corresponding to the critical Newtonian energy
	\(E_N=0\) \cite{Ellis_2013,Ellis_2015}.  The remaining finite-\(N\)
	question, anticipated by \cref{thm:stability,conj:Cq}, is whether this
	homogeneous-ball fixed point is dynamically selected.	
	
	\label{sec:EdS_model}
	
	\subsection{Lyapunov Stability Analysis}
    \begin{figure}[!ht]
		\centering
		\includegraphics[width=\textwidth]{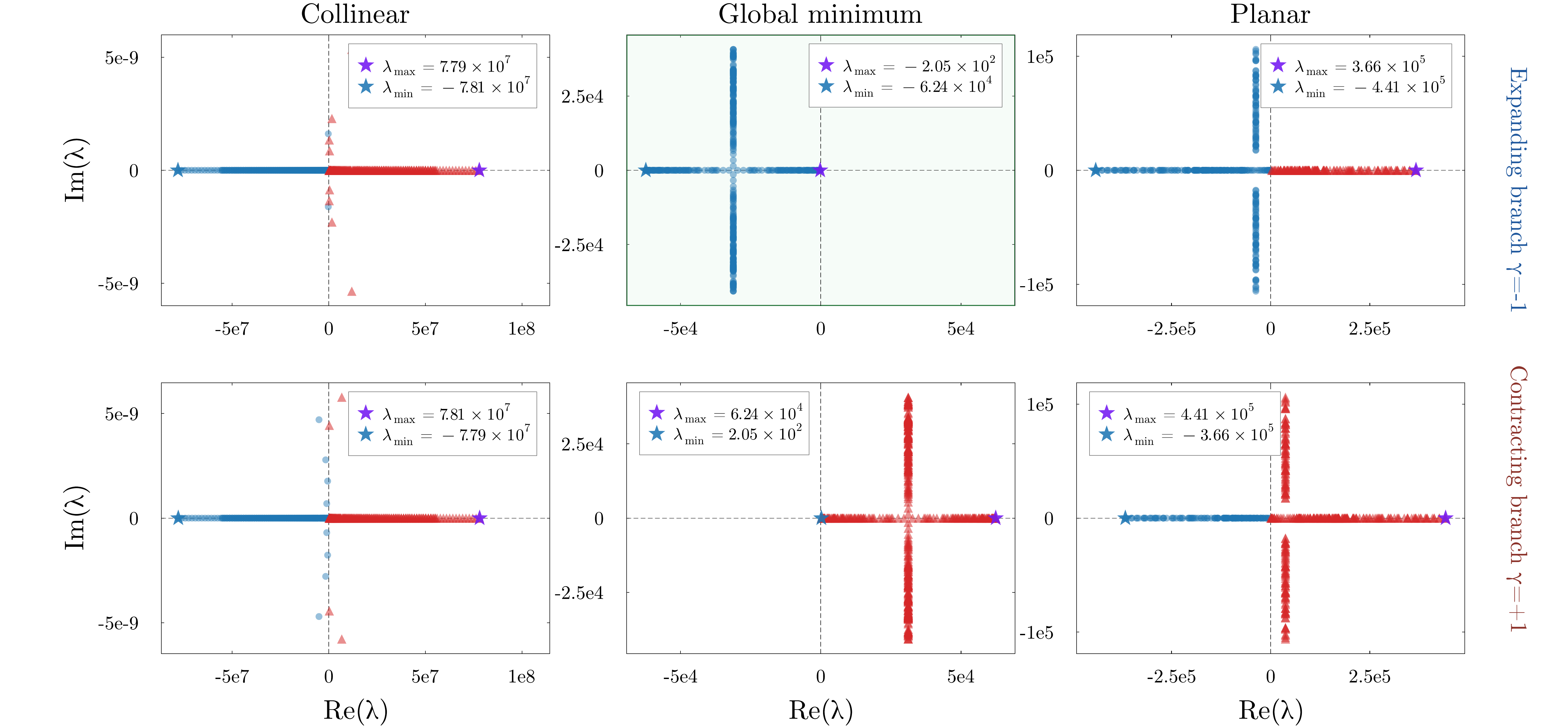}
		\caption{Eigenvalue spectra of the reduced Jacobian for \(N=75\) fixed points corresponding to three representative central configurations: a collinear saddle, the global minimum, and a planar saddle. The top row corresponds to the expanding branch, while the bottom row shows the corresponding contracting branch. The minimum and maximum real parts of the eigenvalue spectrum are indicated in each panel. Stable and unstable eigenvalues are denoted by circles and triangles, respectively. \label{fig:Eig_spec}}
	\end{figure}
	We now test the finite-\(N\) consequences of \cref{thm:stability}.  For \(q=-1\), the regular ball configurations are minima of \(\Xi_{-1}\), so the theorem predicts stability of the expanding minimum branch if \(C_{-1}\) is positive on the physical perturbations, and instability for saddle-type central configurations.  For \((\alpha,\beta)=(-2,1)\), \cref{eq:shape_coordinates_eom} defines an autonomous \(6N\)-dimensional system before constraints and symmetry directions are removed.  For a generic non-collinear fixed point, the translation constraints in \(\mathbf{s}\) and \(\mathbf{u}\), the unit-shape constraint, a \(\gamma\)-family zero mode, and the common-rotation zero modes remove eleven directions, leaving a reduced stability problem of dimension \(6N-11\).  The same projection procedure is used for the collinear example, with any linearly dependent generators discarded during orthonormalization.  The equilateral \(N=3\) case was handled analytically in \cref{prop:equilateral_stability}; for larger \(N\) we compute the reduced Jacobian numerically, as summarized in \cref{app:NumericalImplementation}.
	
	The numerical examples below use representative \(N=75\) central configurations: the estimated global minimum configuration,\footnote{The energy landscape of the gravitational shape potential \(\Xi_{-1}\) is highly rugged, with numerical optimization routines typically converging to one of many local minima. Consequently, certifying that a given configuration is the absolute global minimum, rather than merely a local minimum, for \(N=75\) is a nontrivial computational task. However, the stability results presented here have been verified across multiple distinct local minima obtained from independent optimization runs. This is expected, since \cref{thm:stability} requires only that the constrained Hessian satisfy \(Q_q<0\), a condition that is fulfilled by every local minimum.} together with two saddle configurations: a collinear configuration and a planar configuration. The global minimum estimate is the maximally spread configuration; for \(N=75\) it exhibits a multilayer shell structure, while in the large-\(N\) limit it converges to a homogeneous ball (\cref{thm:distribution}). By contrast, the two saddle configurations are representative of highly inhomogeneous particle distributions. These configurations are the finite-\(N\) analogues of the cases distinguished in \cref{thm:distribution,thm:stability}.   The parameter \(\gamma\) in \eqref{eq:fixed_ansatz} distinguishes expanding and contracting radial fixed points.  Indeed, for \((\alpha,\beta)=(-2,1)\),
	\begin{equation}
		\frac{dR}{dt}=-2A_{(-2,1)} R^{-1/2},
		\label{eq:scale_eqn}
	\end{equation}
	while \eqref{eq:FA_def} gives
	\begin{equation}
		A_{(-2,1)}=\gamma|\gamma|^{-4}S_{-1},\qquad
		S_{-1}=\sum_i\mathbf{s}_i\cdot\mathbf H_i
		=\sum_{i<j}\frac{1}{|\mathbf{s}_i-\mathbf{s}_j|}>0 .
	\end{equation}
	Hence \(\gamma<0\) gives the expanding branch and \(\gamma>0\) gives the contracting branch.
	
	\Cref{fig:Eig_spec} shows the eigenvalue spectrum of the reduced Jacobian for these three classes of central configurations and for both signs of \(\gamma\). The expanding global minimum has all eigenvalues confined to the left half-plane, with no eigenvalue possessing a positive real part, and is therefore asymptotically stable in the reduced dynamics. Its contracting counterpart has the reflected spectrum and is correspondingly unstable. In contrast, saddle-type central configurations possess unstable directions even on the expanding branch. These numerical results are in complete agreement with the stability predictions of \cref{thm:stability,conj:Cq}, namely that regular extrema of \(\Xi_{-1}\) are stable on the expanding branch, whereas saddle configurations are unstable. Thus, we see that the stability analysis singles out the expanding minimum branch, indicating a dynamical preference for homogeneous fixed points over inhomogeneous ones in the distinguished \((\alpha,\beta)=(-2,1)\) model. We revisit this dynamical selection of homogeneity in the following section, where we investigate the large-\(N\) evolution.
    
    Since linear stability establishes only a local basin of attraction, we first test numerically whether these locally stable fixed points indeed emerge dynamically from randomly generated initial conditions.
	
	\subsection{Dynamical Convergence Tests for Small \texorpdfstring{$N$}{N}}
	We now explore the basin of attraction for the dynamics of a small number of particles ($N<100$). We present our results for $N=25$.

	\subsubsection*{Convergence to Newtonian Central Configurations}
    	\begin{figure}[htbp]
		\centering
		\includegraphics[width=\textwidth]{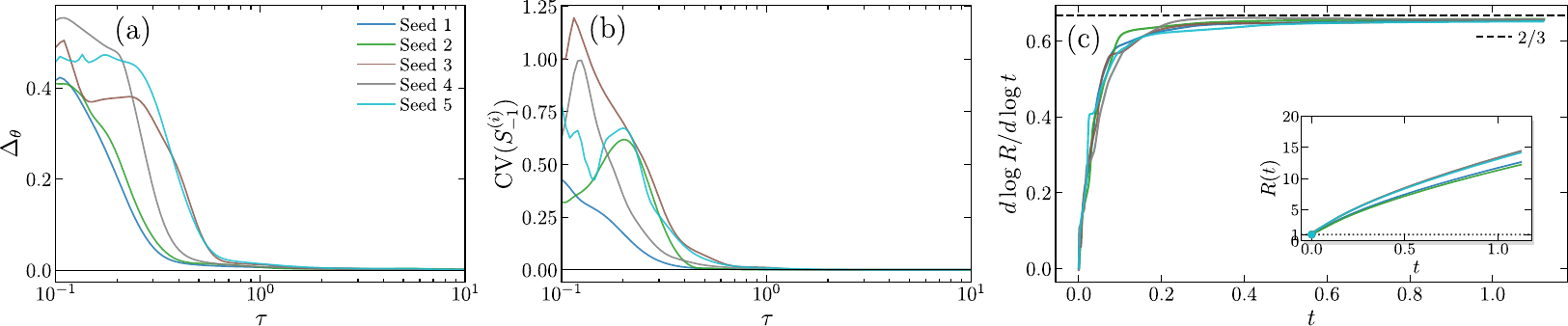}
		\caption{Evolution of $N=25$ particles with shape phase-space $\{\mathbf{s}_i,\mathbf{u}_i\}$ initializations selected from a Gaussian random number generator with mean \(0\) and variance \(1\) but then re-centered and re-scaled to satisfy constraints \cref{eq:shape_coordinate_constraints}. Panels (a) and (b) show angular misalignment and the variation of $S_{-1}^{(i)}$ that diagnose convergence to the $q=-1$ Newtonian central configurations. Panel (c) shows the corresponding scale evolution with \(R(0)\) set to \(1\), the logarithmic slope \(d\log R/d\log t\) approaches \(2/3\) at late times. The inset shows the scale \(R(t)\) settling into expansion for each seed, hence showing that the expanding branch is dynamically selected.}
		\label{fig:min_converge}
	\end{figure}
 We start by verifying if the shape evolution dynamically selects Newtonian central configurations. We use two dimensionless diagnostics to check for the departure from \eqref{central_config}.  The angular misalignment is
	\begin{equation}
		\Delta_\theta = \sqrt{\frac{1}{N}\sum_i \sin^2\theta_i},\qquad
		\sin\theta_i=\frac{|\mathbf H_{-1}^i\times\mathbf{s}_i|}{|\mathbf H_{-1}^i||\mathbf{s}_i|},
	\end{equation}
	and the non-uniformity of the particlewise estimates of \(S_{-1}\) is measured by
	\begin{equation}
		\mathrm{CV}(S_{-1}^{(i)})=\frac{\sigma_{S_{-1}}}{|\overline S_{-1}|},\qquad
		S_{-1}^{(i)}=\frac{\mathbf H_{-1}^i\cdot\mathbf{s}_i}{|\mathbf{s}_i|^2},\qquad
		\overline S_{-1}=\frac{1}{N}\sum_i S_{-1}^{(i)} .
	\end{equation}
	\cref{central_config} is satisfied precisely when both diagnostics vanish.  \Cref{fig:min_converge}(a) and (b) show that random \(N=25\) initial conditions approach this condition within numerical accuracy. Hence, it seems that, at least for the moderate particle count in \cref{fig:min_converge}, the basin of attraction is broad enough to include phase-space initial conditions sampled from a normal distribution.
    Next, we verify whether the system dynamically chooses the expanding branch as is anticipated from \cref{thm:stability} and numerical linear stability results.  
	
	\subsubsection*{Convergence of the Scale Evolution}
	We now reconstruct back the full Hamiltonian flow from reduced shape dynamics and study the evolution of scale \(R(t)\).  \Cref{fig:min_converge}(c) shows that the random seeds studied here evolve onto expanding solutions, as expected from the stability of the \(\gamma<0\) near-minimum fixed point.  It also shows the logarithmic slope \(d\log R/d\log t\) converges to the Einstein-de Sitter value \(2/3\) at late times~\footnote{This convergence plot also has an effect of an offset that is due to not starting at \(R=0\) at \(t=0\). The effect of this offset fades away at late times. The point of this plot is to show convergence indeed happens to EdS scale evolution for various random seeds, and not the exact nature of convergence which is muddled due to offset.}.

    \subsubsection*{Convergence to Global Minimum of the Shape Potential}
    	\begin{figure}[!htbp]
		\centering
		\includegraphics[width=0.5\textwidth]{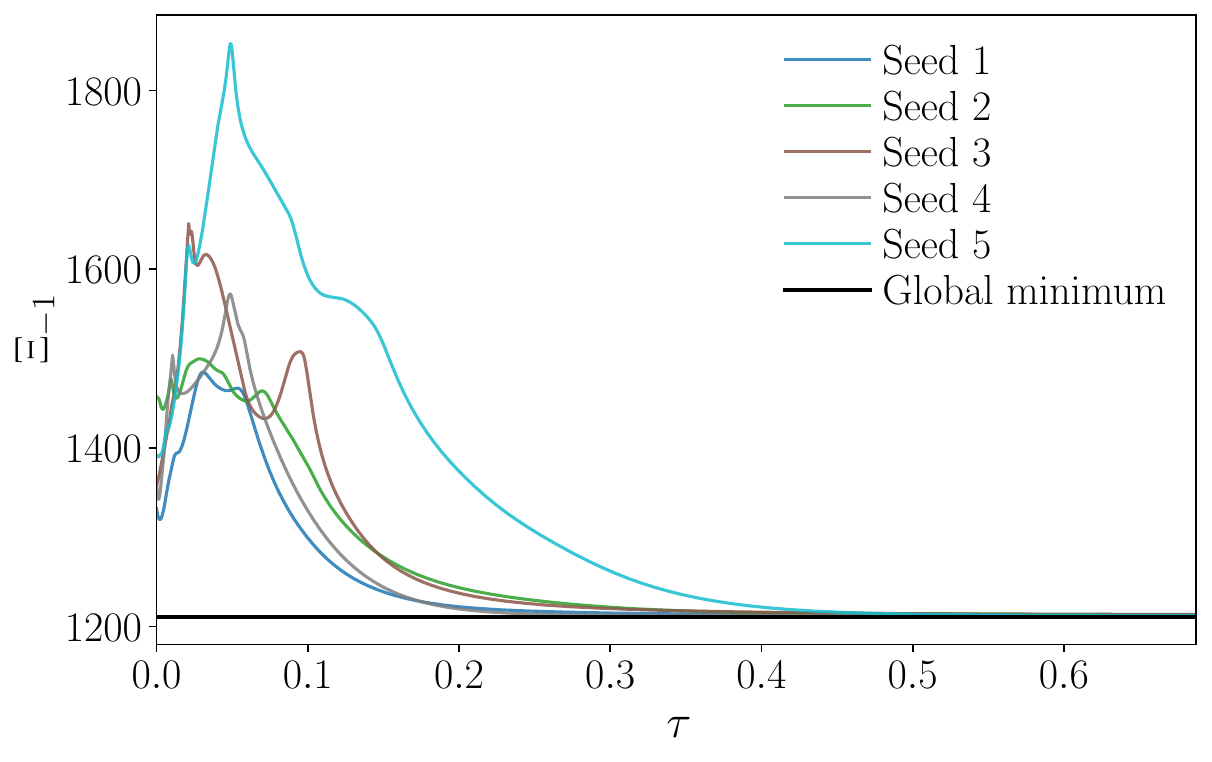}
		\caption{Evolution of shape potential $\Xi_{-1}$ from the same random Gaussian initializations (\(N=25\)) as in \cref{fig:min_converge} showing convergence toward the numerically determined estimate of the global minimum.}
		\label{fig:potential_converge}
	\end{figure}
    
	The global minimum of \(\Xi_{-1}\), the gravitational shape potential for the distinguished model, can be estimated numerically using L-BFGS-B optimization~\cite{LMBFGSB_1995SJSC...16.1190B}.  We then initialize the reduced shape system from random initial data and track \(\Xi_{-1}\) as a function of \(\tau\).  \Cref{fig:potential_converge} shows that, for the system sizes displayed, the trajectories converge close to the ``best-found'' minimum value. This is consistent with \cref{thm:stability} and our preceding numerical stability analysis, which classifies global minimum central configurations as local attractors. 

	These numerical results support the central claim for the moderate \(N\) systems studied here: randomly generated initial data are attracted to the expanding global minimum shape potential fixed point, and the reconstructed scale approaches \(R(t)\propto t^{2/3}\) at late-times. We will next see that a new class of attractor fixed points emerges when increasing the number of particles.
	
	\subsection{Dynamical Convergence Tests for Large \texorpdfstring{$N$}{N} and a Global Conjecture}
	\begin{figure}[!htbp]
		\centering
		\includegraphics[width=\textwidth]{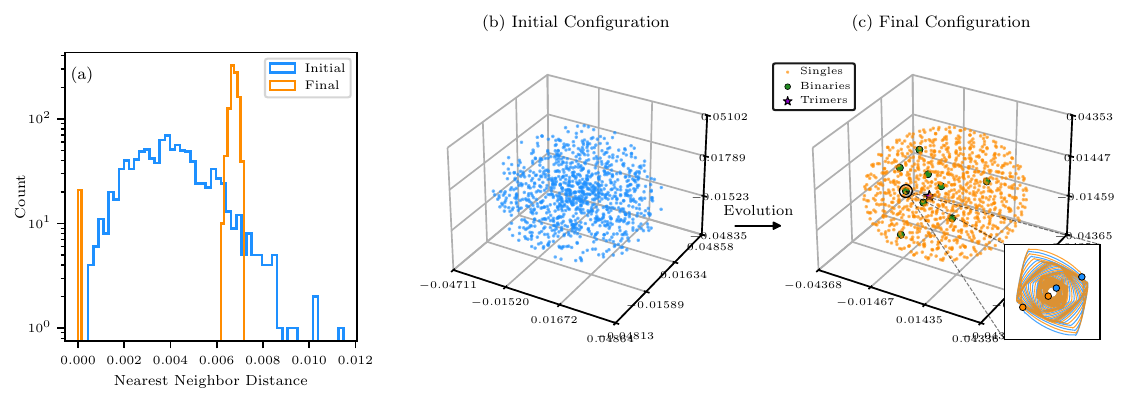}
		\caption{Time evolution of reduced shape dynamics from randomly generated initial conditions for \(N=1000\). (a) Nearest neighbor distance distributions for initial and final states. The time-evolved configuration has two widely separated peaks corresponding to intra- and inter-cluster length scales. (b) An initialization $\{\mathbf{s}_i,\mathbf{u}_i\}$ selected from a Gaussian random number generator with mean \(0\) and variance \(1\) but then re-centered and re-scaled to satisfy constraints \cref{eq:shape_coordinate_constraints}. (c) Late-time configuration with several tightly bound clusters; the inset shows a binary's inspiraling orbital motion. The binary size shrinks, and the ratio of intra- to inter-cluster length scales tends to zero with time. }
		\label{fig:Time_evol_1000_t}
	\end{figure}

	The small-\(N\) simulations indicate convergence to the attractor point corresponding to radial expansion.  For larger \(N\), we discover a new class of attractor fixed points. The reduced dynamics remains dissipative-like in shape space, but instead of all particles freezing independently, a small number of particles form tightly bound clusters with internal motion.  The cluster centers, together with the unclustered particles, then become approximately stationary in shape space.\footnote{It is not that the fixed points identified in \cref{sec:scale_invt_fractons} and \cref{sec:EdS_model} cease to exist. They do, and the expanding minimum branch is still a local attractor. The basin of attraction of this class of fixed points is not large enough to accommodate generic initial phase space data.} Moreover, the clusters monotonically shrink in shape space, and the ratio of the intra-cluster to inter-cluster length scales tends to zero with time. A typical \(N=1000\) evolution from randomly generated initial conditions, illustrating convergence to this new class of fixed points, is shown in \cref{fig:Time_evol_1000_t}.  
		
	We will now take $N>100$ and subject the evolution to the same convergence tests. As shown in \cref{fig:cc_largeN} we find that for all random initializations and particle counts $N>100$, the late-time state settles to a \emph{renormalized Newtonian central configuration} if a cluster of \(m_a\) particles is renormalized as a single point of mass \(m_a\) located at the cluster center $\bar{\mathbf s}_a$,
	
	\begin{equation}
		\widetilde{\mathbf H}_{-1}^{\,a}
		= \widetilde S_{-1}\,\bar{\mathbf{s}}_a,\qquad
		\widetilde{\mathbf H}_{-1}^{\,a}
		:=\sum_{b\neq a}m_b
		\frac{\bar{\mathbf s}_a-\bar{\mathbf s}_b}
		{|\bar{\mathbf s}_a-\bar{\mathbf s}_b|^3},
		\qquad
		\widetilde S_{-1}
		:=\sum_{a<b}\frac{m_am_b}
		{|\bar{\mathbf s}_a-\bar{\mathbf s}_b|}.
		\label{mass_central_config}
	\end{equation}
	instead of the equal-mass version in \cref{central_config}.  The same diagnostics as used for small \(N\) provide evidence for this coarse-grained convergence\footnote{Clusters
	are identified from the late-time state and the same membership is then used to construct the coarse-grained diagnostics at earlier times.} in \cref{fig:cc_largeN}. 
	\begin{figure}[htbp]
		\centering
		\includegraphics[width=\textwidth]{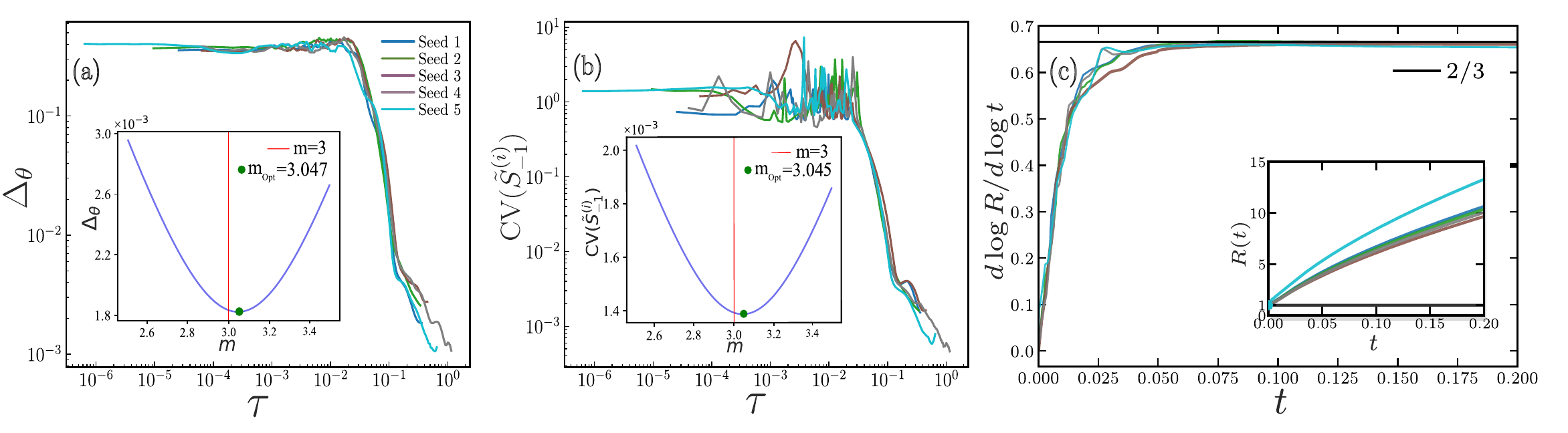} 
		\caption{Evolution of $N=200$ particles from Gaussian shape phase-space $\{\mathbf{s}_i,\mathbf{u}_i\}$ initializations with mean 0 and variance 1, but then re-centered and re-scaled to satisfy \cref{eq:shape_coordinate_constraints}. (a)--(b) Angular misalignment and the variation of $\tilde{S}_{-1}^{(i)}$ that diagnose convergence to Newtonian central configurations of \cref{mass_central_config}, the insets show the same diagnostics for a late-time state when the mass of a trimer cluster is varied \((2.5<m<3.5)\). It is clear that the late-time renormalized configuration is the optimal central configuration if a trimer is treated as a mass \(3\) point particle. Similar renormalization is found to work for clusters of all sizes. (c) Scale evolution (\(R(0)\) set to \(1\)) and the logarithmic slope \(d\log R/d\log t\), which approaches \(2/3\) at late times.}
		\label{fig:cc_largeN}
	\end{figure}
	The numerical evidence suggests a statement about the coarse-grained
	dynamics, rather than only about its fixed points.  We formulate this
	stronger statement as our second conjecture.

	\begin{conjecture}[Effective unequal-mass cluster dynamics]
		\label{conj:renormalized_dynamics}
		For the distinguished model \((\alpha,\beta)=(-2,1)\), every
		late-time attractor state of the reduced shape dynamics admits an asymptotic
		decomposition into tightly bound clusters \(C_a\), with unclustered
		particles regarded as singleton clusters, such that the ratio of the
		intra-cluster to inter-cluster length scales tends to zero.  Once this
		separation of scales has developed, the cluster centers and cluster momenta, defined as
		\begin{equation}
			\bar{\mathbf x}_a
			=\frac{1}{m_a}\sum_{i\in C_a}\mathbf x_i,\qquad
			\mathbf P_a=\sum_{i\in C_a}\mathbf p_i,\qquad
			m_a=|C_a|,
		\end{equation}
		evolve, to leading order, as an otherwise unclustered unequal-mass
		 system governed by
		\begin{equation}
			H^{(m)}_{(-2,1)}
			=\sum_{a<b}m_am_b
			\frac{|\bar{\mathbf x}_a-\bar{\mathbf x}_b|}
			{|\mathbf P_a/m_a-\mathbf P_b/m_b|^2}.
			\label{eq:renormalized_hamiltonian}
		\end{equation}
		In the centered scale-shape variables defined below, this effective
		motion approaches an expanding radial fixed point.  Corrections to the
		postulated dynamics vanish with the ratio of the intra-cluster to
		inter-cluster length scales.
	\end{conjecture}

	The corresponding scale-shape separation is defined by
	\begin{align}
		M&:=\sum_a m_a,\qquad
		\mathbf X_{\rm C}^{(m)}
		:=\frac{1}{M}\sum_a m_a\bar{\mathbf x}_a,\qquad
		\boldsymbol\pi_{\rm C}^{(m)}
		:=\frac{1}{M}\sum_a\mathbf P_a,
		\notag\\
		R_m^2&:=\sum_a m_a
		|\bar{\mathbf x}_a-\mathbf X_{\rm C}^{(m)}|^2,
		\notag\\
		\bar{\mathbf s}_a
		&:=\frac{\bar{\mathbf x}_a-\mathbf X_{\rm C}^{(m)}}{R_m},
		\qquad
		\bar{\mathbf u}_a
		:=R_m^{-1/2}
		\left(\frac{\mathbf P_a}{m_a}
		-\boldsymbol\pi_{\rm C}^{(m)}\right),
		\qquad
		\dd\tau:=\frac{\dd t}{R_m^{3/2}} .
		\label{eq:renormalized_shape_variables}
	\end{align}
	These variables satisfy
	\begin{equation}
		\sum_a m_a\bar{\mathbf s}_a=0,\qquad
		\sum_a m_a|\bar{\mathbf s}_a|^2=1,\qquad
		\sum_a m_a\bar{\mathbf u}_a=0.
		\label{eq:renormalized_shape_constraints}
	\end{equation}
	In the limit of well-separated tight clusters, they agree, up to
	finite-size corrections, with the corresponding microscopic shape
	averages:
	\begin{equation}
		\bar{\mathbf s}_a
		\simeq\frac{1}{m_a}\sum_{i\in C_a}\mathbf s_i,\qquad
		\bar{\mathbf u}_a
		\simeq\frac{1}{m_a}\sum_{i\in C_a}\mathbf u_i .
	\end{equation}
	In these centered, scale-normalized variables, the Hamiltonian and thus energy become independent of the overall scale $R$,
	\begin{equation}
		H^{0,(m)}_{(-2,1)}
		=\sum_{a<b}m_am_b
		\frac{|\bar{\mathbf s}_a-\bar{\mathbf s}_b|}
		{|\bar{\mathbf u}_a-\bar{\mathbf u}_b|^2}.
		\label{eq:renormalized_shape_hamiltonian}
	\end{equation}
	The function \(H^{0,(m)}_{(-2,1)}\) is conserved along the full motion, although it is not itself the generator of the reduced \((\bar{\mathbf{s}}_a,\bar{\mathbf{u}}_a)\) dynamics.
	
	We can verify directly that the renormalized central configurations are
	radial fixed points of this effective dynamics.  Writing
	\(\bar{\mathbf s}_{ab}:=\bar{\mathbf s}_a-\bar{\mathbf s}_b\) and
	\(\bar{\mathbf u}_{ab}:=\bar{\mathbf u}_a-\bar{\mathbf u}_b\), the
	reduced equations obtained from \cref{eq:renormalized_hamiltonian} are
	\begin{align}
		\frac{\dd\bar{\mathbf s}_a}{\dd\tau}
		&=-2\left(\mathbf F_a^{(m)}-A_m\bar{\mathbf s}_a\right),
		&
		\frac{\dd\bar{\mathbf u}_a}{\dd\tau}
		&=-\left(\mathbf G_a^{(m)}-A_m\bar{\mathbf u}_a\right),
		\label{eq:renormalized_shape_eom}\\
		\mathbf F_a^{(m)}
		&:=\sum_{b\neq a}m_b
		|\bar{\mathbf s}_{ab}|
		|\bar{\mathbf u}_{ab}|^{-4}\bar{\mathbf u}_{ab},
		&
		\mathbf G_a^{(m)}
		&:=\sum_{b\neq a}m_b
		|\bar{\mathbf u}_{ab}|^{-2}
		|\bar{\mathbf s}_{ab}|^{-1}\bar{\mathbf s}_{ab},
		\notag\\
		A_m&:=\sum_a m_a\bar{\mathbf s}_a\cdot\mathbf F_a^{(m)} .
		\notag
	\end{align}
	The scale and physical time evolve simultaneously according to
	\begin{equation}
		\frac{\dd R_m}{\dd\tau}=-2A_mR_m,
		\qquad
		\frac{\dd t}{\dd\tau}=R_m^{3/2}.
		\label{eq:renormalized_scale_eom}
	\end{equation}
	On the radial ansatz
	\(\bar{\mathbf u}_a=\gamma\bar{\mathbf s}_a\), with
	\(c:=\gamma|\gamma|^{-4}\), these quantities reduce to
	\begin{equation}
		\mathbf F_a^{(m)}
		=c\,\widetilde{\mathbf H}_{-1}^{\,a},\qquad
		\mathbf G_a^{(m)}
		=\gamma\mathbf F_a^{(m)},\qquad
		A_m=c\,\widetilde S_{-1},
		\label{eq:renormalized_radial_reduction}
	\end{equation}
	where $\widetilde{\mathbf H}_{-1}^{\,a}$ and $\widetilde S_{-1}$ are as defined in \cref{mass_central_config}.	It follows that both right-hand sides of
	\cref{eq:renormalized_shape_eom} vanish whenever
	\(\widetilde{\mathbf H}_{-1}^{\,a}
	=\widetilde S_{-1}\bar{\mathbf s}_a\).  Conversely, the first fixed-point
	equation on the radial ansatz gives
	\(\widetilde{\mathbf H}_{-1}^{\,a}
	=(A_m/c)\bar{\mathbf s}_a\); taking its mass-weighted scalar product
	with \(\bar{\mathbf s}\) and using
	\cref{eq:renormalized_shape_constraints} yields
	\(A_m/c=\widetilde S_{-1}\).  Thus the radial fixed-point condition is
	equivalent to the renormalized central-configuration equation
	\eqref{mass_central_config}.
	
	\begin{corollary}[Renormalized central configurations]
		\label{cor:renormalized_cc}
		Under \cref{conj:renormalized_dynamics}, replace each cluster \(C_a\)
		by a point of mass \(m_a=|C_a|\) at its center.  After centering and
		normalizing the resulting weighted configuration, the renormalized points
		\(\{\bar{\mathbf s}_a\}\) satisfy the Newtonian central-configuration
		equation \eqref{mass_central_config}, up to corrections that vanish as
		the ratio of the intra-cluster to inter-cluster length scales tends to
		zero.  The ordinary equal-mass radial fixed points are the cluster-free
		special case \(m_a=1\) for every \(a\).
	\end{corollary}
	
    \begin{corollary}[Critical unequal-mass Newtonian dual]
		\label{cor:renormalized_newtonian_dual}
		Under \cref{conj:renormalized_dynamics}, the asymptotic radial
		fixed-point motion of the cluster centers is, to leading order, dual to
		a zero-energy Newtonian gravitational system with masses
		\(m_a=|C_a|\).  Its Hamiltonian is
		\begin{equation}
			H^{\mathrm{dual}}
			=\sum_a\frac{|\boldsymbol\Pi_a|^2}{2m_a}
			-G_{\mathrm{eff}}^{(m)}
			\sum_{a<b}\frac{m_am_b}
			{|\bar{\mathbf x}_a-\bar{\mathbf x}_b|},
			\qquad
			\boldsymbol\Pi_a:=m_a\dot{\bar{\mathbf x}}_a ,
			\label{eq:renormalized_newtonian_dual}
		\end{equation}
		where \(\boldsymbol\Pi_a\) is the dual Newtonian momentum, not
		the fracton cluster momentum \(\mathbf P_a\), and
		\begin{equation}
			G_{\mathrm{eff}}^{(m)}
			=\frac{2A_*^2}{\widetilde S_{-1}}
			=\frac{2\widetilde S_{-1}}{|\gamma|^6},
			\qquad
			A_*=\gamma|\gamma|^{-4}\widetilde S_{-1}.
			\label{eq:renormalized_Geff}
		\end{equation}
		Along the homothetic solution
		\(\bar{\mathbf x}_a=R_m\bar{\mathbf s}_a\),  {the Newtonian energy vanishes}:
		\begin{equation}
			E_{N}
			=\frac12\dot R_m^2
			-\frac{G_{\mathrm{eff}}^{(m)}
			\widetilde S_{-1}}{R_m}
			=0 .
			\label{eq:renormalized_dual_zero_energy}
		\end{equation}
	\end{corollary}
	\begin{proof}
		By \cref{cor:renormalized_cc}, the limiting weighted shape satisfies
		the unequal-mass Newtonian central-configuration equation
		\eqref{mass_central_config}. { Hamilton's equations generated by
		\cref{eq:renormalized_newtonian_dual} allow homothetic solutions that preserve this shape, and the dynamics therefore reduces to a single degree of freedom, the scale \(R_m\): }
		\[
			\ddot R_m
			=-\frac{G_{\mathrm{eff}}^{(m)}
			\widetilde S_{-1}}{R_m^2}.
		\]
		On the other hand,
		\cref{eq:renormalized_eds_scale} with
		\(A_m\to A_*\) gives
		\(\dot R_m=-2A_*R_m^{-1/2}\) and hence
		\(\ddot R_m=-2A_*^2R_m^{-2}\).  The two equations agree for
		\cref{eq:renormalized_Geff}; substituting the same relation into the
		Newtonian energy gives \cref{eq:renormalized_dual_zero_energy}.
		The corrections are those already specified in
		\cref{conj:renormalized_dynamics}.
	\end{proof}

	Finite-size effects and the internal motion of each cluster enter through
	the corrections described in \cref{conj:renormalized_dynamics}.  The
	conjecture concerns the coarse-grained cluster centers; a clustered
	attractor may therefore retain nontrivial internal motion and need not be
	a pointwise fixed point of the microscopic reduced variables.  The
	unequal-mass fixed-point construction and the corresponding conditional
	stability criterion are given in
	\cref{app:UnequalMassCentralConfigs,app:UnequalMassStability},
	respectively.
    
    The same reduced equations also determine the overall expansion.  On the
	expanding radial branch \(\gamma<0\), so
	\(A_m\to A_* =c\widetilde S_{-1}<0\) under
	\cref{conj:renormalized_dynamics}.  Transforming
	\cref{eq:renormalized_scale_eom} back to physical time gives
	\begin{equation}
		\frac{\dd R_m}{\dd t}=-2A_mR_m^{-1/2},
		\qquad
		\frac{\dd}{\dd t}R_m^{3/2}=-3A_m.
		\label{eq:renormalized_eds_scale}
	\end{equation}
	Consequently, \(A_m\to A_*<0\) implies
	\begin{equation}
		R_m(t)^{3/2}=-3A_*t+o(t),
		\qquad
		R_m(t)\sim(-3A_*t)^{2/3}\propto t^{2/3} ,
	\end{equation}
	which is the Einstein-de Sitter scaling.  Moreover, the microscopic and
	cluster-center scales obey the exact decomposition
	\begin{equation}
		R^2=R_m^2+
		\sum_a\sum_{i\in C_a}|\mathbf x_i-\bar{\mathbf x}_a|^2.
		\label{eq:microscopic_cluster_scale_decomposition}
	\end{equation}
	
    Thus \(R/R_m\to1\) when the intra-cluster scale is negligible compared
	with the inter-cluster scale, and the full microscopic scale inherits the
	same \(t^{2/3}\) law. As shown in \cref{fig:cc_largeN}(c) randomly generated initial conditions evolve onto expanding solutions and the logarithmic slope \(d\log R/d\log t\) approaches \(2/3\) at late times, which is consistent with \cref{conj:renormalized_dynamics} in view of \cref{eq:renormalized_scale_eom,eq:renormalized_eds_scale}.

\subsection{Emergence of Homogeneity}
\begin{figure}[!htbp]
	\centering
	\includegraphics[width=\textwidth]{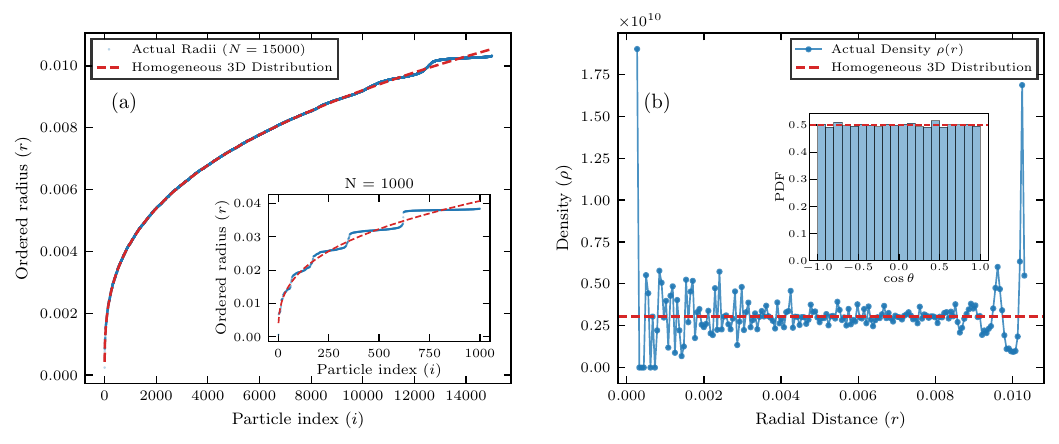}
	\caption{Homogeneity diagnostics for an \(N=15000\) configuration evolved from a Gaussian-sampled random phase-space point .  Panel (a): ordered radii $r_i$ (sorted distances from the centre) versus particle index $i$ for $N=15000$.
		The dashed red line shows the best‑fit homogeneous 3D distribution $r_i \propto i^{1/3}$.
		The inset displays the same quantities for a smaller system with $N=1000$ particles.
		Panel (b): radial number density $\rho(r)$ obtained from spherical shell binning of the $N=15000$ configuration (blue circles with connecting lines).
		The dashed red line indicates the constant density expected for a uniform sphere. The inset displays the probability density function (PDF) for angular distribution, along $\cos\theta$,
        where $\theta$ is the polar angle measured from the $z$-axis. The constant value of $1/2$   (dashed red line) indicates a perfectly isotropic distribution.
		}
	\label{fig:Homogeneity}
\end{figure}

It is clear from \cref{thm:distribution} that global minimum central configurations of \((\alpha,\beta)=(-2,1)\) converge to a homogeneous bulk density as \(N\to\infty\). Moreover, \cref{thm:stability,conj:Cq} indicate these configurations are local attractors. But we now know that the large-$N$ dynamics generically converges to the cluster-forming class of fixed points. We now ask if these late-time states can still exhibit homogeneity at larger length scales.  As shown in \cref{fig:Homogeneity}, cluster-forming trajectories for large-\(N\) still retain a homogeneous large-scale distribution. The reason for this, as explained in ~\cite{Battye_2003}, follows directly from the conjectured global fixed point forming an unequal-mass central configuration shown in \cref{mass_central_config}. It was argued that a particle of mass $m$ evacuates a volume around it proportional to $m$, and the distribution is identical to the homogeneous one resulting from the equal-mass central configuration at large length scales~\footnote{Although \cite{Battye_2003} considers only extremum central configurations, we have independently verified this to be true for our late-time renormalized central configurations.}.  This highlights the mechanism in this model through which small-scale structures can coexist with large-scale bulk homogeneity as in standard cosmology~\cite{Peebles1993,dodelson2020modern}.

\section{Bound Local Dynamics, Mach's Principle and an Arrow of Time}
\label{sec:structure}

In the previous section, we saw that the distinguished model \((\alpha,\beta)=(-2,1)\), in its reduced shape dynamics, has attractors whose characteristic feature is that their spatial arrangement is a Newtonian central configuration. We see that as the number of particles $N$ is increased, the late-time states evolving from randomly generated initial conditions change from equal-mass to cluster-renormalized central configurations. Moreover, we find that the scale evolution follows the Einstein-de Sitter power law where the cluster centers drift apart as \(r_{ab}\propto t^{2/3}\) while having internal motions. At large \(N\) we see the emergence of large-scale homogeneity along with clusters at smaller scales.      

To get a rough sense of how the number and variety of bound structures change with $N$, \Cref{tab:i} reports the cluster count in one representative run for several $N$.  The counts increase across these particular runs, but they are not ensemble averages. We see that as we increase $N$, so do the number and composition of the clusters that form, but these remain a very small fraction of the total particles. We next look at the dynamics of the particles comprising these clusters and highlight some of their noteworthy features. 
\begin{table}[htbp]
	\centering
	\begin{tabular}{cc|c}
		\hline
		$N$ & Cluster Count & Composition\\
		\hline
		100 & 1 & 1 binary\\
		200 & 4 & 4 binaries\\
		500 & 6 & 5 binaries and 1 tetramer\\
		1000 & 11 & 10 binaries and 1 trimer\\
		1500 & 18 & 16 binaries and 2 trimers\\
		\hline
	\end{tabular}
	\caption{Clusters identified at the final recorded time in one representative run for each \(N\). \label{tab:i}}
\end{table}

\subsection{Intra-Cluster Dynamics}
\begin{figure}[!htbp]
	\centering
	\includegraphics[width=0.7\textwidth]{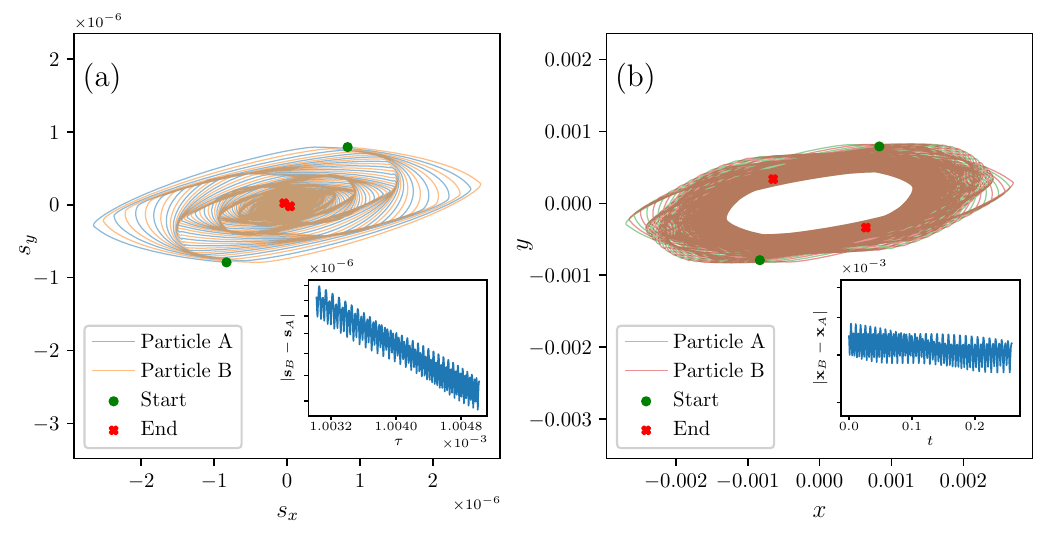}
	\caption{Trajectories in the center-of-mass frame of a binary, shown in (a) shape and (b) physical coordinates.  Insets show the separation, on a logarithmic scale, between the two particles in shape and physical coordinates.  The separation decreases approximately linearly on the logarithmic scale in shape coordinates, while the corresponding drift in physical coordinates is small.}
	\label{fig:Trajectory}
\end{figure}

\Cref{fig:Trajectory} shows a representative binary trajectory in its
	center-of-position frame.  In shape coordinates, its separation shrinks
	approximately exponentially over the displayed reduced-time interval \(\tau\).  The resulting separation between the intra-cluster and inter-cluster scales is consistent with the asymptotic cluster decomposition in \cref{conj:renormalized_dynamics}. On the expanding branch, \eqref{eq:scale_tau} gives
\begin{equation}
	R(\tau)\propto e^{-2A\tau},\qquad A<0.
\end{equation}
Thus, the growth of the physical scale largely compensates for the shrinkage in shape coordinates.  In the run shown, the physical binary size has only a small residual drift~\footnote{We expect this drift to vanish either as we increase the simulation time or $N$, both of which are not accessible within our current numerical analysis.}.  From a relational point of view, such bound subsystems can serve as internal rods and clocks, with orbital separations and periods defining local standards against which the global expansion of larger structures is measured. This behavior parallels gravitationally bound systems which decouple from the expansion of the universe: co-moving separations shrink, while physical sizes remain nearly fixed~\cite{einstein1945influence,carrera2010influence}.

\subsection{Mach's Principle}
The dynamics of cluster-forming late-time states also has a close analog in a familiar regime of Newtonian cosmology. Consider an exact zero-energy expanding homothetic solution of Newtonian \(N\)-body gravity, and replace a few point masses by tightly bound subsystems whose centers of mass initially coincide with the original point particles. At leading multipole order, the centers follow the point-mass motion; finite-size corrections are tidal and are suppressed by powers of the intra- to inter-system length ratio, provided the internal moments remain controlled. The internal motion need not be exactly Keplerian, since the background produces time-dependent tidal perturbations, but these too decay as the intra- to inter-cluster length scale vanishes. Newtonian \(N\)-body gravity therefore has a dynamical regime in which bound subsystems remain localized while their centers of mass form an effective central configuration expanding as \(t^{2/3}\). In summary, Newtonian gravity permits localized subsystems because its force admits a controlled multipole expansion at large separation. How does this compare with the formation of bound structures of the distinguished fracton model? We will now argue that their mechanism of formation is distinct from Newtonian gravity. 

To see this, consider the formation of binaries. We know from the previous discussion that \emph{every} two-fracton
solution in \cref{eq:alphabeta_two_body_3d_highlights} for the
distinguished model \((\alpha,\beta)=(-2,1)\) is asymptotically radial with
\(r(t)\to\infty\) as \(|t|\to\infty\), and hence cannot realize a bound
orbit. Moreover, as shown in \cref{ang_sweep}, every nonradial
\((\ell\neq0)\) solution undergoes a universal \(120^\circ\) total
angular sweep and never completes a full revolution. This means that the long-lived compact binaries with multiple orbits and nearly constant physical size that appear in large-\(N\) dynamics cannot exist in isolation but emerge only as a part of a larger set. Local bound motion is therefore not an intrinsic property of a pair alone; it is induced by the surrounding many-body background. In this sense, the cluster regime has a Machian character. Local inertial and bound dynamics are not imposed as independent microscopic input but arise as properties of subsystems embedded in a global relational background. 

Mach's principle broadly associates local inertial properties with the rest of the matter distribution rather than treating them as wholly intrinsic~\cite{brans1961mach,barbour1995mach,sciama1953origin}. Fracton models represent a particular realization of this theme. The connection of fracton dynamics with Mach's principle has already been pointed out in several previous works~\cite{Pretko2017c,PrakashGorielySondhi_Fractons_2024,BabbarSadkiPrakashSondhi_Fractons_2025,PrakashSadkiSondhi_Fractons_2024}. The Machian nature is evident from the very nature of the Hamiltonians~\cref{eq:H_fracton,eq:H_alpha_beta}, which vanish for a single particle, leaving it immobile, and require others in its vicinity to move. The Machian nature of fracton systems highlighted in this section above is concerned with the \emph{nature} of motion in isolation compared with the presence of a collective background, rather than the ability to move itself. This feature can also be seen in other fracton models--- in Refs.~\cite{PrakashGorielySondhi_Fractons_2024,PrakashSadkiSondhi_Fractons_2024}, it was shown that two-particle bound oscillations are absent for Hamiltonians of the form \cref{eq:H_fracton}, but emerge in the vicinity of one additional particle. 

\subsection{Janus Point and Bidirectional Arrow of Time}
\begin{figure}[htbp]
	\centering
	\includegraphics[width=\textwidth]{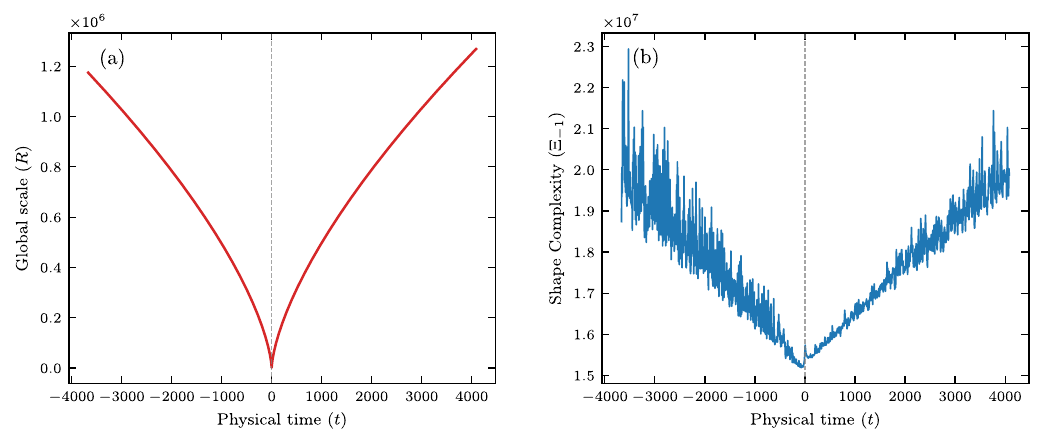}
	\caption{Bidirectional arrow observed in a randomly selected N=1000 trajectory: (a) evolution of the scale with physical time and (b) evolution of the shape complexity with physical time.}
	\label{fig:arrow}
\end{figure}

The Hamiltonian of the distinguished model \((\alpha,\beta)=(-2,1)\), which was shown to have stable expanding attractors,  is also time-reversal invariant, \((\mathbf x,\mathbf p,t)\mapsto(\mathbf x,-\mathbf p,-t)\). If we evolve a randomly initialized trajectory for $t \in (-\infty, +\infty)$, we should see expanding attractors for both $t \rightarrow \pm \infty$ when the scale diverges $R(\pm \infty) \rightarrow \infty$. This naturally generates a ``Janus point"~\cite{BarbourKoslowskiMercati_PhysRevLett.113.181101} on the time axis when the scale is at a minimum. On either side of the Janus point, $t_J^R$, we can define a geometric arrow-of-time that tracks the scale expansion, as shown in \cref{fig:arrow}(a).

The growth of scale $R$ provides an intuitive notion of the flow of time. This can be formalized as a \emph{statistical}  arrow of time by treating $R$ as a macroscopic observable and computing the corresponding Boltzmann entropy. For fixed conserved charges and sufficiently large \(R\), this gives:~\footnote{The $\ldots$ in the argument indicates other macroscopic observables that may be fixed in order to ensure that the phase space volume related to the Boltzmann entropy is finite. See \cref{app:Entropy} for more details}, 
\begin{equation}
    S_B(R,\ldots) = k_B a
		\log\!\left(\frac{R}{R_*}\right) + \text{const}, \qquad a=\left(\frac{9N-20}{2}\right) \label{eq:Entropy}
\end{equation}
where $R_*$ is some arbitrary reference scale. From \cref{eq:Entropy}, we see that Boltzmann entropy too has a minimum at $t_J^R$ and produces, on either side, \emph{statistical} arrows of time characterized by entropy growth.

A complementary analysis can be performed following~\cite{BarbourKoslowskiMercati_PhysRevLett.113.181101}, using the equal-mass shape potential $\Xi_{-1} = \sum_{j<k} |\mathbf{s}_j - \mathbf{s}_k|^{-1}$ as a measure of `shape complexity'. Since we know that the global minimum of \(\Xi_{-1}\) corresponds to homogeneous cluster-free distributions, for trajectories containing bound clusters (which are typical as we increase $N$), $\Xi_{-1}$ increases~\footnote{This should not be confused with the \emph{renormalized} unequal-mass shape potential $\widetilde{\Xi}_{-1} = \sum_{a<b}m_a m_b |\bar{\mathbf s}_a-\bar{\mathbf s}_b|^{-1}$, whose \emph{minima} describe the fixed points near the attractors. The effective clusters, and $\widetilde{\Xi}_{-1}$ are different on the two ends of the time axis, whereas $\Xi_{-1}$ is well-defined everywhere.} as the clusters reduces in overall size in shape space due to divergent intra-cluster contributions to $\Xi_{-1}$ within each cluster ${C}_a$, $\sum_{i<j \in C_a}|\mathbf s_j -\mathbf s_k|^{-1} \rightarrow \infty$ as $|\mathbf s_j -\mathbf s_k| \rightarrow 0$. The minimum of $\Xi_{-1}(t)$ for any given trajectory occurs when the clustering is least developed, maximally away from both asymptotic fixed points $t \rightarrow \pm \infty$. The growth of complexity around this new Janus point, $t_J^{\Xi}$, quantified by $\Xi_{-1}$, defines another bidirectional arrow of time as shown in \cref{fig:arrow}(b). Since the absence of clustering and minimization of overall scale occur maximally away from the fixed points, we expect the two Janus points to occur close to each other, $t_J^{R} \approx t_J^{\Xi}$.

In standard cosmology, the arrows of scale expansion and structure formation are usually traced to a special low-entropy early state, often formulated as a `past hypothesis'~\cite{Penrose1979,albert2000time}. An alternative paradigm~\cite{CarrollChen2004spontaneousinflationoriginarrow,CarrollChen_2005,BarbourKoslowskiMercati_PhysRevLett.113.181101} traces these arrows to the existence of generic Janus points. These are realized when the system operates in an unbounded phase space with a scale surplus~\cite{gryb2025accountarrowtimescale}. Our distinguished model also belongs to this paradigm. However, unlike other models, it has several other robust features relevant to cosmological dynamics discussed above.

	\section{Summary and Conclusions}
	\label{sec:conclusions}
	We have analyzed the scale-invariant, dipole-conserving family of classical fracton Hamiltonians \(H_{(\alpha,\beta)}\).  The exactly solvable two-body problem already displays key features of these models: radial trajectories have a Newtonian description with power-law potential exponent \(2\beta/\alpha\) and identically vanishing Newtonian energy, while in three dimensions every finite-angular-momentum trajectory asymptotically radializes.  On the Einstein-de Sitter (EdS) branch, \(\alpha=-2\beta\), this dual potential is inverse-distance and the separation grows as \(|t|^{2/3}\). The mechanical similarity extends this structure to the many-body problem.  Separating an overall scale from dimensionless position-momentum variables gives an autonomous shape dynamics, while the full dynamics remains Hamiltonian and obeys the Liouville theorem~\cite{Arnold1978,bravetti2022scalingsymmetriescontactreduction,Sloan2018}.  Radial fixed points of the reduced flow are central configurations of the Riesz shape potential with exponent \(q=\alpha+\beta\), and their remaining motion is homothetic, \(R(t)\propto |t|^{\alpha/(\alpha-\beta)}\).  Using the corresponding Riesz-equilibrium problem, we classified the regular large-\(N\) profiles into weighted balls and spherical shells~\cite{FrankMatzke2025,CarazzatoPratelliTopaloglu2026}.  On the EdS branch, we also obtained a conditional local-stability criterion for expanding regular extrema.  The auxiliary positivity condition was proved for the equilateral three-body configuration and verified numerically for the larger regular configurations examined here; its general finite-\(N\) validity is formulated in \cref{conj:Cq}.
	
	These results single out \((\alpha,\beta)=(-2,1)\).  It is the unique model on the EdS branch for which the fracton fixed-point equation coincides with the Newtonian inverse-square central-configuration equation.  Its regular continuum minimum is a homogeneous solid ball, and every homothetic fixed-point trajectory has an exact zero-energy Newtonian gravitational dual, with the effective gravitational coupling fixed by the fracton initial conditions. At moderate \(N\), the random initial conditions studied numerically approach an expanding near-minimum fixed point, pass the Newtonian central-configuration diagnostics, and converge to the EdS scale law. At larger \(N\), the simulations reveal a new and more general class of late-time states.  Tight clusters retain internal motion and remain approximately fixed in physical size, while their coarse-grained centers separate as \(t^{2/3}\), approach unequal-mass Newtonian central configurations, and preserve homogeneity at large scales.  These observations motivate \cref{conj:renormalized_dynamics}: after an asymptotic separation of intra- and inter-cluster scales, the centers should obey an effective unequal-mass fracton Hamiltonian.  The corresponding renormalized central-configuration equation and zero-energy Newtonian dual are corollaries to that conjecture.  Finally, the microscopic Hamiltonian is time-reversal invariant, yet the trajectories examined here exhibit two arrows away from a Janus point: the global scale grows, with leading Boltzmann entropy \(S_B\sim k_B(9N-20)\log R/2\), and the shape complexity grows as clustering develops~\cite{BarbourKoslowskiMercati_PhysRevLett.113.181101,BabbarSadkiPrakashSondhi_Fractons_2025}.
	
	We now turn to the cosmological interpretation and its scope.  In Newtonian cosmology, an FLRW-like dust expansion translates into several special requirements: homothetic motion requires a central configuration; large-scale homogeneity selects a restricted part of that set; spatial flatness fixes the Newtonian scale energy to its critical null value; and a scale-and-structure arrow is normally associated with a special smooth state, {i.e., a perturbed expanding homogeneous central configuration}~\cite{10.1093/qmath/os-5.1.64,Ellis_2013,Ellis_2015,Battye_2003,Penrose1979}.  Within the distinguished fracton model, the analogues of these requirements are tied to one reduced attractor.  Its fixed-point equation selects the central configuration and the homogeneous regular or coarse-grained profile; its scale equation selects both \(R\propto t^{2/3}\) and the zero-energy Newtonian dual; and its bidirectional Janus-point dynamics generates scale and complexity arrows without imposing a low-complexity condition at one temporal boundary.  Thus, for the analytic basins and numerical simulations actually studied, properties that are separately fine-tuned in the Newtonian description become correlated late-time consequences.  This is the sense---and the limited sense---in which the toy model resolves the cosmological fine-tunings posed in the main text.
	
	This fracton \(N\)-body system is not proposed as a replacement for relativistic cosmology. It evolves particles on a fixed Euclidean space with a distinguished time; it has no dynamical spacetime metric, causal light cones or relativistic constraint equations, and it omits radiation, pressure, the observed matter sectors, a cosmological constant, quantum fluctuations and their transfer to observables.  Its scale \(R\) is a many-particle size, not an FLRW metric scale factor, and its homogeneous density is a particle distribution, not an exact spacetime symmetry.  Standard FLRW cosmology, by contrast, is a sector of Einstein gravity with a causal perturbation theory and an {empirically estimated} matter content~\cite{dodelson2020modern,Planck:2018vyg}.  All comparisons made here are consequently intended to illuminate broad dynamical and conceptual mechanisms, not to assert precise physical equivalences.
	
	Viewed through this conceptual lens, the present mechanism and cosmic inflation share a central idea: a strong attractor removes sensitivity to a broad set of initial conditions in the macroscopic variables.  In both cases, trajectories are focused toward a highly uniform background, so apparent fine-tunings can be reinterpreted as attractor data~\cite{guth1981inflationary,Linde1983,RemmenCarroll2013}. This shifts the fine-tuning question from why the universe began on a special macroscopic trajectory to whether the dynamics possesses an appropriate attractor with a sufficiently broad basin, an explanatory strategy shared by inflation and our model despite their very different microscopic content. Inflation redshifts pre-existing curvature, shear, and matter inhomogeneity; at the same time, adiabatic curvature perturbations stretched beyond the Hubble radius remain conserved, and later seed structure formation~\cite{Wald1983,MukhanovFeldmanBrandenberger1992,dodelson2020modern}.  The homogenizing inflationary regime must end, and its energy must be transferred into the hot plasma through reheating before the later radiation- and matter-dominated epochs in which gravitational instability builds structure~\cite{AllahverdiEtAl2010}.  In this sense a viable inflationary history requires an exit from the homogenizing attractor before a distinct clustering epoch. 
    
    Here too, the fracton model suggests a different logical possibility.  Its Machian microscopic dynamics does not permit an isolated pair to support the long-lived bound motion seen inside the many-body state.  Bound two-body subsystems of this kind therefore require a suitable collective background and a separation of scales.  In the simulations they appear as the global motion approaches the conjectured homothetic clustered attractor, where approximately stable local structures coexist with large-scale homogeneity.  The data are thus consistent with the same microscopic dynamics both erasing sensitivity to initial {conditions} and enabling local subsystem dynamics, rather than being explicitly switched off before clustering begins.  There are partial precedents for pieces of this picture: fracton-gravity proposals make inertia depend on the surrounding matter distribution~\cite{Pretko2017c}; relational Janus-point models produce increasingly autonomous bound subsystems as complexity grows~\cite{BarbourKoslowskiMercati_PhysRevLett.113.181101}; and quantum graphity~\cite{KonopkaMarkopoulouSeverini2008} posits a highly connected phase without ordinary locality followed by a low-energy phase in which locality and subsystems emerge.
	
	We close with several concrete questions.  First, the intracluster dynamics should be derived systematically from the microscopic equations.  It would be interesting to determine whether any controlled large-\(N\) or scale-separated limit yields approximately Newtonian laws for internal cluster motion; the present results do not establish such a limit.  The same multiscale analysis should quantify whether local binding and internal motion feed back on the evolution of the global scale, providing a particle-model analogue of the backreaction problem in inhomogeneous cosmology~\cite{Buchert2000}.  The most immediate mathematical step is to prove, weaken or delimit \cref{conj:renormalized_dynamics}, including its asserted scale separation and effective unequal-mass Hamiltonian.
	
	Second, our analysis has emphasized the late-time clustered state rather than classified the full approach to it.  A natural test is whether typical large-\(N\) histories separate into two dynamical regimes: an initial homogenizing stage, giving an analog to early-universe smoothness, followed by a regime of local collapse and cluster formation similar to gravitational instability and structure formation.  Establishing such a sequence requires time-resolved observables and basin studies; it should not be inferred from the late-time data alone. In particular, it remains open whether generic trajectories reach, in finite time, a neighbourhood in which the local stability analysis applies, and hence approach a regular or clustered fixed point asymptotically in each time direction.  Beyond this toy framework, one may ask whether a related attractor can be embedded in relativistic cosmology, and whether continuum and tensor-gauge descriptions of fractons can connect to an Einstein-gravity regime without inheriting the nonrelativistic limitations emphasized above~\cite{Pretko2017c,JainJensen2022,Bidussi2023,SadkiPrakashSondhi_ContinuumFractons_2025}.
	
	Finally, the reduction itself can be made more intrinsic.  The variables in \cref{eq:shape_coordinates,eq:shape_t} remove the overall position scale but retain a distinguished clock~\cite{bravetti2022scalingsymmetriescontactreduction}.  A canonical gauge theory on shape space obtained through symplectic reduction~\cite{MarsdenWeinstein1974,GomesGrybKoslowski2011,Barbour2012}, or a full quotient by the independent position and momentum scalings, would provide a more complete relational formulation.

\acknowledgments
The authors thank Subir Sarkar, James Binney, and Julian Barbour for helpful discussions. A.P. is grateful to St. John's College and Rudolf Peierls Centre for Theoretical Physics at the University of Oxford for their hospitality during the completion of this work. D.P.J. acknowledges support from ICTP through the Associates Programme (2022--2027).  S.L.S. acknowledges support from Leverhulme International Professorship grant LIP-202-014 and EPSRC grant EP\slash X030881\slash 1.

\section*{Statement on the Use of Generative AI}
	Generative AI tools (OpenAI's ChatGPT and Codex) were used, for the most part, as assistive tools for editorial polishing; exploratory assistance with algebraic derivations and proof checking; generating, debugging, and documenting code for numerical analyses and plots derived from the resulting data; and supplementary cross-checks of calculations, references, and manuscript consistency. These uses were predominantly labour- and time-saving.

	We also wish to acknowledge two AI-assisted inputs that played a more substantive role in the scientific development of this work. AI-assisted literature exploration brought the Riesz-gas and equilibrium-measure literature to our attention, in particular the asymptotic-distribution formulas used in \cref{sec:regular_asymptotic_distributions,app:LargeNDistributions}. AI-assisted exploration also revealed a possibility that we had not anticipated: that a conditional local-stability analysis could be carried out. The analysis identified the restricted assumptions under which this result could be obtained and then worked through the corresponding derivation presented in \cref{sec:conditional_local_stability,app:StabilityDetails}. In our judgment, these two materially important results were unlikely to have emerged from our usual unaided research process.

	All AI outputs were treated as unverified suggestions rather than as evidence or independent validation. The authors independently located and checked the primary literature; reconstructed and verified all derivations and proofs; reviewed and tested all code, calculations, numerical results, and plotting outputs; and made all scientific and interpretive decisions. The authors take full responsibility for the accuracy, originality, and integrity of the manuscript.

    \appendix
\crefalias{section}{appendix}
\crefalias{subsection}{appendix}

	\section{Exact Three-dimensional Two-Particle Motion}
	\label{app:two_particle_3d}
	This appendix supplies the full two-particle calculation summarized in
	\cref{sec:Warmup,Scale_inv}.  We first treat the general
	\((\alpha,\beta)\) Hamiltonian and then recover the normalization used in
	the warm-up.
	
	\subsection{Relative Dynamics and Conserved Quantities}
	Set
	\(\mathbf r=\mathbf x_1-\mathbf x_2\),
	\(\mathbf p=\mathbf p_1-\mathbf p_2\), and \(r=|\mathbf r|\).  For two
	particles, \cref{eq:H_alpha_beta} and its relative equations are
	\begin{align}
		H
		&=|\mathbf p|^\alpha r^\beta ,\qquad
		\dot{\mathbf r}
		=2\alpha|\mathbf p|^{\alpha-2}\mathbf p\,r^\beta,\qquad 
		\dot{\mathbf p}
		=-2\beta|\mathbf p|^\alpha r^{\beta-2}\mathbf r .
		\label{eq:alphabeta_two_body_3d_eom}
	\end{align}
	We have a conserved fracton energy \(E_f>0\) and \(\alpha\ne0\), and work on a smooth patch away
	from any singular sets.  The center of position is fixed, and
	rotational invariance conserves
	\begin{equation}
		\boldsymbol\ell:=\mathbf r\times\mathbf p .
		\label{eq:alphabeta_two_body_angular_momentum}
	\end{equation}
	Thus \(\mathbf r(t)\) remains in the plane orthogonal to
\(\boldsymbol\ell\) whenever \(\boldsymbol\ell\neq0\), while
\(\boldsymbol\ell=0\) defines an invariant radial submanifold.
	
	\subsection{The Radial Submanifold and its Newtonian Dual}
	On \(\boldsymbol\ell=0\), the direction of \(\mathbf r\) is fixed.
	Introduce an independent Newtonian phase space with positions
	\(\mathbf x_a\) and canonical momenta \(\boldsymbol\pi_a\), and define
	\begin{equation}
		\kappa:=\frac{2\beta}{\alpha},\qquad
		g_{\alpha,\beta}(E_f)
		:=\alpha^2E_f^{\frac{2(\alpha-1)}{\alpha}} .
		\label{eq:alphabeta_dual_parameters}
	\end{equation}
	The pairwise Newtonian Hamiltonian
	\begin{equation}
		H^{\rm dual}(\mathbf x,\boldsymbol\pi;E_f)
		=\frac12\sum_{a=1}^2|\boldsymbol\pi_a|^2
		-g_{\alpha,\beta}(E_f)
		|\mathbf x_1-\mathbf x_2|^\kappa
		\label{eq:alphabeta_radial_newtonian_hamiltonian}
	\end{equation}
	generates precisely the radial fracton position-space trajectories when
	\(\boldsymbol\pi_a=\dot{\mathbf x}_a\).  Indeed, the fracton energy
	identity gives
	\begin{equation}
		\frac12\sum_{a=1}^2|\dot{\mathbf x}_a|^2
		=g_{\alpha,\beta}(E_f)r^\kappa ,
		\qquad\Longrightarrow\qquad
		E_{N}=0 ,
		\label{eq:alphabeta_radial_newtonian_energy}
	\end{equation}
	while direct differentiation of the radial fracton velocity gives
	\begin{equation}
		\ddot{\mathbf x}_1
		=2\alpha\beta
		E_f^{\frac{2(\alpha-1)}{\alpha}}
		r^{\frac{2\beta}{\alpha}-2}\mathbf r
		=g_{\alpha,\beta}(E_f)\kappa
		r^{\kappa-2}\mathbf r,
		\qquad
		\ddot{\mathbf x}_2=-\ddot{\mathbf x}_1 ,
		\label{eq:alphabeta_radial_newtonian_force}
	\end{equation}
	which is Hamilton's equation from
	\cref{eq:alphabeta_radial_newtonian_hamiltonian}.
	
	After the irrelevant center-of-mass translation and uniform motion are
	removed, the matching Newtonian initial data obey
	\begin{equation}
		\boldsymbol\pi_1+\boldsymbol\pi_2=0,\qquad
		\mathbf r\times(\boldsymbol\pi_1-\boldsymbol\pi_2)=0,\qquad
		E^{\rm dual}=0 .
		\label{eq:alphabeta_dual_initial_data}
	\end{equation}
	The last two conditions are genuine restrictions in the Newtonian
	problem: the data must be radial and lie on its zero-energy surface.  In
	the fracton problem, by contrast, the zero Newtonian energy follows
	identically from the freely chosen conserved \(E_f\).  The coupling
	\(g_{\alpha,\beta}(E_f)\) is likewise generated dynamically by \(E_f\),
	rather than inserted as a parameter of the fracton Hamiltonian.
	
	This does not generally extend to a nonradial orbit.  For example,
	the mechanical relative angular momentum built from the fracton
	velocity is
	\begin{equation}
		\mathbf r\times\dot{\mathbf r}
		=2\alpha E_f^{\frac{\alpha-2}{\alpha}}
		r^{\frac{2\beta}{\alpha}}\boldsymbol\ell .
		\label{eq:alphabeta_mechanical_angular_momentum}
	\end{equation}
	For a nontrivial interaction \(\beta\ne0\), this is not conserved along
	a generic nonradial orbit, whereas a position-only central Newtonian
	Hamiltonian must conserve it.  Thus
	\cref{eq:alphabeta_radial_newtonian_hamiltonian} is an exact dual of the
	radial invariant dynamics, not of the full three-dimensional fracton
	phase space.  The trivial \(\beta=0\) case reduces instead to free
	motion.
	
	\subsection{Exact Angular Motion and Asymptotic Aadialization}
	Let \(p_r=\mathbf p\cdot\widehat{\mathbf r}\),
	\(\ell=|\boldsymbol\ell|\), and
	\begin{equation}
		\nu:=1-\frac{\beta}{\alpha}
		=\frac{\alpha-\beta}{\alpha}.
	\end{equation}
	Energy and angular-momentum conservation imply
	\begin{equation}
		|\mathbf p|
		=E_f^{1/\alpha}r^{-\beta/\alpha},
		\qquad
		p_r^2
		=E_f^{2/\alpha}r^{-2\beta/\alpha}
		-\frac{\ell^2}{r^2}.
		\label{eq:alphabeta_two_body_radial_momentum}
	\end{equation}
	For \(\alpha\ne\beta\), choose \(t_*\) at the radial turning point.
	Integrating first the radial and then the angular equation gives
	\begin{align}
		r(t)^{2\nu}
		&=E_f^{-2/\alpha}
		\left[
		\ell^2+4(\alpha-\beta)^2E_f^2(t-t_*)^2
		\right],
		\label{eq:alphabeta_two_body_exact_radius}\\
		\phi(t)
		&=\phi_*+\frac{\alpha}{\alpha-\beta}
		\arctan\left[
		\frac{2(\alpha-\beta)E_f(t-t_*)}{\ell}
		\right],
		\qquad \ell\ne0 .
		\label{eq:alphabeta_two_body_exact_angle}
	\end{align}
	The turning point is a minimum separation for \(\nu>0\) and a maximum separation
	for \(\nu<0\).  On setting \(\ell=0\), one obtains the exact radial
	solution
	\begin{equation}
		r(t)
		=\left|
		2(\alpha-\beta)
		E_f^{\frac{\alpha-1}{\alpha}}(t-t_*)
		\right|^{\frac{\alpha}{\alpha-\beta}},
		\label{eq:alphabeta_two_body_3d_radial_solution}
	\end{equation}
	which is the one-dimensional result \eqref{eq:two_particle_alphabeta}.
	
	The exact solution also proves the asymptotic statement used in the
	main text.  If \(\chi\) is the angle between \(\mathbf p\) and
	\(\mathbf r\), then
	\begin{align}
		\sin\chi
		&=\frac{\ell}
		{\sqrt{\ell^2+
				4(\alpha-\beta)^2E_f^2(t-t_*)^2}},\qquad
		\cos\chi
		=\frac{2(\alpha-\beta)E_f(t-t_*)}
		{\sqrt{\ell^2+
				4(\alpha-\beta)^2E_f^2(t-t_*)^2}} .
		\label{eq:alphabeta_two_body_radial_alignment}
	\end{align}
	Hence every finite-\(\ell\) solution with \(\alpha\ne\beta\) approaches
a radial fixed point of the orientation dynamics and, since
\(\dot{\mathbf r}\propto\mathbf p\), the trajectory itself becomes
asymptotically radial, with
	\begin{equation}
		r(t)\sim
		\left|
		2(\alpha-\beta)
		E_f^{\frac{\alpha-1}{\alpha}}t
		\right|^{\frac{\alpha}{\alpha-\beta}} .
		\label{eq:alphabeta_two_body_asymptotic_radius}
	\end{equation}
	The full nonradial curve is not a trajectory of
	\cref{eq:alphabeta_radial_newtonian_hamiltonian}; only its leading
	asymptotic radial trajectory is reproduced by choosing the corresponding
	fine-tuned data \eqref{eq:alphabeta_dual_initial_data}.
	
	The case \(\alpha=\beta \neq 0\) is exceptional.  Here
	\(r|\mathbf p|=E_f^{1/\alpha}\), so \(\chi\) is constant rather than
	approaching a radial value.  The exact orbit is a logarithmic spiral. Choosing an arbitrary reference
time \(t_0\), with \(r_0=r(t_0)\) and \(\phi_0=\phi(t_0)\), one finds
	\begin{align}
		r(t)
		&=r_0\exp\left[
		2\alpha E_f^{\frac{\alpha-1}{\alpha}}
		\cos\chi\,(t-t_0)
		\right],
		\qquad
		\phi(t)
		=\phi_0+
		2\alpha\ell E_f^{\frac{\alpha-2}{\alpha}}(t-t_0).
		\label{eq:alphabeta_equal_exact_solution}
	\end{align}
	Only its \(\ell=0\) submanifold is radial and therefore described by the
	dual Newtonian Hamiltonian above.
	
	\subsection{Einstein-de Sitter Line}
	On the EdS line \(\alpha=-2\beta\) of
	\cref{eq:EdS_branch}, one has
	\(\nu=3/2\), \(\kappa=-1\), and
	\begin{align}
		r(t)^3
		&=E_f^{-2/\alpha}
		\left[\ell^2+9\alpha^2E_f^2(t-t_*)^2\right],
	\qquad
		\phi(t)
		=\phi_*+\frac23
		\arctan\left[
		\frac{3\alpha E_f(t-t_*)}{\ell}
		\right].
		\label{eq:alphabeta_two_body_eds_exact}
	\end{align}
	Every finite-\(\ell\) trajectory therefore satisfies
\(r(t)\to\infty\) as \(|t|\to\infty\), with
\(r(t)\propto|t|^{2/3}\), and is therefore unbounded. On the radial submanifold, the
	dual interaction is exactly gravitational,
	\begin{equation}
		H^{\rm dual}
		=\frac12\sum_{a=1}^2|\boldsymbol\pi_a|^2
		-\frac{G_{\rm eff}^{(2)}}{r},
		\qquad
		G_{\rm eff}^{(2)}
		=\alpha^2E_f^{\frac{2(\alpha-1)}{\alpha}},
		\qquad E_{N}=0 .
		\label{eq:alphabeta_two_body_eds_dual}
	\end{equation}
	
	Finally, the warm-up Hamiltonian with \(K(r)=1/r\) is
	\(H=\tfrac12|\mathbf p|^2/r\), whereas
	\cref{eq:H_alpha_beta} at \((\alpha,\beta)=(2,-1)\) has no factor of
	\(1/2\).  Accounting for the corresponding energy and time
	normalizations gives
	\begin{align}
		r(t)^3
		&=\frac{\ell^2}{2E_f}+18E_f(t-t_*)^2,
		\qquad
		\phi(t)
		=\phi_*+\frac23
		\arctan\left[\frac{6E_f(t-t_*)}{\ell}\right],
		\qquad \ell\ne0 ,
		\label{eq:two_fracton_3d_exact_solution}
	\end{align}
	and its radial Newtonian dual is
	\begin{equation}
		H^{\rm dual}
		=\frac12\sum_{a=1}^2|\boldsymbol\pi_a|^2-\frac{2E_f}{r}=0 .
	\end{equation}
	This recovers all the three-dimensional claims quoted in
	\cref{sec:Warmup}.

\subsection{Universal Angular Sweep}
\label{ang_sweep}

For \(\alpha\neq\beta\) and \(\ell\neq0\), the exact solution
\eqref{eq:alphabeta_two_body_exact_angle} immediately gives the total
angular sweep between the two asymptotic ends of the trajectory,
\begin{equation}
	|\Delta\phi|
	:=\left|\phi(+\infty)-\phi(-\infty)\right|
	=\left|
	\frac{\alpha\pi}{\alpha-\beta}
	\right|.
	\label{eq:alphabeta_total_angular_sweep}
\end{equation}
Remarkably, \(|\Delta\phi|\) is completely determined by the Hamiltonian
exponents \((\alpha,\beta)\), and is independent of both the conserved
energy \(E_f\) and the canonical angular momentum \(\ell\). Consequently,
every nonradial two-particle trajectory generated by the Hamiltonian
\eqref{eq:H_alpha_beta} undergoes the same total angular sweep.

Equation~\eqref{eq:alphabeta_total_angular_sweep} further implies that a
complete angular revolution occurs if and only if
\begin{equation}
	\left|
	\frac{\alpha}{\alpha-\beta}
	\right|
	\ge 2.
	\label{eq:alphabeta_complete_revolution}
\end{equation}
For the Einstein-de Sitter line \(\alpha=-2\beta\), one has
\begin{equation}
	|\Delta\phi|=\frac{2\pi}{3},
\end{equation}
so every nonradial trajectory undergoes a universal \(120^\circ\) total
angular sweep and never completes a full revolution.
	
\section{Central Configurations and Shape Potentials}
\label{app:CentralConfigs}

\subsection{Critical Points of the Constrained Shape Potential}
The shape potential is
\begin{equation}
	\Xi_q =
	\begin{cases}
		\displaystyle \sum_{i<j} |\mathbf{s}_i-\mathbf{s}_j|^q, & q\neq0,\\[1mm]
		\displaystyle \sum_{i<j} \log|\mathbf{s}_i-\mathbf{s}_j|, & q=0,
	\end{cases}
	\label{appeq:shape_potential}
\end{equation}
The logarithmic definition is the renormalized \(q\to0\) limit
\begin{equation}
	\Xi_0
	=\lim_{q\to0}
	\frac{\displaystyle\sum_{i<j}|\mathbf s_i-\mathbf s_j|^q
		-\binom N2}{q}.
	\label{appeq:log_shape_potential_limit}
\end{equation}
The subtracted term is independent of the configuration, but keeping it
explicit makes the \(q=0\) variational problem the smooth limit of the
power-law branch.  The shape coordinates obey
\begin{equation}
	\sum_{i=1}^N |\mathbf{s}_i|^2=1,\qquad
	\sum_{i=1}^N \mathbf{s}_i=0.
	\label{appeq:shape_coordinate_constraint}
\end{equation}
Let \(c_q=q\) for \(q\neq0\) and \(c_0=1\).  Then
\begin{equation}
	\nabla_{\mathbf s_i}\Xi_q=c_q\,\mathbf H_q^i,\qquad
	\mathbf H_q^i=\sum_{j\neq i}
	\frac{\mathbf s_i-\mathbf s_j}{|\mathbf s_i-\mathbf s_j|^{2-q}} .
\end{equation}
The logarithmic case has the same formula with \(q=0\), because
\(\nabla_{\mathbf r}\log|\mathbf r|=\mathbf r/|\mathbf r|^2\).

Introduce Lagrange multipliers by
\begin{equation}
	\mathcal L=\Xi_q
	-\frac{\lambda}{2}\left(\sum_i|\mathbf s_i|^2-1\right)
	-\mathbf b\cdot\sum_i\mathbf s_i . \label{eq:Xi_lagrange}
\end{equation}
Stationarity gives
\begin{equation}
	c_q\mathbf H_q^i=\lambda\mathbf s_i+\mathbf b.
\end{equation}
Summing over \(i\) and using \(\sum_i\mathbf s_i=0\) gives \(\mathbf b=0\).
Taking the scalar product with \(\mathbf s_i\), summing over \(i\), and
using \(\sum_i|\mathbf s_i|^2=1\) gives
\begin{equation}
	\lambda=c_q S_q,\qquad
	S_q:=\sum_i\mathbf s_i\cdot\mathbf H_q^i
	=\sum_{i<j}|\mathbf s_i-\mathbf s_j|^q .
\end{equation}
Thus every constrained critical point satisfies
\begin{equation}
	 {\mathbf H_q^i=S_q\mathbf s_i,\qquad i=1,\ldots,N,}
	\label{appeq:central_config}
\end{equation}
which is the central-configuration equation used in the main text. \cref{eq:Xi_lagrange} reduces to
\begin{equation}
	\mathcal L=\Xi_q
	-\frac{c_q S_q}{2}\left(\sum_i|\mathbf s_i|^2-1\right)
	. \label{appeq:Xi_lagrange_fixed}
\end{equation}

\subsection{The Fixed Second Moment Identity}
For any centered shape, the pair distances obey
\begin{equation}
	\sum_{i<j}|\mathbf s_i-\mathbf s_j|^2
	=N\sum_i|\mathbf s_i|^2-\left|\sum_i\mathbf s_i\right|^2=N .
	\label{appeq:S2identity}
\end{equation}
This identity is the finite-\(N\) origin of the harmonic constraint in the
continuum Riesz problem.  It also explains why \(\Xi_2=N\) is flat on the
shape sphere.

\subsection{Unequal-mass Central Configurations}
\label{app:UnequalMassCentralConfigs}
The preceding construction has a direct unequal-mass extension that is
useful for the renormalized configurations appearing in
\cref{mass_central_config}.  Let \(m_i>0\), let \(\mathbf P_i\) be the
canonical momentum of the \(i\)-th effective particle, and define its
specific momentum by \(\boldsymbol\pi_i=\mathbf P_i/m_i\).  The natural
mass-weighted version of \cref{eq:H_alpha_beta} is
\begin{equation}
	H^{(m)}_{(\alpha,\beta)}
	=\sum_{i<j}m_im_j
	\left|\boldsymbol\pi_i-\boldsymbol\pi_j\right|^\alpha
	\left|\mathbf x_i-\mathbf x_j\right|^\beta .
	\label{appeq:mass_hamiltonian}
\end{equation}
When the effective particles represent clusters of identical microscopic
particles, the factor \(m_im_j\) counts the number of inter-cluster pairs.
Writing \(M=\sum_i m_i\), the corresponding centered and scale-normalized
variables may be chosen as
\begin{align}
	\mathbf X_{\rm C}
	&=\frac{1}{M}\sum_i m_i\mathbf x_i,
	&
	\boldsymbol\pi_{\rm C}
	&=\frac{1}{M}\sum_i\mathbf P_i,
	&
	R^2
	&=\sum_i m_i|\mathbf x_i-\mathbf X_{\rm C}|^2,
	\notag\\
	\mathbf s_i
	&=\frac{\mathbf x_i-\mathbf X_{\rm C}}{R},
	&
	\mathbf u_i
	&=R^{\beta/\alpha}
	(\boldsymbol\pi_i-\boldsymbol\pi_{\rm C}).
	\label{appeq:mass_shape_variables}
\end{align}
They obey the mass-weighted constraints
\begin{equation}
	\sum_i m_i\mathbf s_i=0,\qquad
	\sum_i m_i|\mathbf s_i|^2=1,\qquad
	\sum_i m_i\mathbf u_i=0.
	\label{appeq:mass_shape_constraints}
\end{equation}

The shape potential and central-configuration force are replaced by
\begin{align}
	\Xi_{q,m}
	&=
	\begin{cases}
		\displaystyle\sum_{i<j}m_im_j|\mathbf s_i-\mathbf s_j|^q,
		& q\neq0,\\[1mm]
		\displaystyle\sum_{i<j}m_im_j\log|\mathbf s_i-\mathbf s_j|,
		& q=0,
	\end{cases}
	\label{appeq:mass_shape_potential}\\
	\mathbf H_{q,m}^i
	&=\sum_{j\neq i}m_j
	|\mathbf s_i-\mathbf s_j|^{q-2}
	(\mathbf s_i-\mathbf s_j),
	\notag\\
	S_{q,m}
	&:=\sum_i m_i\mathbf s_i\cdot\mathbf H_{q,m}^i
	=\sum_{i<j}m_im_j|\mathbf s_i-\mathbf s_j|^q .
	\label{appeq:mass_HS}
\end{align}
Repeating the Lagrange-multiplier argument with the constraints in
\cref{appeq:mass_shape_constraints} gives
\begin{equation}
	 {\mathbf H_{q,m}^i=S_{q,m}\mathbf s_i.}
	\label{appeq:mass_central_config}
\end{equation}
Thus the radial fixed points of \cref{appeq:mass_hamiltonian} are precisely
the unequal-mass central configurations of the weighted shape potential.
For \(q=-1\), \cref{appeq:mass_central_config} is the renormalized
Newtonian central-configuration equation \eqref{mass_central_config}, with
\(S_{-1,m}=\widetilde S_{-1}\).

For completeness, the fixed-second-moment identity also has the weighted
form
\begin{equation}
	\sum_{i<j}m_im_j|\mathbf s_i-\mathbf s_j|^2
	=M\sum_i m_i|\mathbf s_i|^2
	-\left|\sum_i m_i\mathbf s_i\right|^2
	=M .
	\label{appeq:mass_S2_identity}
\end{equation}
Consequently the harmonic-constraint interpretation and the flatness of
the \(q=2\) potential carry over with \(N\) replaced by the total mass
\(M\).

\section{Dictionary to Known Power-law Equilibrium Measures}
	\label{app:LargeNDistributions}

	Here we explain how known results for attractive-repulsive power-law
	energies give \cref{thm:distribution}.  Only two translations are needed:
	we express the particle cloud in order-one coordinates and match our shape
	energy to the normalization and sign conventions of Frank and
	Matzke~\cite[Theorems~1 and 2 and Lemma~10]{FrankMatzke2025}.  Their work
	proves the global character and uniqueness of the profiles quoted below;
	we do not repeat that analysis here.

Because a typical shape coordinate is of order \(N^{-1/2}\), introduce the
	order-one coordinates \(\mathbf x_i=\sqrt N\,\mathbf s_i\).  The normalized
	particle density is
	\begin{equation}
		\mathbf x_i=\sqrt N\,\mathbf s_i,\qquad
		\mu_N=\frac1N\sum_{i=1}^N\delta(\mathbf x- \mathbf x_i).
		\label{appeq:empirical_measure}
	\end{equation}
	Thus \(\mu_N\) places weight \(1/N\) at each particle.  The shape constraints
	simply say that this density has unit total weight, center at the origin,
	and unit mean-square radius:
	\begin{equation}
		\int_{\R^3}\dd\mu_N=1,\qquad
		\int_{\R^3}\mathbf x\,\dd\mu_N=0,\qquad
		\int_{\R^3}|\mathbf x|^2\dd\mu_N=1.
		\label{appeq:mu_constraints}
	\end{equation}
	For a continuum cloud described by \(\mu\), denote the average power-law pair
	interaction by
	\begin{equation}
		I_q[\mu]=\iint |\mathbf x-\mathbf y|^q
		\dd\mu(\mathbf x)\dd\mu(\mathbf y)\quad(q\neq0),
		\qquad
		I_0[\mu]=\iint\log|\mathbf x-\mathbf y|
		\dd\mu(\mathbf x)\dd\mu(\mathbf y).
		\label{appeq:Iq_definition}
	\end{equation}
	For \(N\) particles we omit \(i=j\), since a particle does not interact
	with itself, and define
	\begin{align}
		I_q^{\mathrm{off}}[\mu_N]
		&:=\frac1{N^2}\sum_{i\neq j}
		|\mathbf x_i-\mathbf x_j|^q,
		&&q\neq0, \notag\\
		I_0^{\mathrm{off}}[\mu_N]
		&:=\frac1{N^2}\sum_{i\neq j}
		\log|\mathbf x_i-\mathbf x_j|.
		\label{appeq:empirical_off_diagonal_energy}
	\end{align}
	In this notation,
	\begin{align}
		\Xi_q
		&=\frac12N^{2-q/2}I_q^{\mathrm{off}}[\mu_N]
		\quad(q\neq0), \notag\\
		\Xi_0
		&=\frac12N^2I_0^{\mathrm{off}}[\mu_N]
		-\frac14N(N-1)\log N.
		\label{appeq:finiteN_energy_scaling}
	\end{align}
	The last term at \(q=0\) is independent of the positions of the particles.
Therefore, maximizing or minimizing \(\Xi_q\) is exactly the same as
	maximizing or minimizing the corresponding empirical pair energy.

	The key simplification is that the quadratic pair interaction cannot
	distinguish two allowed shapes.  Indeed, the center and size constraints
	fix it to the same value, \(I_2=2\), for every cloud:
	\begin{equation}
		\frac12 I_2[\mu]
		=\int|\mathbf x|^2\dd\mu(\mathbf x)
		-\left|\int\mathbf x\,\dd\mu(\mathbf x)\right|^2
		=1.
		\label{appeq:quadratic_pair_identity}
	\end{equation}
	Frank and Matzke classify the global minima of the two-power energy
	\begin{equation}
		\mathcal F_{a,b}[\mu]
		=\frac12\iint\left(
		\frac{|\mathbf x-\mathbf y|^a}{a}
		-\frac{|\mathbf x-\mathbf y|^b}{b}\right)
		\dd\mu(\mathbf x)\dd\mu(\mathbf y),
		\label{appeq:FrankMatzke_energy}
	\end{equation}
	where \(r^b/b\) is replaced by \(\log r\) when \(b=0\).  Since the
	quadratic part is fixed on our shape space, their energy reduces there to
	\begin{align}
		\mathcal F_{2,q}[\mu]
		&=\frac12-\frac{1}{2q}I_q[\mu],
		&&q\neq0,\quad q<2,
		\label{appeq:F2q_dictionary}\\
		\mathcal F_{2,0}[\mu]
		&=\frac12-\frac12I_0[\mu],
		\label{appeq:F20_dictionary}\\
		\mathcal F_{q,2}[\mu]
		&=\frac{1}{2q}I_q[\mu]-\frac12,
		&&q>2.
		\label{appeq:Fq2_dictionary}
	\end{align}
	It follows directly that minimizing \(\mathcal F_{2,q}\) means minimizing
	\(I_q\) for \(q<0\) and maximizing \(I_q\) for \(0\le q<2\);
	minimizing \(\mathcal F_{q,2}\) is equivalent to minimizing \(I_q\) for
	\(q>2\).  In Ref.~\cite{FrankMatzke2025} the cloud is also free to choose
	its overall radius, whereas our shape convention fixes
	\(\langle|\mathbf x|^2\rangle=1\).  This difference affects only the scale,
	not the profile.  To see this, rescale any unit-size profile as
	\(\mathbf x\to\lambda\mathbf x\): its \(q\)-power energy scales as
	\(\lambda^q\), while its quadratic energy scales as \(\lambda^2\) (and the
	logarithmic energy shifts by \(\log\lambda\)).  Minimization over
	\(\lambda\) therefore chooses the preferred radius of each profile; the
	remaining comparison between profiles is precisely
	\eqref{appeq:F2q_dictionary}--\eqref{appeq:Fq2_dictionary}.  We may thus
	use their profile and simply recenter and rescale it to unit mean-square
	radius.

	In three dimensions, Theorem~2 of Ref.~\cite{FrankMatzke2025}, applied with
	\((a,b)=(2,q)\) and rescaled to our unit-mean-square-radius convention, gives for
	\(-3<q<1\)
	\begin{equation}
		\dd\mu_q(\mathbf x)
		=\frac{(R_q^2-|\mathbf x|^2)^{-(q+1)/2}}
		{2\pi R_q^{\,2-q}B\!\left(\frac32,\frac{1-q}{2}\right)}
		\mathbf 1_{\{|\mathbf x|<R_q\}}\dd^3x,
		\qquad
		R_q^2=\frac{4-q}{3}.
		\label{appeq:classified_ball}
	\end{equation}
	The density exponent in the cited result becomes
	\((2-q-d)/2=-(q+1)/2\) for \(d=3\).  Its mean-square radius is
	\(3R_q^2/(4-q)\), so setting this equal to one gives the stated \(R_q\).
	The logarithmic case \(q=0\) is included by continuity.

		Theorem~1 of Ref.~\cite{FrankMatzke2025} instead gives a uniform shell:
	take \((a,b)=(2,q)\) for \(1\le q<2\), and \((a,b)=(q,2)\) for
	\(2<q<4\).  Unit mean-square radius places the shell at \(|\mathbf x|=1\):
	\begin{equation}
		\dd\mu_q=\frac{\dd\Omega}{4\pi},\qquad |\mathbf x|=1.
		\label{appeq:classified_shell}
	\end{equation}
	At \(q=2\), \(I_2=2\) is constant on the constraint surface, so there is
	no distinguished extremum.  The endpoint \(q=4\) is likewise excluded from
	\cref{thm:distribution}: it is the non-unique exceptional case
	\((a,b)=(4,2)\) in the cited classification.

	Finally, express the result in the original shape coordinates.  If
	\(\rho_q(\mathbf x)\) is the normalized number density in the
	order-one coordinates, then the particle-number density in shape space is
	\begin{equation}
		\nu_{q,N}(\mathbf s)
		=N^{5/2}\rho_q(\sqrt N\,\mathbf s).
		\label{appeq:pushforward_density}
	\end{equation}
	For the shell, the same conversion gives
	\[
		\nu_{q,N}(\mathbf s)
		=\frac{N^2}{4\pi}\delta\!\left(|\mathbf s|-N^{-1/2}\right).
	\]
	These are the ball and shell formulas stated in \cref{thm:distribution}.

The cited theorems determine the global continuum extremum.  After the sign
	mapping above (and a harmless additive constant), the
	discrete-to-continuum limit for the full class of
	attractive-repulsive kernels used here follows from
	Ref.~\cite[Theorem~5 and Proposition~6]{CarazzatoPratelliTopaloglu2026}.
	In particular, a sequence of discrete global extrema whose energy per pair
	approaches the continuum value is precompact up to translations, and every
	macroscopic limit is a continuum extremum.  Our center-of-mass condition
	fixes the translation.  Since the continuum profile is unique, every smooth,
	macroscopic density measurement approaches its value for the profile above.
	This statement concerns global extrema only: it neither classifies other
	finite-\(N\) central configurations nor proves that the dynamics selects the
	global extremum.

\section{Details of the Local Stability Calculation}
\label{app:StabilityDetails}

We spell out the calculation used in \cref{thm:stability}.  Throughout this
appendix a bold symbol such as \(\mathbf z\) denotes the whole list
\((\mathbf z_1,\ldots,\mathbf z_N)\), and
\begin{equation}
	\langle\mathbf y,\mathbf z\rangle
	:=\sum_i\mathbf y_i\cdot\mathbf z_i .
\end{equation}
For pair differences we write
\[
\mathbf z_{ij}:=\mathbf z_i-\mathbf z_j,\qquad
\mathbf s_{ij}:=\mathbf s_i-\mathbf s_j,\qquad
s_{ij}:=|\mathbf s_{ij}| .
\]
The central configuration is assumed collision-free, so \(s_{ij}>0\).

\subsection{The Fixed Point and the Allowed Perturbations}
The autonomous shape equations used in the main text have the form
\begin{equation}
	\frac{\dd\mathbf s_i}{\dd\tau}
	=\alpha\left(\mathbf F_i-A\mathbf s_i\right),\qquad
	\frac{\dd\mathbf u_i}{\dd\tau}
	=-\beta\left(\mathbf G_i-A\mathbf u_i\right),
	\label{appeq:shape_equations_recalled}
\end{equation}
where
\begin{equation}
	\mathbf F_i=\sum_{j\ne i}
	|\mathbf s_{ij}|^\beta|\mathbf u_{ij}|^{\alpha-2}\mathbf u_{ij},
	\qquad
	\mathbf G_i=\sum_{j\ne i}
	|\mathbf u_{ij}|^\alpha|\mathbf s_{ij}|^{\beta-2}\mathbf s_{ij},
	\qquad
	A=\sum_i\mathbf s_i\cdot\mathbf F_i .
\end{equation}
At a radial fixed point
\begin{equation}
	\mathbf u_i=\gamma\mathbf s_i,\qquad \gamma\neq0.
\end{equation}
Let
\begin{equation}
	q:=\alpha+\beta,\qquad
	c:=\gamma|\gamma|^{\alpha-2}.
\end{equation}
Then
\begin{equation}
	\mathbf F_i=c\,\mathbf H_q^i,\qquad
	\mathbf H_q^i=\sum_{j\ne i}s_{ij}^{q-2}\mathbf s_{ij}.
\end{equation}
The fixed-point equation is the central-configuration equation
\begin{equation}
	\mathbf H_q^i=S_q\mathbf s_i,\qquad
	S_q=\sum_{i<j}s_{ij}^q. \label{appeq:stability_central_config}
\end{equation}
Consequently \(A=cS_q\) at the fixed point.

Now perturb the fixed point by
\begin{equation}
	\mathbf s_i\mapsto\mathbf s_i+\epsilon\mathbf a_i,\qquad
	\mathbf u_i\mapsto\gamma\mathbf s_i+\epsilon\mathbf b_i .
	\label{appeq:ab_perturbations}
\end{equation}
The position constraints \(\sum_i\mathbf s_i=0\) and
\(\sum_i|\mathbf s_i|^2=1\), and the momentum-centering constraint
\(\sum_i\mathbf u_i=0\), give the linear constraints
\begin{equation}
	\sum_i\mathbf a_i=0,\qquad
	\langle\mathbf s,\mathbf a\rangle=0,\qquad
	\sum_i\mathbf b_i=0 .
	\label{appeq:linear_constraints}
\end{equation}
We also work on a fixed value of the conserved shape energy
\[
H^0_{\alpha,\beta}
=\sum_{i<j}|\mathbf u_{ij}|^\alpha|\mathbf s_{ij}|^\beta .
\]
For one local term in the Hamiltonian, the first variation, using \(\mathbf u_{ij}=\gamma\mathbf s_{ij}\), is
\begin{equation}
	\delta h_{ij}
	=s_{ij}^{q-2}
	\left[
	\beta|\gamma|^\alpha\,\mathbf s_{ij}\cdot\mathbf a_{ij}
	+\alpha c\,\mathbf s_{ij}\cdot\mathbf b_{ij}
	\right].
\end{equation}
The useful pair-sum identity is
\begin{equation}
	\sum_{i<j}s_{ij}^{q-2}\mathbf s_{ij}\cdot\mathbf z_{ij}
	=\sum_i\mathbf z_i\cdot\mathbf H_q^i
	=S_q\langle\mathbf s,\mathbf z\rangle ,
	\label{appeq:pair_sum_identity}
\end{equation}
where the last equality uses the fixed point central configuration condition in \cref{appeq:stability_central_config}.  Applying
\cref{appeq:pair_sum_identity} once with \(\mathbf z=\mathbf a\) and once
with \(\mathbf z=\mathbf b\) gives
\begin{equation}
	\delta H^0_{\alpha,\beta}
	=\beta|\gamma|^\alpha S_q\langle\mathbf s,\mathbf a\rangle
	+\alpha cS_q\langle\mathbf s,\mathbf b\rangle .
\end{equation}
The first term vanishes by \(\langle\mathbf s,\mathbf a\rangle=0\).  Since
\(\alpha cS_q\neq0\) on the branch considered here, fixed energy gives
\begin{equation}
	 {\langle\mathbf s,\mathbf b\rangle=0 .}
	\label{appeq:fixed_energy_constraint}
\end{equation}

\paragraph{The pair operators:}
We next define the three linear operations used below.  For any centered
vector list \(\mathbf z\),
\begin{align}
	(L\mathbf z)_i&:=\sum_{j\ne i}s_{ij}^{q-2}\mathbf z_{ij},
	\label{appeq:L_definition}\\
	(P\mathbf z)_i&:=\sum_{j\ne i}s_{ij}^{q-4}\mathbf s_{ij}
	\left(\mathbf s_{ij}\cdot\mathbf z_{ij}\right),
	\label{appeq:P_definition}\\
	(I\mathbf z)_i&:=\mathbf z_i .
\end{align}
The operator \(L\) keeps the full pair displacement \(\mathbf z_{ij}\).
The operator \(P\) keeps only the component of \(\mathbf z_{ij}\) parallel
to the original pair direction \(\mathbf s_{ij}\). Their symmetry is a pairwise statement.  For any two centered lists
\(\mathbf y,\mathbf z\),
\begin{align}
	\langle\mathbf y,L\mathbf z\rangle
	&=\sum_{i<j}s_{ij}^{q-2}\mathbf y_{ij}\cdot\mathbf z_{ij}
	=\langle L\mathbf y,\mathbf z\rangle,
	\notag\\
	\langle\mathbf y,P\mathbf z\rangle
	&=\sum_{i<j}s_{ij}^{q-4}
	(\mathbf s_{ij}\cdot\mathbf y_{ij})
	(\mathbf s_{ij}\cdot\mathbf z_{ij})
	=\langle P\mathbf y,\mathbf z\rangle .
	\label{appeq:LP_symmetry}
\end{align}
In the first line, the two ordered contributions \((i,j)\) and \((j,i)\)
combine into the single unordered-pair term
\(\mathbf y_{ij}\cdot\mathbf z_{ij}\).  In the second line they combine
into the product of the two components parallel to \(\mathbf s_{ij}\).

The constrained-Hessian operator, with the overall factor \(q\) stripped
off, is
\begin{equation}
	Q_q:=L+(q-2)P-S_qI,
	\label{appeq:Q_definition}
\end{equation}
and the auxiliary operator appearing in the dynamics is
\begin{equation}
	C_q:=L-2P+S_qI.
	\label{appeq:C_definition}
\end{equation}
These operators too are symmetric with respect to
\(\langle\cdot,\cdot\rangle\).  The quadratic form of \(Q_q\) is
\begin{align}
	\mathcal{Q}_q[a]:=  \langle\mathbf a,Q_q\mathbf a\rangle
	={}&\sum_{i<j}s_{ij}^{q-2}
	\left[
	|\mathbf a_{ij}|^2
	+(q-2)\frac{(\mathbf s_{ij}\cdot\mathbf a_{ij})^2}{s_{ij}^2}
	\right]
	-S_q\sum_i|\mathbf a_i|^2 .
	\label{appeq:Q_quadratic_form}
\end{align}
Indeed, under
\(\mathbf s_i\mapsto\mathbf s_i+\epsilon\mathbf a_i\), the pair
difference changes as
\(\mathbf s_{ij}\mapsto\mathbf s_{ij}+\epsilon\mathbf a_{ij}\).
Differentiating a single pair term gives
\begin{equation}
	\delta^2|\mathbf s_{ij}|^q
	=qs_{ij}^{q-2}|\mathbf a_{ij}|^2
	+q(q-2)s_{ij}^{q-4}
	(\mathbf s_{ij}\cdot\mathbf a_{ij})^2,
\end{equation}
and summing this over pairs gives the ambient, straight-line second
variation.  The physical shape variation is constrained to the unit shape
sphere.  Since \(\langle\mathbf s,\mathbf a\rangle=0\), a curve that stays
on that sphere to second order is
\[
\mathbf s_i(\epsilon)
=\frac{\mathbf s_i+\epsilon\mathbf a_i}
{\sqrt{1+\epsilon^2\sum_k|\mathbf a_k|^2}}
=\mathbf s_i+\epsilon\mathbf a_i
-\frac{\epsilon^2}{2}
\left(\sum_k|\mathbf a_k|^2\right)\mathbf s_i
+O(\epsilon^3).
\]
At a central configuration,
\(\nabla_{\mathbf s_i}\Xi_q=qS_q\mathbf s_i\).  The second-order radial
correction in this normalized curve therefore contributes
\[
\sum_i\nabla_{\mathbf s_i}\Xi_q\cdot\mathbf s_i''(0)
=-qS_q\sum_i|\mathbf a_i|^2 .
\]
Equivalently, this is the Lagrange-multiplier contribution from the
constrained functional
\(\Xi_q-(qS_q/2)(\sum_i|\mathbf s_i|^2-1)\) shown in \cref{appeq:Xi_lagrange_fixed}.  Thus, with
\(\delta^2_{\mathrm{shape}}\) denoting the second variation restricted to the
unit shape sphere, we have
\begin{equation}
	 {\delta^2_{\mathrm{shape}}\Xi_q
		=q\mathcal{Q}_q[a] .}
	\label{appeq:Xi_second_variation}
\end{equation}
This is the sign input from the shape potential.  For regular minima with
\(-3<q<0\), and for regular maxima with \(0<q<2\), it gives
\(\mathcal{Q}_q[a]<0\) on every nonzero physical
shape-distortion mode.

The second quadratic form is
\begin{align}
	\langle\mathbf z,C_q\mathbf z\rangle
	={}&\sum_{i<j}s_{ij}^{q-2}
	\left[
	|\mathbf z_{ij}|^2
	-2\frac{(\mathbf s_{ij}\cdot\mathbf z_{ij})^2}{s_{ij}^2}
	\right]
	+S_q\sum_i|\mathbf z_i|^2 .
	\label{appeq:C_quadratic_form}
\end{align}
The positivity of this form is the extra hypothesis in
\cref{thm:stability}; it is not implied just by being an extremum of
\(\Xi_q\).

\paragraph{Linearizing the shape equations:}
The only derivatives needed for the force terms are the direct pair
derivatives
\begin{align}
	\delta_{\mathbf s}
	\left(s_{ij}^{m-2}\mathbf s_{ij}\right)
	&=s_{ij}^{m-2}
	\left[
	\mathbf a_{ij}
	+(m-2)\mathbf s_{ij}
	\frac{\mathbf s_{ij}\cdot\mathbf a_{ij}}{s_{ij}^2}
	\right],
	\notag\\
	\delta_{\mathbf u}
	\left(|\mathbf u_{ij}|^{m-2}\mathbf u_{ij}\right)
	&=|\mathbf u_{ij}|^{m-2}
	\left[
	\mathbf b_{ij}
	+(m-2)\mathbf u_{ij}
	\frac{\mathbf u_{ij}\cdot\mathbf b_{ij}}{|\mathbf u_{ij}|^2}
	\right],
	\label{appeq:basic_vector_derivative}
\end{align}
Using \cref{appeq:basic_vector_derivative} pair by pair at the radial
fixed point gives
\begin{align}
	\delta\mathbf F&=c\beta P\mathbf a+|\gamma|^{\alpha-2}
	\{L+(\alpha-2)P\}\mathbf b, \notag\\
	\delta\mathbf G&= c \alpha P\mathbf b+|\gamma|^\alpha
	\{L+(\beta-2)P\}\mathbf a .
	\label{appeq:FG_variations}
\end{align}

We also need the variation of the scalar \(A=\langle\mathbf s,\mathbf F\rangle\),
\[
\delta A=\langle\mathbf a,\mathbf F\rangle
+\langle\mathbf s,\delta\mathbf F\rangle .
\]
At the fixed point \(\mathbf F=cS_q\mathbf s\), and
\[
L\mathbf s=P\mathbf s=S_q\mathbf s.
\]
Using \cref{appeq:FG_variations} and the symmetry in
\cref{appeq:LP_symmetry}, this gives
\begin{align*}
	\delta A
	={}&cS_q\langle\mathbf s,\mathbf a\rangle
	+c\beta\langle\mathbf s,P\mathbf a\rangle
	+|\gamma|^{\alpha-2}
	\langle\mathbf s,\{L+(\alpha-2)P\}\mathbf b\rangle\\
	={}&cS_q(1+\beta)\langle\mathbf s,\mathbf a\rangle
	+|\gamma|^{\alpha-2}(\alpha-1)S_q
	\langle\mathbf s,\mathbf b\rangle .
\end{align*}
The first term vanishes by the unit-shape constraint in
\cref{appeq:linear_constraints}; the second vanishes by the fixed-energy
constraint in \cref{appeq:fixed_energy_constraint}.  Therefore
\begin{equation}
	\delta A=0.
	\label{appeq:delta_A_zero}
\end{equation}

It is helpful to divide out the fixed number \(\gamma\) from the momentum
perturbation.  Define
\begin{equation}
	\mathbf z:=\frac{\mathbf b}{\gamma},\qquad
	\mathbf w:=\mathbf z-\mathbf a=\frac{\mathbf b}{\gamma}-\mathbf a .
	\label{appeq:w_definition}
\end{equation}
Thus \(\mathbf w=0\) means that the perturbation preserves the radial
relation \(\mathbf u=\gamma\mathbf s\) to first order.

Substituting \cref{appeq:FG_variations,appeq:delta_A_zero} into the
linearization of \cref{appeq:shape_equations_recalled}, and then replacing
\(\mathbf z\) by \(\mathbf a+\mathbf w\), gives
\begin{align}
	\frac{\dd\mathbf a}{\dd\tau}
	&=\alpha c\left[
	Q_q\mathbf a+\{L+(\alpha-2)P\}\mathbf w
	\right],
	\label{appeq:linear_a_general}\\
	\frac{\dd\mathbf w}{\dd\tau}
	&=-c\left[
	qQ_q\mathbf a+
	\{\alpha(L+(q-2)P)-\beta S_qI\}\mathbf w
	\right].
	\label{appeq:linear_w_general}
\end{align}
These two equations are the fixed-energy linearization for general
\(\alpha,\beta\).  Notice that the fixed points depend only on
\(q=\alpha+\beta\), but the linearized dynamics still remembers \(\alpha\) and
\(\beta\) separately.

\paragraph{Reduction on the EdS branch:}
On the EdS branch,
\begin{equation}
	\alpha=2q,\qquad \beta=-q .
\end{equation}
The expanding fixed point has \(qc>0\).  We therefore use the rescaled time
\begin{equation}
	\rho:=qc\,\tau
	=q\gamma|\gamma|^{2(q-1)}\tau .
	\label{appeq:rho_time}
\end{equation}
In this appendix \(\rho\) is a time variable, unrelated to the density
\(\rho(\mathbf x)\) used in \cref{app:LargeNDistributions}.  Because
\(qc>0\) on the expanding branch, increasing \(\rho\) is the same time
orientation as increasing \(\tau\).

With this time variable, \cref{appeq:linear_a_general,appeq:linear_w_general}
become
\begin{align}
	\frac{\dd\mathbf a}{\dd\rho}
	&=2\left[
	Q_q\mathbf a+\{L+(2q-2)P\}\mathbf w
	\right],
	\label{appeq:linear_a_EdS}\\
	\frac{\dd\mathbf w}{\dd\rho}
	&=-Q_q\mathbf a
	-\left[2\{L+(q-2)P\}+S_qI\right]\mathbf w .
	\label{appeq:linear_w_EdS}
\end{align}
Now define one final perturbation variable
\begin{equation}
	\mathbf v:=\mathbf a+2\mathbf w .
	\label{appeq:v_definition}
\end{equation}
Adding \cref{appeq:linear_a_EdS} to twice
\cref{appeq:linear_w_EdS} cancels the \(Q_q\mathbf a\) terms and gives
\begin{equation}
	\frac{\dd\mathbf v}{\dd\rho}
	=-2(L-2P+S_qI)\mathbf w
	=-2C_q\mathbf w .
	\label{appeq:v_first_order}
\end{equation}
Also, since \(\mathbf a=\mathbf v-2\mathbf w\),
\cref{appeq:linear_w_EdS} becomes
\begin{equation}
	\frac{\dd\mathbf w}{\dd\rho}
	=-Q_q\mathbf v-3S_q\mathbf w .
	\label{appeq:w_first_order}
\end{equation}
Finally differentiate \cref{appeq:v_first_order} once and substitute
\cref{appeq:w_first_order}:
\begin{align}
	\frac{\dd^2\mathbf v}{\dd\rho^2}
	&=-2C_q\frac{\dd\mathbf w}{\dd\rho}\notag\\
	&=2C_qQ_q\mathbf v+6S_qC_q\mathbf w\notag\\
	&=2C_qQ_q\mathbf v
	-3S_q\frac{\dd\mathbf v}{\dd\rho}.
\end{align}
Thus the stability problem reduces to
\begin{equation}
	 {
		\frac{\dd^2\mathbf v}{\dd\rho^2}
		+3S_q\frac{\dd\mathbf v}{\dd\rho}
		-2C_qQ_q\mathbf v=0.}
	\label{appeq:damped_equation}
\end{equation}
The middle term is ordinary friction, and \(S_q>0\).

\paragraph{The sign argument:}
Assume \(C_q\) is positive on the physical perturbation space.  Since
\(C_q\) and \(Q_q\) are symmetric, \(C_qQ_q\) is similar to the symmetric
matrix \(C_q^{1/2}Q_qC_q^{1/2}\).  Therefore its eigenvalues are real, and
its signs are the same as those of \(Q_q\).

For the regular extrema covered by \cref{thm:stability}, \(Q_q\) is
negative on physical shape distortions.  Hence each physical mode obeys
\begin{equation}
	C_qQ_q\mathbf v=-\kappa\mathbf v,\qquad \kappa>0 .
\end{equation}
For a mode \(\mathbf v\propto e^{\lambda\rho}\),
\cref{appeq:damped_equation} gives
\begin{equation}
	\lambda^2+3S_q\lambda+2\kappa=0 .
\end{equation}
If the roots are real, both are negative; if they are complex, their real
part is \(-3S_q/2<0\).  The expanding branch is therefore linearly
asymptotically stable.

If instead \(Q_q\) has a positive physical direction and \(C_q\) is still
positive, then \(C_qQ_q\mathbf v=+\kappa\mathbf v\).  The characteristic
equation becomes
\[
\lambda^2+3S_q\lambda-2\kappa=0,
\]
whose roots have negative product.  One root is positive, so that mode is
unstable.  This is the instability statement for saddle-type central configurations.

\subsection{Unequal-mass Stability Criterion}
\label{app:UnequalMassStability}
For \cref{appeq:mass_hamiltonian}, the stability calculation above has a
direct mass-weighted analogue.  We record the required replacements rather
than repeat the linearization.  The natural inner product on perturbations is
\begin{equation}
	\langle\mathbf y,\mathbf z\rangle_m
	:=\sum_i m_i\mathbf y_i\cdot\mathbf z_i .
	\label{appeq:mass_inner_product}
\end{equation}
Accordingly, the physical linearized constraints become
\begin{equation}
	\sum_i m_i\mathbf a_i=0,\qquad
	\langle\mathbf s,\mathbf a\rangle_m=0,\qquad
	\sum_i m_i\mathbf b_i=0,\qquad
	\langle\mathbf s,\mathbf b\rangle_m=0,
	\label{appeq:mass_linear_constraints}
\end{equation}
where the last condition again follows by restricting to a fixed-energy
surface.

The pair operators are
\begin{align}
	(L_m\mathbf z)_i
	&:=\sum_{j\neq i}m_j s_{ij}^{q-2}\mathbf z_{ij},
	\notag\\
	(P_m\mathbf z)_i
	&:=\sum_{j\neq i}m_j s_{ij}^{q-4}\mathbf s_{ij}
	(\mathbf s_{ij}\cdot\mathbf z_{ij}).
	\label{appeq:mass_LP}
\end{align}
Both are self-adjoint with respect to
\cref{appeq:mass_inner_product}; for example,
\begin{equation}
	\langle\mathbf y,L_m\mathbf z\rangle_m
	=\sum_{i<j}m_im_j s_{ij}^{q-2}
	\mathbf y_{ij}\cdot\mathbf z_{ij}.
\end{equation}
The two operators controlling the stability problem are therefore
\begin{equation}
	Q_{q,m}:=L_m+(q-2)P_m-S_{q,m}I,\qquad
	C_{q,m}:=L_m-2P_m+S_{q,m}I .
	\label{appeq:mass_CQ}
\end{equation}
For \(q\neq0\), the constrained Hessian of the weighted shape potential is
\begin{equation}
	\delta^2_{\rm shape}\Xi_{q,m}
	=q\langle\mathbf a,Q_{q,m}\mathbf a\rangle_m .
	\label{appeq:mass_second_variation}
\end{equation}

On the EdS branch, the definitions
\(\mathbf w=\mathbf b/\gamma-\mathbf a\),
\(\mathbf v=\mathbf a+2\mathbf w\), and
\(\rho=qc\tau\) reduce the unequal-mass linearization to
\begin{equation}
	\frac{\dd^2\mathbf v}{\dd\rho^2}
	+3S_{q,m}\frac{\dd\mathbf v}{\dd\rho}
	-2C_{q,m}Q_{q,m}\mathbf v=0 .
	\label{appeq:mass_damped_equation}
\end{equation}
The sign argument is unchanged: if \(Q_{q,m}<0\) and \(C_{q,m}>0\) on
the physical perturbation space after quotienting common rotations, the
expanding fixed point is linearly asymptotically stable.  In particular,
for the distinguished model \(q=-1\), a nondegenerate minimum of
\(\Xi_{-1,m}\) supplies the required negative sign of \(Q_{-1,m}\).
Positivity of \(C_{-1,m}\), however, remains an additional mass-dependent
hypothesis, just as in the equal-mass analysis.

This weighted formulation can be written in an ordinary Euclidean inner
product by setting
\(\widehat{\mathbf z}_i=\sqrt{m_i}\,\mathbf z_i\).  The conjugated
operators are then symmetric, so the spectral sign proof proceeds exactly
as above.  This change of variables does not remove the masses from the
operators and therefore does not reduce a general unequal-mass
configuration to an equal-mass one.  Finally, the criterion concerns the
effective point-cluster degrees of freedom.  Stability of a microscopic
clustered state with nontrivial internal motion additionally requires
control of the internal and finite-size modes, which is beyond the
fixed-point analysis considered here.

\subsection{Coordinate Check for the Equilateral \texorpdfstring{\(N=3\)}{N=3} Proposition}
\label{app:TriangleCheck}
This subsection gives the coordinate calculation behind \cref{prop:equilateral_stability}. 
Choose
\begin{equation}
	\mathbf s_1=\left(\frac1{\sqrt3},0,0\right),\quad
	\mathbf s_2=\left(-\frac1{2\sqrt3},\frac12,0\right),\quad
	\mathbf s_3=\left(-\frac1{2\sqrt3},-\frac12,0\right).
\end{equation}
Then \(\sum_i\mathbf s_i=0\), \(\sum_i|\mathbf s_i|^2=1\), every side has
length one, and \(S_q=3\). This satisfies the central configuration equation $\mathbf{H}^i_q = 3 \mathbf{s}_i$.  Because all side lengths are one, the weights
in \(L\) are all one.  Hence, for any centered perturbation,
\[
(L\mathbf z)_i=\sum_{j\ne i}(\mathbf z_i-\mathbf z_j)
=3\mathbf z_i-\sum_j\mathbf z_j=3\mathbf z_i .
\]
The tangent space splits into three common rotations and two genuine shape
distortions.  Any centered in-plane perturbation of the equilateral
triangle can be written as \(\mathbf z_i=M\mathbf s_i\) for a
\(2\times2\) matrix \(M\).  Decomposing \(M\) into its trace,
antisymmetric, and symmetric traceless parts separates the uniform scale
direction, the in-plane rotation, and the two genuine shape distortions.
The scale direction is removed by the unit-shape constraint
\(\langle\mathbf s,\mathbf z\rangle=0\), and the rotation is a common
rotational zero mode.  Thus the remaining in-plane shape modes have
symmetric traceless \(M\).  For a common infinitesimal rotation,
\begin{equation}
	\mathbf a_i = \Omega \mathbf s_i,\qquad \mathbf b_i=\Omega  \mathbf u_i =\gamma\,\Omega \mathbf s_i, \qquad \Omega = \begin{pmatrix}
		0 & -1 & 0 \\
		1 & 0 & 0\\
		0 & 0 & 0 
	\end{pmatrix}.
\end{equation}
\(\mathbf s_{ij}\cdot\mathbf z_{ij}=0\) for every pair and therefore
\(P\mathbf z=0\).  The two genuine shape modes can be represented as
\(\mathbf z_i=M\mathbf s_i\), where \(M\) is a symmetric, traceless \(2\times2\) matrix acting in the triangle plane.  Explicitly, for such a
mode one may take
\[
\mathbf a_i=\xi\,M\mathbf s_i,\qquad
\mathbf b_i=\chi\,M\mathbf s_i,
\qquad i=1,2,3,
\]
where \(\xi\) and \(\chi\) are scalar amplitudes for the position and
momentum parts of that one mode.  There are two independent choices of
\(M\) with the following representation on the full 3d coordinates,
\begin{equation}
	M_1 = \begin{pmatrix}
		1 & 0 & 0\\
		0 & -1 & 0\\
		0 & 0 & 0
	\end{pmatrix},\qquad M_2 = \begin{pmatrix}
		0 & 1 & 0\\
		1 & 0 & 0 \\
		0 & 0 & 0
	\end{pmatrix}
\end{equation}
The symmetric traceless condition makes
\(\sum_i\mathbf a_i=\sum_i\mathbf b_i=0\) and
\(\langle\mathbf s,\mathbf a\rangle
=\langle\mathbf s,\mathbf b\rangle=0\).  Substituting the two basis choices
for \(M\) in \cref{appeq:P_definition} gives 
\(P\mathbf z=(3/2)\mathbf z\).  Therefore
\begin{equation}
	C_q=6I-2P
	=
	\begin{cases}
		6I,&\text{rotations},\\
		3I,&\text{shape distortions},
	\end{cases}
\end{equation}
so \(C_q\) is strictly positive on the full centered tangent space; the common rotational zero modes are quotiented only when interpreting the linearized dynamics.
Similarly
\begin{equation}
	Q_q=L+(q-2)P-S_qI=(q-2)P.
\end{equation}
Thus \(Q_q=0\) on common rotations and
\(Q_q=\frac32(q-2)I\) on the two genuine shape modes.  For every \(q<2\),
the physical Hessian sign is the stable sign \(Q_q<0\).

For the distinguished model \(q=-1\), the distortion modes have
\(C_{-1}=3I\), \(Q_{-1}=-9I/2\), and \(S_{-1}=3\).  The characteristic
equation is
\begin{equation}
	\lambda^2+9\lambda+27=0,
\end{equation}
with roots
\begin{equation}
	\lambda_\pm=-\frac92\pm\frac{3\sqrt3}{2}\,i .
\end{equation}
The nonzero rotational companions have \(\lambda=-9\).  After quotienting
the zero common rotations, all physical eigenvalues have negative real
part.

\subsection{The Collinear \texorpdfstring{\(N=3\)}{N=3} Instability}
\label{app:CollinearCheck}
The same calculation gives a simple analytic check that the collinear
three-body central configuration is unstable.  Take
\begin{equation}
	\mathbf s_1=\left(-\frac1{\sqrt2},0,0\right),\qquad
	\mathbf s_2=(0,0,0),\qquad
	\mathbf s_3=\left(\frac1{\sqrt2},0,0\right).
\end{equation}
The constraints \(\sum_i\mathbf s_i=0\) and
\(\sum_i|\mathbf s_i|^2=1\) are immediate.  The pair separations are
\(s_{12}=s_{23}=1/\sqrt2\) and \(s_{13}=\sqrt2\).  Therefore
\begin{equation}
	S_q=\sum_{i<j}s_{ij}^{q}
	=2^{1-q/2}+2^{q/2}.
\end{equation}
For instance, for particle \(1\),
\[
\mathbf H_q^1
=-\left[
\left(\frac1{\sqrt2}\right)^{q-1}
+\left(\sqrt2\right)^{q-1}
\right]\hat{\mathbf x}
=S_q\mathbf s_1,
\]
and the equations for particles \(2\) and \(3\) follow by symmetry.
Thus this collinear configuration is a \(q\)-central configuration for the
full branch, not only for \(q=-1\).

Now consider the transverse perturbation
\begin{equation}
	\mathbf z_1=(0,1,0),\qquad
	\mathbf z_2=(0,-2,0),\qquad
	\mathbf z_3=(0,1,0).
\end{equation}
It is centered, satisfies \(\langle\mathbf s,\mathbf z\rangle=0\), and is
orthogonal to the common-rotation zero modes.  Since all
\(\mathbf s_{ij}\) are parallel to the \(x\)-axis while all
\(\mathbf z_{ij}\) are transverse, \(P\mathbf z=0\).  The weighted
Laplacian \(L\) has nearest-neighbor weights
\[
s_{12}^{q-2}=s_{23}^{q-2}=2^{1-q/2}.
\]
The long edge does not contribute because \(\mathbf z_1=\mathbf z_3\).
Thus
\begin{equation}
	L\mathbf z=3\,2^{1-q/2}\mathbf z .
\end{equation}
Consequently
\begin{equation}
	Q_q\mathbf z
	=\left(2^{2-q/2}-2^{q/2}\right)\mathbf z,\qquad
	C_q\mathbf z
	=\left(2^{3-q/2}+2^{q/2}\right)\mathbf z .
\end{equation}
For \(q<2\), the \(Q_q\) eigenvalue is positive.  This is the saddle sign:
the collinear central configuration has a physical direction opposite to
the stable sign of the regular extrema.  In the expanding time \(\rho\),
the mode obeys
\begin{equation}
	\lambda^2+3S_q\lambda
	-2\left(2^{3-q/2}+2^{q/2}\right)
	\left(2^{2-q/2}-2^{q/2}\right)=0.
\end{equation}
For \(q<2\) the constant term is negative, so one root is positive:
\begin{equation}
	\lambda_+
	=\frac{-3S_q+
		\sqrt{9S_q^2
			+8\left(2^{3-q/2}+2^{q/2}\right)
			\left(2^{2-q/2}-2^{q/2}\right)}}{2}>0 .
\end{equation}
Thus the expanding collinear \(N=3\) central configuration is linearly
unstable throughout the nonzero EdS range \(q<2\) covered by the stability
discussion.  For the distinguished model \(q=-1\), this reduces to
\[
S_{-1}=\frac{5\sqrt2}{2},\qquad
Q_{-1}\mathbf z=\frac{7\sqrt2}{2}\mathbf z,\qquad
C_{-1}\mathbf z=\frac{17\sqrt2}{2}\mathbf z,
\]
and hence
\[
\lambda^2+\frac{15\sqrt2}{2}\lambda-119=0,\qquad
\lambda_+=-\frac{15\sqrt2}{4}
+\frac12\sqrt{\frac{1177}{2}}>0 .
\]
This is the analytic three-body counterpart of the collinear instability
seen numerically in \cref{fig:Eig_spec}.

\section{Numerical Implementation for the Distinguished Model}
\label{app:NumericalImplementation}
This appendix summarizes the numerical procedure used for the finite-\(N\)
results in the distinguished model.  We spell out how \cref{eq:shape_coordinates_eom} was evaluated,
linearized, and projected onto the physical perturbation space.

Throughout this appendix \((\alpha,\beta)=(-2,1)\), so the reduced time is
\(\tau\), \(d\tau=dt/R^{3/2}\).  We use
\[
\mathbf s_{ij}:=\mathbf s_i-\mathbf s_j,\qquad
\mathbf u_{ij}:=\mathbf u_i-\mathbf u_j,\qquad
s_{ij}:=|\mathbf s_{ij}|.
\]
The reduced autonomous equations used in the numerics are
\begin{align}
	\dot{\mathbf s}_i&=-2(\mathbf F_i-A\mathbf s_i),\notag\\
	\dot{\mathbf u}_i&=A\mathbf u_i-\mathbf G_i,
	\label{appeq:numerical_reduced_flow}
\end{align}
where a dot means \(d/d\tau\), and
\begin{align}
	\mathbf F_i&=\sum_{j\neq i}|\mathbf u_{ij}|^{-4}
	\mathbf u_{ij}\,s_{ij},\notag\\
	\mathbf G_i&=\sum_{j\neq i}|\mathbf u_{ij}|^{-2}
	\frac{\mathbf s_{ij}}{s_{ij}},\notag\\
	A&=\sum_i\mathbf s_i\cdot\mathbf F_i .
	\label{appeq:numerical_FGA}
\end{align}
These are simply \cref{eq:shape_coordinates_eom} specialized to the
distinguished model.

\subsection{Fixed Points used for the Spectra}
A candidate central configuration is first centered and normalized:
\begin{equation}
	\sum_i\mathbf s_i=0,\qquad
	\sum_i|\mathbf s_i|^2=1.
	\label{appeq:numerical_shape_constraints}
\end{equation}
The radial fixed point is then formed by setting
\begin{equation}
	\mathbf u_i=\gamma\,\mathbf s_i .
	\label{appeq:numerical_radial_fixed_point}
\end{equation}
In the spectra shown in \cref{fig:Eig_spec}, \(\gamma=-1\) is the expanding
branch and \(\gamma=+1\) is the contracting branch.  The shape itself must
satisfy
\begin{equation}
	\mathbf H_{-1}^i
	=\sum_{j\neq i}\frac{\mathbf s_{ij}}{s_{ij}^{3}}
	=S_{-1}\mathbf s_i,\qquad
	S_{-1}=\sum_{i<j}\frac1{s_{ij}} .
	\label{appeq:numerical_central_config}
\end{equation}
This condition is checked directly before linearizing the flow.  Equivalently,
if \(\mathbf X_*=(\mathbf s_1,\ldots,\mathbf s_N,\mathbf u_1,\ldots,\mathbf u_N)\)
and \(\mathcal V(\mathbf X)\) denotes the right-hand side of
\cref{appeq:numerical_reduced_flow}, the fixed-point residual is
\begin{equation}
	r_\infty=\max_a |(\mathcal V(\mathbf X_*))_a| .
	\label{appeq:numerical_residual}
\end{equation}
The spectra are computed only after this residual is ; in practice the
tolerance used for accepting a fixed point was \(r_\infty\lesssim10^{-8}\).

\subsection{Finite-difference Jacobian}
The full state vector has \(6N\) real components.  The unprojected Jacobian
is
\begin{equation}
	J=\left.\frac{\partial\mathcal V}{\partial\mathbf X}\right|_{\mathbf X_*}
	\in\mathbb R^{6N\times6N}.
\end{equation}
We evaluate this derivative by symmetric finite differences.  For the
coordinate unit vector \(\mathbf e_k\),
\begin{equation}
	J_{\cdot k}\simeq
	\frac{\mathcal V(\mathbf X_*+\varepsilon\mathbf e_k)
		-\mathcal V(\mathbf X_*-\varepsilon\mathbf e_k)}
	{2\varepsilon},
	\qquad \varepsilon=10^{-7}.
	\label{appeq:numerical_finite_difference}
\end{equation}

\subsection{Projection to the Physical Tangent Space}
The eigenvalues of the full \(J\) are not the physical stability spectrum,
because \(J\) acts on directions that either violate the constraints or
move along redundant symmetry directions.  We remove those directions before
diagonalizing.

Let
\[
\langle\mathbf a,\mathbf b\rangle
:=\sum_i\mathbf a_i\cdot\mathbf b_i .
\]
The constraint normals and neutral generators used in the projection are:
\begin{align}
	&\delta\mathbf s_i=\mathbf e_a,\quad
	\delta\mathbf u_i=0,
	&&a=1,2,3, \notag\\
	&\delta\mathbf s_i=0,\quad
	\delta\mathbf u_i=\mathbf e_a,
	&&a=1,2,3, \notag\\
	&\delta\mathbf s_i=\mathbf s_i,\quad
	\delta\mathbf u_i=0, \notag\\
	&\delta\mathbf s_i=0,\quad
	\delta\mathbf u_i=\mathbf s_i, \notag\\
	&\delta\mathbf s_i=\boldsymbol\omega_a\times\mathbf s_i,\quad
	\delta\mathbf u_i=\boldsymbol\omega_a\times\mathbf u_i,
	&&a=1,2,3 .
	\label{appeq:numerical_removed_vectors}
\end{align}
The first two lines enforce the centering constraints in \(\mathbf s\) and
\(\mathbf u\).  The third line is the normal to
\(\sum_i|\mathbf s_i|^2=1\).  The fourth line is the infinitesimal displacement
along the one-parameter family \(\mathbf u_i=\gamma\mathbf s_i\), which contributes one additional neutral direction.  The last
line removes common spatial rotations.  For a generic non-collinear shape
these give eleven independent directions.  If a special configuration has a
dependent or zero generator, it is discarded during orthonormalization.

After Gram-Schmidt orthonormalization of the removed vectors, let
\(Q_{\rm rem}\) be the matrix whose columns are the independent
orthonormal directions to be removed.  The orthogonal projector onto the
remaining tangent space is
\begin{equation}
	\Pi_{\rm phys}=I-Q_{\rm rem}Q_{\rm rem}^T .
\end{equation}
An orthonormal basis \(Q_{\rm phys}\) for the image of
\(\Pi_{\rm phys}\) is obtained by keeping the nonzero singular-vector
directions of this projector.  The matrix whose spectrum is plotted is then
\begin{equation}
	J_{\rm phys}=Q_{\rm phys}^T J Q_{\rm phys}.
	\label{appeq:numerical_projected_jacobian}
\end{equation}
A fixed point is linearly asymptotically stable in the reduced dynamics
when every eigenvalue of \(J_{\rm phys}\) has negative real part.  A single
eigenvalue with positive real part is an instability.

	\section{Calculation of Boltzmann Entropy}
	\label{app:Entropy}
	Consider the general model \(H_{(\alpha,\beta)}\), with
	\(\alpha\neq0\).  We fix \(N\), the energy \(E\), the conserved centers
	\(\mathbf X_{\rm C}\) and \(\mathbf P_{\rm C}\), and the intrinsic angular
	momentum
	\(\mathbf J=\sum_i(\mathbf x_i-\mathbf X_{\rm C})\times
	(\mathbf p_i-\mathbf P_{\rm C})\).  The fixed-charge phase-space volume is, up to an overall
	normalization,
	\begin{equation}
		\Omega_{\mathcal M}=\mathcal C\int\prod_i\dd^3x_i\,\dd^3p_i\,
		\delta(H_{(\alpha,\beta)}-E)
		\delta^{(3)}(\mathbf X_{\rm C}-\mathbf X_0)
		\delta^{(3)}(\mathbf P_{\rm C}-\mathbf P_0)
		\delta^{(3)}(\mathbf J-\mathbf J_0).
		\label{appeq:entropy_fixed_charge_volume}
	\end{equation}
	We are interested in the volume of
	\(\mathcal M\) where the scale takes values
	\(\bar R\leq R<\bar R+\Delta R\) with $(\Delta R)<<R$. We can write this as $\Omega_{\bar R\leq R<\bar R+\Delta R} \approx \Omega_{\bar {R}} \Delta R$ where
	\begin{equation}
		\Omega_{\bar {R}} = \mathcal C\int\prod_i\dd^3x_i\,\dd^3p_i\,
		\delta(R - \bar{R})\delta(H_{(\alpha,\beta)}-E)
		\delta^{(3)}(\mathbf X_{\rm C}-\mathbf X_0)
		\delta^{(3)}(\mathbf P_{\rm C}-\mathbf P_0)
		\delta^{(3)}(\mathbf J-\mathbf J_0). \label{appeq:entropy_fixed_charge_volume_scale}
	\end{equation}
	To compute \cref{appeq:entropy_fixed_charge_volume_scale}, we now use the scale-shape reduction introduced in
	\cref{eq:shape_coordinates}, 
	\begin{equation}
		\mathbf{s}_i = \frac{\mathbf{x}_i-\mathbf{X}_{\rm C}}{R}, \quad \mathbf{u}_i = R^{\beta/\alpha}\left(\mathbf{p}_i-\mathbf{P}_{\rm C}\right),    \quad    R = \sqrt{\sum_{i=1}^N |\mathbf{x}_i-\mathbf{X}_{\rm C}|^2}. 
	\end{equation}
	satisfying the constraints in \cref{eq:shape_coordinate_constraints}. The
	Hamiltonian \(H^0_{(\alpha,\beta)}\) becomes scale independent.   
	\cref{appeq:entropy_fixed_charge_volume_scale} then becomes
	\begin{align}
		\Omega_{\bar R}
		={}&\mathcal C\,\bar R^{d\nu-1}
		\int\prod_{i=1}^N\dd^3s_i\,\dd^3u_i\,
		\delta^{(3)}\!\left(\sum_i\mathbf s_i\right)
		\delta\!\left(\sum_i|\mathbf s_i|^2-1\right)\notag\\
		&\times
		\delta^{(3)}\!\left(\sum_i\mathbf u_i\right)
		\delta\!\left(H^0_{(\alpha,\beta)}(\mathbf s,\mathbf u)-E\right)
		\delta^{(3)}\!\left(\bar R^\nu\mathbf j-\mathbf J_0\right),
		\label{appeq:entropy_reduced_fixed_scale}
	\end{align}
	where \(d=3N-3\) and 
	\begin{align}
		&\prod_{i=1}^N\dd^3x_i\,\dd^3p_i\,
		\delta^{(3)}(\mathbf X_{\rm C}-\mathbf X_0)
		\delta^{(3)}(\mathbf P_{\rm C}-\mathbf P_0)\notag\\
		&\qquad\propto R^{d\nu-1}\dd R
		\prod_{i=1}^N\dd^3s_i\,\dd^3u_i\,
		\delta^{(3)}\!\left(\sum_i\mathbf s_i\right)
		\delta\!\left(\sum_i|\mathbf s_i|^2-1\right)
		\delta^{(3)}\!\left(\sum_i\mathbf u_i\right),\notag\\
		&\qquad
		\mathbf J=R^\nu\mathbf j,\quad
		\mathbf j=\sum_i\mathbf s_i\times\mathbf u_i,\quad \nu=1-\frac{\beta}{\alpha} .
		\label{appeq:entropy_general_scaling}
	\end{align}
	The proportionality constant in this measure is absorbed into
	\(\mathcal C\). The remaining angular-momentum delta function obeys
	\begin{equation}
		\delta^{(3)}\!\left(\bar R^\nu\mathbf j-\mathbf J_0\right)
		=\bar R^{-3\nu}
		\delta^{(3)}\!\left(
		\mathbf j-\frac{\mathbf J_0}{\bar R^\nu}\right).
		\label{appeq:entropy_angular_delta}
	\end{equation}
	It is useful to introduce
	\begin{equation}
		a_{\alpha,\beta}
		=(d-3)\nu-1
		=3(N-2)\left(1-\frac{\beta}{\alpha}\right)-1
		\label{appeq:entropy_general_exponent}
	\end{equation}
	and the residual shape-momentum integral
	\begin{align}
		Z_{\alpha,\beta}(E,\boldsymbol\ell)
		:={}&\int\prod_{i=1}^N\dd^3s_i\,\dd^3u_i\,
		\delta^{(3)}\!\left(\sum_i\mathbf s_i\right)
		\delta\!\left(\sum_i|\mathbf s_i|^2-1\right)\notag\\
		&\times
		\delta^{(3)}\!\left(\sum_i\mathbf u_i\right)
		\delta\!\left(H^0_{(\alpha,\beta)}(\mathbf s,\mathbf u)-E\right)
		\delta^{(3)}(\mathbf j-\boldsymbol\ell).
		\label{appeq:entropy_formal_Z}
	\end{align}
	Then
	\begin{equation}
		\Omega_{\bar R}
		=\mathcal C\,\bar R^{a_{\alpha,\beta}}
		Z_{\alpha,\beta}\!\left(
		E,\frac{\mathbf J_0}{\bar R^\nu}\right).
		\label{appeq:entropy_fixed_scale_Z}
	\end{equation}
	For \(\nu>0\), fixed physical angular momentum becomes negligible in the
	reduced variables when \(\bar R\) is large.  Formally taking this limit gives
	\begin{equation}
		\Omega_{\bar R\leq R<\bar R+\Delta R}
		\simeq
		\mathcal C \bar R^{a_{\alpha,\beta}} \Delta R\,
		Z_{\alpha,\beta}(E,\mathbf0).
		\label{appeq:entropy_formal_factorization}
	\end{equation}
	This separates the scale factor \(\bar R^{a_{\alpha,\beta}}\Delta R\)
	from the residual integral \(Z_{\alpha,\beta}\).  The latter is nevertheless divergent for the
	many-body \(\alpha<0\) models of interest: reduced momentum space is
	noncompact, and relative momentum boosts can become arbitrarily large
	without changing the fixed energy.
	
	A finite Boltzmann volume is obtained by refining the macrostate with the
	additional reduced momentum observable
	\begin{equation}
		P_u(\mathbf u)
		=\frac12\sum_{i<j}|\mathbf u_i-\mathbf u_j|^2
		=\frac N2\sum_i|\mathbf u_i|^2.
		\label{appeq:entropy_Pu}
	\end{equation}
	This is the three-dimensional, scale-reduced version of the
	momentum-complexity macro-observable used in
	Ref.~\cite{BabbarSadkiPrakashSondhi_Fractons_2025}.  If \(P_\pi\) denotes
	the same expression formed from the physical centered momenta, then
	\(P_\pi=R^{-2\beta/\alpha}P_u\). For the macrocell \(p_u\leq P_u<p_u+\Delta p_u\), replace \cref{appeq:entropy_formal_Z} by $Z_{\alpha,\beta}(E,\boldsymbol\ell;p_u\leq P_u<p_u+\Delta p_u) \approx Z_{\alpha,\beta}(E,\boldsymbol\ell;p_u) \Delta p_u $ where
	\begin{align}
		Z_{\alpha,\beta}(E,\boldsymbol\ell;p_u)
		:={}&\int\prod_{i=1}^N\dd^3s_i\,\dd^3u_i\,
		\delta^{(3)}\!\left(\sum_i\mathbf s_i\right)
		\delta\!\left(\sum_i|\mathbf s_i|^2-1\right)
		\delta^{(3)}\!\left(\sum_i\mathbf u_i\right)\notag\\
		&\times
		\delta\!\left(H^0_{(\alpha,\beta)}(\mathbf s,\mathbf u)-E\right)
		\delta^{(3)}(\mathbf j-\boldsymbol\ell)
		\delta\!\left(P_u(\mathbf u)-p_u\right).
		\label{appeq:entropy_finite_Z}
	\end{align}
	Because a finite \(P_u\) bin bounds all the reduced momenta, this integral
	is finite for regular constraint values. The phase space volume we use to compute the Boltzmann entropy is 
	\begin{equation}
		\Omega_{\bar R\leq R<\bar R+\Delta R;~p_u\leq P_u<p_u+\Delta p_u}
		\simeq
		\mathcal C \bar R^{a_{\alpha,\beta}} Z_{\alpha,\beta}(E,\mathbf J_0/\bar R^\nu;p_u) \left(\Delta R \Delta p_u\right)
		\label{appeq:entropy_finite_factorization}
	\end{equation}
	The Boltzmann entropy of this macrocell is therefore
	\begin{equation}
		S_B(\bar R,p_u)
		=k_B\log\!\left[
		\frac{\mathcal C\,\bar R^{a_{\alpha,\beta}}
			Z_{\alpha,\beta}(E,\mathbf J_0/\bar R^\nu;p_u)\,
			\Delta R\,\Delta p_u}
		{\Omega_*}\right],
		\label{appeq:entropy_general}
	\end{equation}
	where \(\Omega_*\) is a reference phase-space volume.  The dependence on
	the fixed energy and reduced angular momentum is left implicit in the
	notation \(S_B(\bar R,p_u)\).  For \(\nu>0\), at fixed conserved charges
	and sufficiently large \(\bar R\), the residual \(\bar R\)-dependence of
	\(Z_{\alpha,\beta}(E,\mathbf J_0/\bar R^\nu;p_u)\) drops out.  
   \[Z_{\alpha,\beta}(E,\mathbf J_0/\bar R^\nu;p_u) \approx Z_{\alpha,\beta}(E,0;p_u)\] 
    The leading scale dependence of the Boltzmann entropy is therefore
	\begin{equation}
		S_B(\bar R,p_u)
		\approx k_B\,a_{\alpha,\beta}\log\!\left(\frac{\bar R}{R_*}\right)
		+\text{\(\bar R\)-independent terms},
		\label{appeq:entropy_general_large_scale}
	\end{equation}
	where \(R_*\) is an arbitrary reference scale. For the distinguished model \((\alpha,\beta)=(-2,1)\),
	\begin{equation}
		\nu=\frac32,\qquad
		a_{-2,1}=\frac{9N-20}{2},\qquad
		\boldsymbol\ell=\frac{\mathbf J_0}{\bar R^{3/2}},\qquad
		P_\pi=R P_u .
		\label{appeq:entropy_distinguished_parameters}
	\end{equation}
	The bar only distinguishes the fixed macrovalue from the integration
	variable \(R\).  After the volume integral has been performed, it may be
	dropped, giving the notation used in the main text.

	\bibliography{biblio}{}

@article{RemmenCarroll2013,
  author = {Remmen, Grant N. and Carroll, Sean M.},
  title = {{Attractor Solutions in Scalar-Field Cosmology}},
  journal = {Physical Review D},
  volume = {88},
  number = {8},
  pages = {083518},
  year = {2013},
  doi = {10.1103/PhysRevD.88.083518},
  eprint = {1309.2611},
  archivePrefix = {arXiv},
  primaryClass = {gr-qc}
}

@article{Wald1983,
  author = {Wald, Robert M.},
  title = {{Asymptotic Behavior of Homogeneous Cosmological Models in the Presence of a Positive Cosmological Constant}},
  journal = {Physical Review D},
  volume = {28},
  number = {8},
  pages = {2118--2120},
  year = {1983},
  doi = {10.1103/PhysRevD.28.2118}
}

@article{MukhanovFeldmanBrandenberger1992,
  author = {Mukhanov, Viatcheslav F. and Feldman, H. A. and Brandenberger, Robert H.},
  title = {{Theory of Cosmological Perturbations}},
  journal = {Physics Reports},
  volume = {215},
  number = {5--6},
  pages = {203--333},
  year = {1992},
  doi = {10.1016/0370-1573(92)90044-Z}
}

@article{AllahverdiEtAl2010,
  author = {Allahverdi, Rouzbeh and Brandenberger, Robert and Cyr-Racine, Francis-Yan and Mazumdar, Anupam},
  title = {{Reheating in Inflationary Cosmology: Theory and Applications}},
  journal = {Annual Review of Nuclear and Particle Science},
  volume = {60},
  pages = {27--51},
  year = {2010},
  doi = {10.1146/annurev.nucl.012809.104511},
  eprint = {1001.2600},
  archivePrefix = {arXiv},
  primaryClass = {hep-th}
}

@article{KonopkaMarkopoulouSeverini2008,
  author = {Konopka, Tomasz and Markopoulou, Fotini and Severini, Simone},
  title = {{Quantum Graphity: A Model of Emergent Locality}},
  journal = {Physical Review D},
  volume = {77},
  number = {10},
  pages = {104029},
  year = {2008},
  doi = {10.1103/PhysRevD.77.104029},
  eprint = {0801.0861},
  archivePrefix = {arXiv},
  primaryClass = {hep-th}
}

@article{Buchert2000,
  author = {Buchert, Thomas},
  title = {{On Average Properties of Inhomogeneous Fluids in General Relativity I: Dust Cosmologies}},
  journal = {General Relativity and Gravitation},
  volume = {32},
  pages = {105--125},
  year = {2000},
  doi = {10.1023/A:1001800617177},
  eprint = {gr-qc/9906015},
  archivePrefix = {arXiv}
}

@article{MarsdenWeinstein1974,
  author = {Marsden, Jerrold and Weinstein, Alan},
  title = {{Reduction of Symplectic Manifolds with Symmetry}},
  journal = {Reports on Mathematical Physics},
  volume = {5},
  number = {1},
  pages = {121--130},
  year = {1974},
  doi = {10.1016/0034-4877(74)90021-4}
}

@article{BabbarSadkiPrakashSondhi_Fractons_2025,
  author = {Babbar, Aryaman and Sadki, Ylias and Prakash, Abhishodh and Sondhi, S. L.},
  title = {{Classical fractons: Local chaos, global broken ergodicity, and an arrow of time}},
  journal = {Physical Review B},
  volume = {111},
  number = {24},
  pages = {245134},
  year = {2025},
  doi = {10.1103/g2l1-s2vy},
  eprint = {2501.12445},
  archivePrefix = {arXiv},
  primaryClass = {cond-mat.stat-mech}
}

@article{SadkiPrakashSondhiArovas_PhaseSpaceFractons_2026,
  author = {Sadki, Ylias and Prakash, Abhishodh and Sondhi, S. L. and Arovas, Daniel P.},
  title = {{Phase space fractons}},
  journal = {Physical Review Letters},
  volume = {136},
  number = {12},
  pages = {126504},
  year = {2026},
  doi = {10.1103/b974-mpkc},
  eprint = {2502.02650},
  archivePrefix = {arXiv},
  primaryClass = {cond-mat.stat-mech}
}

@article{SadkiPrakashSondhi_ContinuumFractons_2025,
  author = {Sadki, Ylias and Prakash, Abhishodh and Sondhi, S. L.},
  title = {{Continuum fractons: Quantization and the few-body problem}},
  journal = {Physical Review B},
  volume = {114},
  number = {4},
  pages = {045102},
  year = {2026},
  doi = {10.1103/2pt1-d469},
  eprint = {2510.00110},
  archivePrefix = {arXiv},
  primaryClass = {cond-mat.str-el}
}

@article{ClassenHowesSenesePrakash_FreezingTransitions_2025,
  author = {Classen-Howes, Jonathan and Senese, Riccardo and Prakash, Abhishodh},
  title = {{Universal Freezing Transitions of Dipole-Conserving Chains}},
  journal = {Physical Review B},
  volume = {112},
  pages = {125148},
  year = {2025},
  doi = {10.1103/2h1v-yx5l},
  eprint = {2408.10321},
  archivePrefix = {arXiv},
  primaryClass = {cond-mat.str-el}
}

@article{PrakashGorielySondhi_Fractons_2024,
  author = {Prakash, Abhishodh and Goriely, Alain and Sondhi, S. L.},
  title = {{Classical nonrelativistic fractons}},
  journal = {Physical Review B},
  volume = {109},
  number = {5},
  pages = {054313},
  year = {2024},
  doi = {10.1103/PhysRevB.109.054313}
}

@article{PrakashSadkiSondhi_Fractons_2024,
  author = {Prakash, Abhishodh and Sadki, Ylias and Sondhi, S. L.},
  title = {{Machian fractons, Hamiltonian attractors, and nonequilibrium steady states}},
  journal = {Physical Review B},
  volume = {110},
  number = {2},
  pages = {024305},
  year = {2024},
  doi = {10.1103/PhysRevB.110.024305}
}

@article{10.1093/qmath/os-5.1.64,
  author = {Milne, E. A.},
  title = {{A Newtonian Expanding Universe}},
  journal = {The Quarterly Journal of Mathematics},
  volume = {os-5},
  number = {1},
  pages = {64--72},
  year = {1934},
  doi = {10.1093/qmath/os-5.1.64}
}

@article{Bagla_1997,
  author = {Bagla, J. S. and Padmanabhan, T.},
  title = {{Cosmological N-body simulations}},
  journal = {Pramana},
  volume = {49},
  number = {2},
  pages = {161--192},
  year = {1997},
  doi = {10.1007/BF02845853}
}

@article{Ellis_2013,
  author = {Ellis, George F. R. and Gibbons, Gary W.},
  title = {{Discrete Newtonian cosmology}},
  journal = {Classical and Quantum Gravity},
  volume = {31},
  number = {2},
  pages = {025003},
  year = {2014},
  doi = {10.1088/0264-9381/31/2/025003},
  eprint = {1308.1852},
  archivePrefix = {arXiv},
  primaryClass = {astro-ph.CO}
}

@book{landau1976mechanics,
  author = {Landau, L. D. and Lifshitz, E. M.},
  title = {{Mechanics}},
  series = {{Course of Theoretical Physics}},
  volume = {1},
  edition = {3},
  publisher = {Butterworth-Heinemann},
  address = {Oxford},
  year = {1976},
  isbn = {9780750628969}
}

@article{Anderson_2006,
  author = {Anderson, Edward},
  title = {{Relational particle models: I. Reconciliation with standard classical and quantum theory}},
  journal = {Classical and Quantum Gravity},
  volume = {23},
  number = {7},
  pages = {2469--2490},
  year = {2006},
  doi = {10.1088/0264-9381/23/7/016}
}

@article{Barbour_2003,
  author = {Barbour, Julian},
  title = {{Scale-invariant gravity: particle dynamics}},
  journal = {Classical and Quantum Gravity},
  volume = {20},
  number = {8},
  pages = {1543--1570},
  year = {2003},
  doi = {10.1088/0264-9381/20/8/310},
  eprint = {gr-qc/0211021},
  archivePrefix = {arXiv}
}

@article{10.1098/rspa.1982.0102,
  author = {Barbour, J. B. and Bertotti, B.},
  title = {{Mach's principle and the structure of dynamical theories}},
  journal = {Proceedings of the Royal Society of London. A. Mathematical and Physical Sciences},
  volume = {382},
  number = {1783},
  pages = {295--306},
  year = {1982},
  doi = {10.1098/rspa.1982.0102}
}

@article{bravetti2022scalingsymmetriescontactreduction,
  author = {Bravetti, Alessandro and Jackman, Connor and Sloan, David},
  title = {{Scaling symmetries, contact reduction and Poincar{\'e}'s dream}},
  journal = {Journal of Physics A: Mathematical and Theoretical},
  volume = {56},
  number = {43},
  pages = {435203},
  year = {2023},
  doi = {10.1088/1751-8121/acfddd},
  eprint = {2206.09911},
  archivePrefix = {arXiv},
  primaryClass = {math-ph}
}

@article{Battye_2003,
  author = {Battye, Richard A. and Gibbons, Gary W. and Sutcliffe, Paul M.},
  title = {{Central configurations in three dimensions}},
  journal = {Proceedings of the Royal Society A: Mathematical, Physical and Engineering Sciences},
  volume = {459},
  number = {2032},
  pages = {911--943},
  year = {2003},
  doi = {10.1098/rspa.2002.1061},
  eprint = {hep-th/0201101},
  archivePrefix = {arXiv}
}

@article{Ellis_2015,
  author = {Ellis, George F. R. and Gibbons, Gary W.},
  title = {{Discrete Newtonian cosmology: perturbations}},
  journal = {Classical and Quantum Gravity},
  volume = {32},
  number = {5},
  pages = {055001},
  year = {2015},
  doi = {10.1088/0264-9381/32/5/055001},
  eprint = {1409.0395},
  archivePrefix = {arXiv},
  primaryClass = {gr-qc}
}

@article{Lourenco:2026uto,
  author = {Louren{\c{c}}o, Maria I. R. and Barbour, Julian and Lobo, Francisco S. N.},
  title = {{Scale Invariance, Variety and Central Configurations}},
  eprint = {2602.11225},
  archivePrefix = {arXiv},
  primaryClass = {physics.hist-ph},
  year = {2026}
}

@article{Lourenco:2026lbr,
  author = {Louren{\c{c}}o, Maria I. R. and Barbour, Julian and Lobo, Francisco S. N.},
  title = {{Emergence of measured geometry in self-gravitating systems}},
  journal = {Physical Review D},
  volume = {113},
  number = {10},
  pages = {104068},
  year = {2026},
  doi = {10.1103/dgnr-kgdw},
  eprint = {2602.18115},
  archivePrefix = {arXiv},
  primaryClass = {gr-qc}
}

@article{BarbourKoslowskiMercati_PhysRevLett.113.181101,
  author = {Barbour, Julian and Koslowski, Tim and Mercati, Flavio},
  title = {{Identification of a Gravitational Arrow of Time}},
  journal = {Physical Review Letters},
  volume = {113},
  number = {18},
  pages = {181101},
  year = {2014},
  doi = {10.1103/PhysRevLett.113.181101},
  eprint = {1409.0917},
  archivePrefix = {arXiv},
  primaryClass = {gr-qc}
}

@article{gryb2025accountarrowtimescale,
      title={{An account of the arrow of time when scale is surplus}}, 
      author={Sean Gryb and Simon Friederich},
      year={2025},
      eprint={2510.03041},
      archivePrefix={arXiv},
      primaryClass={physics.hist-ph},
      url={https://arxiv.org/abs/2510.03041}, 
}

@article{GoldsteinTumulkaZanghi_arrow_PhysRevD.94.023520,
  title = {{Is the hypothesis about a low entropy initial state of the Universe necessary for explaining the arrow of time?}},
  author = {Goldstein, Sheldon and Tumulka, Roderich and Zangh\`{\i}, Nino},
  journal = {Phys. Rev. D},
  volume = {94},
  issue = {2},
  pages = {023520},
  numpages = {7},
  year = {2016},
  month = {Jul},
  publisher = {American Physical Society},
  doi = {10.1103/PhysRevD.94.023520},
  url = {https://link.aps.org/doi/10.1103/PhysRevD.94.023520}
}

@article{Chamon2005,
  author = {Chamon, Claudio},
  title = {{Quantum Glassiness in Strongly Correlated Clean Systems: An Example of Topological Overprotection}},
  journal = {Physical Review Letters},
  volume = {94},
  number = {4},
  pages = {040402},
  year = {2005},
  doi = {10.1103/PhysRevLett.94.040402}
}

@article{Haah2011,
  author = {Haah, Jeongwan},
  title = {{Local stabilizer codes in three dimensions without string logical operators}},
  journal = {Physical Review A},
  volume = {83},
  number = {4},
  pages = {042330},
  year = {2011},
  doi = {10.1103/PhysRevA.83.042330},
  eprint = {1101.1962},
  archivePrefix = {arXiv},
  primaryClass = {quant-ph}
}

@article{Vijay2015,
  author = {Vijay, Sagar and Haah, Jeongwan and Fu, Liang},
  title = {{A new kind of topological quantum order: A dimensional hierarchy of quasiparticles built from stationary excitations}},
  journal = {Physical Review B},
  volume = {92},
  number = {23},
  pages = {235136},
  year = {2015},
  doi = {10.1103/PhysRevB.92.235136}
}

@article{Pretko2017a,
  author = {Pretko, Michael},
  title = {{Subdimensional particle structure of higher rank U(1) spin liquids}},
  journal = {Physical Review B},
  volume = {95},
  number = {11},
  pages = {115139},
  year = {2017},
  doi = {10.1103/PhysRevB.95.115139},
  eprint = {1604.05329},
  archivePrefix = {arXiv},
  primaryClass = {cond-mat.str-el}
}

@article{Pretko2017b,
  author = {Pretko, Michael},
  title = {{Generalized electromagnetism of subdimensional particles: A spin liquid story}},
  journal = {Physical Review B},
  volume = {96},
  number = {3},
  pages = {035119},
  year = {2017},
  doi = {10.1103/PhysRevB.96.035119},
  eprint = {1606.08857},
  archivePrefix = {arXiv},
  primaryClass = {cond-mat.str-el}
}

@article{PretkoRadzihovsky2018,
  author = {Pretko, Michael and Radzihovsky, Leo},
  title = {{Fracton-Elasticity Duality}},
  journal = {Physical Review Letters},
  volume = {120},
  number = {19},
  pages = {195301},
  year = {2018},
  doi = {10.1103/PhysRevLett.120.195301},
  eprint = {1711.11044},
  archivePrefix = {arXiv},
  primaryClass = {cond-mat.str-el}
}

@article{Gromov2019,
  author = {Gromov, Andrey},
  title = {{Chiral Topological Elasticity and Fracton Order}},
  journal = {Physical Review Letters},
  volume = {122},
  number = {7},
  pages = {076403},
  year = {2019},
  doi = {10.1103/PhysRevLett.122.076403},
  eprint = {1712.06600},
  archivePrefix = {arXiv},
  primaryClass = {cond-mat.str-el}
}

@article{Pretko2017c,
  author = {Pretko, Michael},
  title = {{Emergent gravity of fractons: Mach's principle revisited}},
  journal = {Physical Review D},
  volume = {96},
  number = {2},
  pages = {024051},
  year = {2017},
  doi = {10.1103/PhysRevD.96.024051},
  eprint = {1702.07613},
  archivePrefix = {arXiv},
  primaryClass = {cond-mat.str-el}
}

@article{Yan2019,
  author = {Yan, Han},
  title = {{Hyperbolic fracton model, subsystem symmetry, and holography}},
  journal = {Physical Review B},
  volume = {99},
  number = {15},
  pages = {155126},
  year = {2019},
  doi = {10.1103/PhysRevB.99.155126},
  eprint = {1807.05942},
  archivePrefix = {arXiv},
  primaryClass = {hep-th}
}

@article{Gorantla2020,
  author = {Gorantla, Pranay and Lam, Ho Tat and Seiberg, Nathan and Shao, Shu-Heng},
  title = {{A modified Villain formulation of fractons and other exotic theories}},
  journal = {Journal of Mathematical Physics},
  volume = {62},
  number = {10},
  pages = {102301},
  year = {2021},
  doi = {10.1063/5.0060808},
  eprint = {2103.01257},
  archivePrefix = {arXiv},
  primaryClass = {cond-mat.str-el}
}

@article{SeibergShao2021,
  author = {Seiberg, Nathan and Shao, Shu-Heng},
  title = {{Exotic Symmetries, Duality, and Fractons in 2+1-Dimensional Quantum Field Theory}},
  journal = {SciPost Physics},
  volume = {10},
  number = {2},
  pages = {027},
  year = {2021},
  doi = {10.21468/SciPostPhys.10.2.027},
  eprint = {2003.10466},
  archivePrefix = {arXiv},
  primaryClass = {cond-mat.str-el}
}

@article{SlagleKim2017,
  author = {Slagle, Kevin and Kim, Yong Baek},
  title = {{Fracton topological order from nearest-neighbor two-spin interactions and dualities}},
  journal = {Physical Review B},
  volume = {96},
  number = {16},
  pages = {165106},
  year = {2017},
  doi = {10.1103/PhysRevB.96.165106},
  eprint = {1704.03870},
  archivePrefix = {arXiv},
  primaryClass = {cond-mat.str-el}
}

@article{Gromov2019curved,
  author = {Gromov, Andrey},
  title = {{Towards Classification of Fracton Phases: The Multipole Algebra}},
  journal = {Physical Review X},
  volume = {9},
  number = {3},
  pages = {031035},
  year = {2019},
  doi = {10.1103/PhysRevX.9.031035},
  eprint = {1812.05104},
  archivePrefix = {arXiv},
  primaryClass = {cond-mat.str-el}
}

@article{LMBFGSB_1995SJSC...16.1190B,
  author = {Byrd, Richard H. and Lu, Peihuang and Nocedal, Jorge and Zhu, Ciyou},
  title = {{A Limited Memory Algorithm for Bound Constrained Optimization}},
  journal = {SIAM Journal on Scientific Computing},
  volume = {16},
  number = {5},
  pages = {1190--1208},
  year = {1995},
  doi = {10.1137/0916069}
}

@article{DiacuPerezChavelaSantoprete2006,
  author = {Diacu, Florin and P{\'e}rez-Chavela, Ernesto and Santoprete, Manuele},
  title = {{Central configurations and total collisions for quasihomogeneous n-body problems}},
  journal = {Nonlinear Analysis: Theory, Methods \& Applications},
  volume = {65},
  number = {7},
  pages = {1425--1439},
  year = {2006},
  doi = {10.1016/j.na.2005.10.023}
}

@article{CarazzatoPratelliTopaloglu2026,
  author  = {Carazzato, Davide and Pratelli, Aldo and Topaloglu, Ihsan},
  title   = {{Particle Approximation of Nonlocal Interaction Energies}},
  journal = {Nonlinear Analysis},
  volume  = {263},
  pages   = {113974},
  year    = {2026},
  doi     = {10.1016/j.na.2025.113974},
  eprint  = {2506.02905},
  archivePrefix = {arXiv},
  primaryClass = {math.AP}
}

@article{FrankMatzke2025,
  author  = {Frank, Rupert L. and Matzke, Ryan W.},
  title   = {{Minimizers for an Aggregation Model with
             Attractive--Repulsive Interaction}},
  journal = {Archive for Rational Mechanics and Analysis},
  volume  = {249},
  pages   = {15},
  year    = {2025},
  doi     = {10.1007/s00205-025-02084-1}
}

@article{Maderna2013HomogeneousNBody,
  author = {Maderna, Ezequiel},
  title = {{Minimizing configurations and Hamilton--Jacobi equations of homogeneous N-body problems}},
  journal = {Regular and Chaotic Dynamics},
  volume = {18},
  number = {6},
  pages = {656--673},
  year = {2013},
  doi = {10.1134/S1560354713060063},
  eprint = {1308.5578},
  archivePrefix = {arXiv},
  primaryClass = {math.DS}
}

@article{CarrollChen2004spontaneousinflationoriginarrow,
      title={{Spontaneous Inflation and the Origin of the Arrow of Time}}, 
      author={Sean M. Carroll and Jennifer Chen},
      year={2004},
      eprint={hep-th/0410270},
      archivePrefix={arXiv},
      primaryClass={hep-th},
      url={https://arxiv.org/abs/hep-th/0410270}, 
}

@article{CarrollChen_2005,
   title={{Does Inflation Provide Natural Initial Conditions for the Universe?}},
   volume={14},
   ISSN={1793-6594},
   url={http://dx.doi.org/10.1142/S0218271805008054},
   DOI={10.1142/s0218271805008054},
   number={12},
   journal={International Journal of Modern Physics D},
   publisher={World Scientific Pub Co Pte Lt},
   author={Carroll, Sean M. and Chen, Jennifer},
   year={2005},
   month=Dec, pages={2335–2339} }

@article{JainJensen2022,
  author = {Jain, Akash and Jensen, Kristan},
  title = {{Fractons in curved space}},
  journal = {SciPost Physics},
  volume = {12},
  number = {4},
  pages = {142},
  year = {2022},
  doi = {10.21468/SciPostPhys.12.4.142},
  eprint = {2111.03973},
  archivePrefix = {arXiv},
  primaryClass = {hep-th}
}

@article{Bidussi2023,
  author = {Bidussi, Leo and Hartong, Jelle and Have, Emil and Musaeus, J{\o}rgen and Prohazka, Stefan},
  title = {{Fractons, dipole symmetries and curved spacetime}},
  journal = {SciPost Physics},
  volume = {12},
  number = {6},
  pages = {205},
  year = {2022},
  doi = {10.21468/SciPostPhys.12.6.205},
  eprint = {2111.03668},
  archivePrefix = {arXiv},
  primaryClass = {hep-th}
}

@article{GomesGrybKoslowski2011,
  author  = {H. Gomes and S. Gryb and T. Koslowski},
  title   = {{Einstein gravity as a 3D conformally invariant theory}},
  journal = {Classical and Quantum Gravity},
  volume  = {28},
  number  = {4},
  pages   = {045005},
  year    = {2011},
  doi     = {10.1088/0264-9381/28/4/045005}
}

@incollection{Barbour2012,
  title={{Shape dynamics. An introduction}},
  author={Barbour, Julian},
  booktitle={Quantum field theory and gravity: Conceptual and mathematical advances in the search for a unified framework},
  editor={Finster, Felix and M{\"u}ller, Olaf and Nardmann, Marc and Tolksdorf, J{\"u}rgen and Zeidler, Eberhard},
  pages={257--297},
  year={2012},
  publisher={Birkh{\"a}user Basel},
  doi={10.1007/978-3-0348-0043-3_13}
}

@book{Strogatz2015,
  author = {Strogatz, Steven H.},
  title = {{Nonlinear Dynamics and Chaos: With Applications to Physics, Biology, Chemistry, and Engineering}},
  edition = {2},
  publisher = {Westview Press},
  address = {Boulder, Colorado},
  year = {2015},
  isbn = {9780813349107}
}

@article{WilsonKogut1974,
  author = {Wilson, Kenneth G. and Kogut, John},
  title = {{The Renormalization Group and the Epsilon Expansion}},
  journal = {Physics Reports},
  volume = {12},
  number = {2},
  pages = {75--199},
  year = {1974},
  doi = {10.1016/0370-1573(74)90023-4}
}

@article{Bertschinger,
author = {Bertschinger, Edmund},
year = {1998},
month = {September},
pages = {599--654},
title = {{Simulations of structure formation in the universe}},
volume = {36},
journal = {Annual Review of Astronomy and Astrophysics},
doi = {10.1146/annurev.astro.36.1.599}
}

@book{Arnold1978,
  author    = {V. I. Arnold},
  title     = {{Mathematical Methods of Classical Mechanics}},
  publisher = {Springer},
  year      = {1978},
  note      = {Section 11, ``The Method of Similarity,'' in Chapter 2, for mechanical similarity and scaling symmetries}
}

@incollection{Penrose1979,
  author    = {Roger Penrose},
  title     = {{Singularities and Time-Asymmetry}},
  booktitle = {General Relativity: An Einstein Centenary Survey},
  editor    = {S.~W.~Hawking and W.~Israel},
  pages     = {581--638},
  publisher = {Cambridge University Press},
  year      = {1979}
}

@article{Linde1983,
  author  = {A.~D.~Linde},
  title   = {{Chaotic inflation}},
  journal = {Physics Letters B},
  volume  = {129},
  pages   = {177--181},
  year    = {1983}
}

@book{dodelson2020modern,
  title={{Modern Cosmology}},
  author={Dodelson, Scott and Schmidt, Fabian},
  edition={2nd},
  year={2020},
  publisher={Academic Press},
  isbn={978-0128159484}
}

@article{Planck:2018vyg,
  author       = {Planck Collaboration and Aghanim, N. and others},
  title        = {{Planck 2018 results. VI. Cosmological parameters}},
  journal      = {Astronomy \& Astrophysics},
  volume       = {641},
  pages        = {A6},
  year         = {2020},
  doi          = {10.1051/0004-6361/201833910},
  eprint       = {1807.06209},
  archivePrefix= {arXiv},
  primaryClass = {astro-ph.CO}
}

@book{Peebles1993,
  author    = {P. J. E. Peebles},
  title     = {{Principles of Physical Cosmology}},
  publisher = {Princeton University Press},
  year      = {1993}
}

@book{albert2000time,
  title={{Time and Chance}},
  author={Albert, David Z.},
  year={2000},
  publisher={Harvard University Press},
  address={Cambridge, MA}
}

@article{guth1981inflationary,
  title={{Inflationary universe: A possible solution to the horizon and flatness problems}},
  author={Guth, Alan H.},
  journal={Physical Review D},
  volume={23},
  number={2},
  pages={347--356},
  year={1981},
  publisher={American Physical Society},
  doi={10.1103/PhysRevD.23.347}
}

@article{einstein1945influence,
  title={{The influence of the expansion of space on the gravitation fields surrounding the individual stars}},
  author={Einstein, Albert and Straus, Ernst G.},
  journal={Reviews of Modern Physics},
  volume={17},
  number={2-3},
  pages={120--124},
  year={1945},
  publisher={American Physical Society},
  doi={10.1103/RevModPhys.17.120}
}

@article{carrera2010influence,
  title={{Influence of global cosmological expansion on local dynamics and kinematics}},
  author={Carrera, Matteo and Giulini, Domenico},
  journal={Reviews of Modern Physics},
  volume={82},
  number={1},
  pages={169--208},
  year={2010},
  publisher={American Physical Society},
  doi={10.1103/RevModPhys.82.169}
}

@article{brans1961mach,
  title={{Mach's Principle and a Relativistic Theory of Gravitation}},
  author={Brans, Carl and Dicke, Robert H.},
  journal={Physical Review},
  volume={124},
  number={3},
  pages={925--935},
  year={1961},
  publisher={APS},
  doi={10.1103/PhysRev.124.925}
}

@book{barbour1995mach,
  title={{Mach's Principle: From Newton's Bucket to Quantum Gravity}},
  editor={Barbour, Julian B. and Pfister, Herbert},
  series={Einstein Studies, Vol. 6},
  year={1995},
  publisher={Birkh{\"a}user},
  address={Boston},
  isbn={978-0817638238}
}

@article{sciama1953origin,
  title={{On the origin of inertia}},
  author={Sciama, Dennis W.},
  journal={Monthly Notices of the Royal Astronomical Society},
  volume={113},
  number={1},
  pages={34--42},
  year={1953},
  doi={10.1093/mnras/113.1.34}
}

@article{Sloan2018,
  author  = {Sloan, David},
  title   = {{Dynamical similarity}},
  journal = {Physical Review D},
  volume  = {97},
  number  = {12},
  pages   = {123541},
  year    = {2018},
  doi     = {10.1103/PhysRevD.97.123541}
}
	\bibliographystyle{jhep}
\end{document}